\documentclass{article}
\usepackage{amsmath,amsfonts,amssymb,amsthm}
\usepackage{mathtools}
\usepackage{hyperref}
\usepackage{float}
\usepackage{multicol}
\usepackage{textcomp}
\usepackage{xcolor}
\usepackage{balance}
\usepackage{xspace}
\RequirePackage{xparse}
\usepackage[numbers]{natbib}

\newtheorem{theorem}{Theorem}
\newtheorem{lemma}{Lemma}
\newtheorem{proposition}{Proposition}
\newtheorem{remark}{Remark}

\theoremstyle{definition}
\newtheorem{definition}{Definition}

\usepackage{thm-restate}

\usepackage{paper}
\usepackage{cleveref}
\usepackage[inline]{enumitem}
\usepackage{sbnf}
\renewcommand{\bnfcmt}[1]{\emph{#1}}
\renewcommand{\bnfmid}{\mathrel{\mbox{\Large{$\mid$}}}}
\usepackage[margin=1.2in]{geometry}
\usepackage{tikz}
\usepackage{authblk}

\usetikzlibrary{decorations.pathmorphing}
\usetikzlibrary{positioning}
\usetikzlibrary{calc}
\usetikzlibrary{fit}

\newcommand{\arxiv}{}

\newcommand{\DD}[1]{\ifdefined\commentaire{{{\color{orange!60!black!55} \texttt{[#1]}}}}\fi}
\newcommand{\TR}[1]{\ifdefined\commentaire{{{\color{blue!60!black!55} \texttt{[#1]}}}}\fi}

\author{Davide Davoli}
\author{Martin Avanzini}
\author{Tamara Rezk}
\affil{Inria, Université Côte d'Azur}

\title{On Kernel's Safety in the Spectre Era \\(And KASLR is Formally
  Dead)\footnote{This is the extended version of the conference
    paper~\cite{ConferencePaper}. In comparison to the conference
    version, the following changes have been made: (i) the notion of
    system has been made more liberal, by permitting system call
    invocation during kernel mode execution, (ii) the semantics and
    proofs have been extended to reflect this enriched model; (iii)
    the constraints on system call's capabilities and
    the definition of $\delta_\mu$ have been modified. }}

\begin{document}

\maketitle
\begin{abstract}
  The efficacy of address space layout randomization has been formally demonstrated in a shared-memory model by Abadi et al., contingent on specific assumptions about victim programs. However, modern operating systems, implementing layout randomization in the kernel, diverge from these assumptions and operate on a separate memory model with communication through system calls. 
In this work, we relax Abadi et al.'s language assumptions while demonstrating that layout randomization offers a comparable safety guarantee in a system with memory separation. 
However, in practice, speculative execution and side-channels are recognized threats to layout randomization. 
We show that kernel safety cannot be restored for attackers capable of using side-channels and speculative execution and introduce a new condition,  that allows us to formally prove kernel safety in the Spectre era. Our research demonstrates that under this condition, the system remains safe without relying on layout randomization. 
We also demonstrate that our condition can be sensibly weakened, leading to enforcement mechanisms  that can guarantee
kernel safety for safe system calls in the Spectre era.

\DD{todo:
  \begin{itemize}
  \item compare with exorcising spectres with secure compilers
  \item Say that ordinary mitigations to spectre vunerabilities do not stop probing
  \end{itemize}
}
\end{abstract}

\section{Introduction}
\label{sec:introduction}
Memory safety violations on kernel memory can result in serious ramifications for security, such as e.g. arbitrary code execution, privilege escalation, or information leakage. 
In order to mitigate safety violations, operating systems --- such as Linux --- employ 
address space layout randomization~\cite{LWNKASLR,MacOSXASLR,AndroidASLR,iOSASLR,OpenBSD6.3,FreeBSDASLR}. This protection measure can prevent attacks that depend on knowledge of specific data or function location, as it introduces randomization of  these addresses.

On the one hand, the efficacy of layout randomization has been formally demonstrated in Abadi et al.'s line of work~\cite{Abadi, Abadi2, Abadi3}, as a protective measure within a {\it shared-memory model} between the attacker and the victim. These results, however, are contingent on specific assumptions regarding victim programs, notably the absence of pointer arithmetic,  introspection, or indirect jumps. These precise constraints shaped a controlled environment where memory safety could be enforced effectively via layout randomization.
However, operating systems employing layout randomization on kernel (a.k.a. KASLR in Linux e.g.~\cite{LWNKASLR}) diverge from these assumptions. Notably, they operate on a  separate memory model, wherein, kernel code --- acting as the victim --- resides on
kernel memory, while user code --- acting as the potential attacker --- resides in user space. The interaction between the two occurs through a limited set  of functions provided via interfaces or system calls~\cite{Tanenbaum}.
In the operating system's realm, system calls may be written in C and assembly code, further deviating from 
  the restricted conditions outlined by Abadi et al.  This introduces a distinction not only in the expressiveness of victim code considered but also in the underlying memory model.
  
Hence, our first research question  emerges: can we relax the language assumptions proposed by Abadi et al.~\cite{Abadi,Abadi2,Abadi3} while concurrently demonstrating that layout randomization offers a comparable safety guarantee in a system  with memory 
separation? We affirmatively respond to this question by showcasing that layout randomization probabilistically ensures kernel safety within a classic attacker model, where users of an operating system execute without privileges and victims can feature pointer arithmetic, introspection, and indirect jumps.

On the other hand, in the current state-of-the-art of security, often referred to as the Spectre era, speculative execution and side-channels are well known to be effective vectors for compromising layout randomization~\cite{BlindSide,Meltdown,cacheKASLR,EntryBleed,TagBleed,kaslrfbr}.\TR{add example Blindside here?? } 
Indeed, our first result neglects the impact of speculative execution and side-channels. Recognizing this limitation, our second research question arises:  can we restore a similar safety result in the Spectre era?

In this regard, we formally acknowledge that by relying solely on layout randomization it is not possible to restore kernel safety. We then introduce a new condition, called {\it speculative layout non-interference} akin to speculative constant-time~\cite{CTFundations}, which intuitively asserts that victims should not unintentionally leak information on the kernel's layout through side-channels. Our research formally demonstrates that under this assumption, the system is safe, and perhaps surprisingly, without the necessity of layout randomization.
Later, we show that speculative layout non-interference is not a necessary requirement, and this motivates us to study how safety can be enforced without requiring that property.

Our third contribution is to show that kernels can be protected even without requiring {\it speculative layout non-interference}. Following other similar works~\cite{Blade,ProSpeCT}, we do so by relating safety in the classic execution model to the speculative one.
%
%
We achieve this result by defining a program transformation by which we enforce safety against speculative attackers on a system that does not conform to this property.
%
This transformation, in turn, requires the system to enjoy a notion of safety that cannot be provided by layout randomization, but that is sensibly weaker than the safety property it enforces. \DD{last sentence needs to be rephrased}
This marks the first formal step toward strengthening kernel safety in the presence of speculative and side-channel vulnerabilities,
and the surpassing of layout randomization as a system level protection mechanism.


In summary, our contributions are: 
\begin{itemize}
\item  We formally demonstrate the effectiveness of layout randomization to provide kernel safety for a classic operating system scenario, with system calls offered as 
interfaces to attackers and different privilege execution modes, as well as kernel and user memory separation. 
\item We empower attackers in our first scenario to execute side-channel attacks and utilize speculative execution. Demonstrating that kernel safety is not maintained under this more potent attacker model, we subsequently present a sufficient condition to ensure kernel safety.
\item We show that it is possible to enforce safety against speculative attackers on a system that enjoys weaker security guarantees by the application of a program transformation.
\end{itemize} 

The paper is structured as follows: in \Cref{sec:motivations} we give an overview
of the contributions of this paper, motivated by some concrete examples.
In \Cref{sec:language}, we introduce our execution model by giving its
language and semantics; in \Cref{sec:threatmodel}, we establish
threat models.
\Cref{sec:safety1} is devoted to showing that layout randomization is an effective
protection measure for attacks that do not rely on speculative execution
and side-channel observations.
In \Cref{sec:safety2} we first extend the model of \Cref{sec:language}
to encompass time-channel info leaks
and speculative execution, then we show that layout randomization
is not a viable protection mechanism in this scenario.
In \Cref{sec:sksenforcement} we show that it is feasible
to convert any system that is safe against classic attackers
into an equivalent system that is safe against speculative attackers.
Finally, we consider related
work in \Cref{sec:relatedwork}, and we conclude in
\Cref{sec:conclusion}.
\ifdefined\conference{An extended version of this paper with the complete semantics and the omitted proofs is available online,~\cite{ExtendedVersion}.}\fi


\section{Motivation}
\label{sec:motivations}
\begin{figure}[t]
  \centering
  \begin{framed}
    {\small
\begin{code}[emph={buf,recv,send,valid},morekeywords={socket,size_t}]
int buf[K+1][H];

int recv(socket* s, size_t idx) {
  if (valid(s, idx)) return buf[*s][idx];
  return 0;
}

void send(socket* s, size_t idx, int msg) {
  if (valid(s, idx)) {
    buf[*s][idx] = msg;
    if (buf[K][0] != NULL) (*buf[K][0])(s, idx);
  }
}
      \end{code}
    }%
  \end{framed}
  \caption{System Calls vulnerable to memory corruption}
  \label{fig:msgpassing}
\end{figure}

Each year, dozens of vulnerabilities are found in commodity operating systems' kernels, and the majority of them are memory corruption vulnerabilities~\cite{ThreatOverview}. A kernel suffers a memory corruption vulnerability when an unprivileged attacker can trigger it to read or write its memory in an \emph{unexpected} way, usually, by issuing a sequence of system calls with maliciously crafted arguments. In \Cref{fig:msgpassing}, we show a pair of system calls of a hypothetical system that are subject to this kind of vulnerability.
\def\ex{\cc[emph={buf, recv, send, valid, buf, native_write_cr4, sc_leak, foo}]}
The \ex{recv} and \ex{send} system calls are meant to implement a simple message passing protocol.
The implementation supports up to \ex{K} sockets, each socket storing up to \ex{H} messages.
A user can send messages by invoking the system call \ex{send}, and read them with the system call \ex{recv}.
These system calls employ a shared buffer \ex!buf! that stores messages, together with a hook for a customizable callback \ex{buf[K][0]}.
If specified, this callback is executed after a message is sent. Such a callback may, for instance, signal the receiver that a new message is available.

These system calls are meant to interact only with the memory containing the buffer, the code of the called functions and with the resources that these function in turn access.
In the following, we will refer to the set of memory resources that a system call may access rightfully as the \emph{\tcaps} of that system call.
Depending on the implementation of the \ex{valid} function, these system calls can suffer from memory corruption vulnerabilities.
For instance, if the \ex{valid} function does not perform any bound checks on the value of \ex{idx},
these two system calls can be used by the attacker to perform arbitrary read and write operations.
In particular, if the attacker supplies an out-of-bounds value for \ex{idx} to the \ex{recv} system call,
the system call can be used to perform an unrestricted memory read.
Similarly, the \ex{send} system call can be used to overwrite any value of kernel memory and,
in particular, to overwrite the function pointer to the callback that is stored within the buffer.
This means that the attacker can turn this memory-vulnerability into a control-flow vulnerability, as it can deviate the control flow from its intended paths.
When this happens, we talk about violations of control flow integrity (CFI)~\cite{CFI}.

However, if the system that implements these system calls is protected with layout randomization --- such as many commodity operating systems~\cite{LWNKASLR,MacOSXASLR,AndroidASLR,iOSASLR,OpenBSD6.3,FreeBSDASLR} --- the exploitation of these vulnerabilities is not a straightforward operation. For instance, if an attacker wants to mount a privilege-escalation attack, one of the viable ways is to disable the SMEP protection by running the \ex{native_write_cr4} function. When this protection is disabled, the attacker is allowed to run any \emph{payload} stored in user-space. To this aim, the attacker can use the \ex{send} system call twice: the first time to run the \ex{native_write_cr4} function instead of the callback, and the second time to run the \emph{payload}.  
However, in order to do so, the attacker has to first infer the address of \ex{native_write_cr4}.
In the absence of side channel info-leaks, an attacker has to effectively guess this address and, due to layout randomization, the probability of success is low.

This is what we show with our first result (\Cref{thm:scenario1}):
without side-channel leaks (and speculative execution), if a system is protected with layout randomization, the probability that an unprivileged attacker leads the system to perform an unsafe memory access is very low, provided the address space is sufficiently large. Of course, the precise probability depends on the concrete randomization scheme.
We emphasize that this result is compatible with the large number
of kernel attacks that break Linux's kernel layout randomization, e.g.\ by means of heap overflows~(\cite{HeapFengShui2022}).
The distribution of Linux's heap addresses lack entropy~\cite{FreelistRandomization}, in consequence,
the probability of mounting a successful attack are relatively high.

Although it was already well known that layout randomization can provide some security guarantees~\cite{Abadi, Abadi2, Abadi3}, the novelty of \Cref{thm:scenario1} lies in showing that these guarantees are valid even if victims can perform pointer arithmetic and indirect jumps. 

Despite this positive result, the threat model considered in \Cref{thm:scenario1} is unrealistic nowadays. In particular, it does not take in account the ability of the attackers to access side-channel info-leaks and to steer speculative execution. There is evidence that, by leveraging similar features, the attackers can leak information on the kernel's layout~\cite{cacheKASLR,TLBKASLR,TagBleed, EntryBleed, kaslrfbr} and compromise the security guarantees offered by layout randomization~\cite{BlindSide,Meltdown}.

In particular, if the system under consideration suffers from side-channel info-leaks that involve the layout, an attacker may break the protection offered by randomization.
As an illustrative example, suppose the system contains the following system call:
\ifdefined\conference{\small
\begin{code}[emph={sc_leak,native_write_cr4}]
int sc_leak(x){
  if ((void*) x  == (void*) native_write_cr4)
    for(int i = 0; i < K; i++);
  return 0;
}
\end{code}
\normalsize
}\fi
\ifdefined\arxiv{
\begin{code}[emph={sc_leak,native_write_cr4}]
int sc_leak(x){
  if ((void*) x  == (void*) native_write_cr4)
    for(int i = 0; i < K; i++);
  return 0;
}
\end{code}
}\fi
By measuring the execution time of the system call, an attacker may deduce information on the location of \ex{native_write_cr4}.
If a call \ex{sc_leak(a)} takes sufficiently long to execute, the attacker can deduce that the address \ex{a} corresponds to that of \ex{native_write_cr4}.
Once deduced, the attacker will be effectively able to disable SMEP protection via the vulnerable system call \ex{send}.

Similar attacks can be mounted by taking advantage of speculative execution:
in our example from \Cref{fig:msgpassing}, an attacker can make use of the \ex{read} primitive to probe for readable data without crashing the system. This can be done by supplying to the system call arguments \ex{s} and \ex{idx} such that \ex{valid(s, idx)} returns false --- ideally, causing an out of bound access when the return value is fetched from memory. If the attacker manages in mis-training the branch predictor, the access to \ex{buf[*s][idx]} is performed in transient execution. Depending on the allocation state of the address referenced by \ex{buf[*s][idx]}, two cases arise. If that address does not store any readable data, the memory violation is not raised to the architectural state, because it occurred during transient execution. Most importantly, if that address stores writable data, this operation loads a new line in the system's cache and, as soon as the system detects the mis-prediction, the execution backtracks to the latest safe state. Although this operation does not affect the architectural state, the insertion of a new line in the cache can be detected from user-space. Thus, the attacker can infer that the address referenced by \ex{buf[*s][idx]} contains readable data, and it can make use of the vulnerabilities of the \ex{send} and the \ex{recv} system calls to read or write the content of that memory address. 
This form of \emph{speculative probing} is very similar to what happens, for instance, in the BlindSide attack~\cite{BlindSide} that effectively defeats Linux's KASLR.

The reader may observe that these two attacks rely on the attacker's ability to reconstruct the kernel's memory layout by collecting side-channel info-leaks. For this reason, a natural question is whether these attacks can be prevented by imposing that no information of the layouts leaks to the architectural and the micro-architectural state during the execution of system calls. It turns out that this is the case, as we show in~\Cref{thm:scenario2}. In practice, this mitigation is of little help though, as it would effectively rule out all system calls that access memory at runtime.

However, we are able to show that any operating system can be pragmatically turned into another system that is architecturally equivalent to it, but that is not subject to vulnerabilities that are due to transient execution. With this approach, showing that a kernel is safe in the speculative execution model, reduces to showing that the kernel under consideration is safe in the classic execution model. Notably, this holds independently of the technique that is used to show safety in the classic model. Concretely, with this approach, the attack we showed above would be prevented by disallowing the transient execution of the unsafe load operation. In turn, this can be achieved by placing an instruction that stops transient execution before that operation. The efficacy of this technique is shown in~\Cref{thm:mitigation}. 


\section{The Language}
\label{sec:language}
In this section, we introduce the language that we employ throughout the following to study the effectiveness of kernel address space layout randomization. We are considering a simple imperative \texttt{while} language. The address space is explicit, and segregated into user and kernel space. We start by describing the syntax.
\subsection{Syntax}

\begin{figure}
  \centering
  \begin{framed}
    \begin{align*}
      \Expr \ni \expr ,\exprtwo &\bnfdef
                                  \val
                                  \bnfmid  \vx 
                                  \bnfmid  \ar 
                                  \bnfmid  \fn 
                                  \bnfmid  \op (\expr_1,\dots,\expr_n)\\
      \Instr \ni \stat,\stattwo &\bnfdef
                                  \cskip 
                                  \bnfmid
                                  \vx \ass \expr 
      \bnfmid \cmemread \vx \expr 
      \bnfmid \cmemass \expr \exprtwo 
      \bnfmid \ccall \exprtwo{\expr_1,\dots,\expr_n} 
      \bnfmid \csyscall \syscall {\expr_1,\dots,\expr_n} \\ 
                                &  \bnfmid \cif{\expr}{\cmd}{\cmdtwo} 
      \bnfmid \cwhile{\expr}{\cmd} \\ 
                                          \Cmd \ni \cmd, \cmdtwo &\bnfdef \cnil \bnfmid \stat\sep\cmd
    \end{align*}
  \end{framed}
  \caption{Syntax of the language. Here, $\val$ is a value, $\vx$ a register, $\ar$ an array identifier, $\fn$ a procedure identifier, and $\op$ is an operator.}
  \label{fig:syntax}
\end{figure}


The set $\Cmd$ of \emph{commands} is given in \Cref{fig:syntax}.
Memories may store \emph{procedures} and \emph{arrays}, i.e.,
sequences of \emph{values} $\val \in \Val$ organized as contiguous
regions.  The set of values is left abstract, but we assume that it
encompasses at least \emph{Boolean values}
$\cBool \defsym \{\ctrue,\cfalse\}$, \emph{(memory) addresses} $\Add$,
modeled as non-negative integers, and an \emph{undefined value}
$\cnull$.  Within expressions, $\vx \in \Reg$ ranges over
\emph{registers}, $\ar \in \Ar$ and $\fn \in \Fn$ over \emph{array}
and \emph{procedure identifiers}, and $\op \in \Op$ over
\emph{operators}. \emph{Identifiers} $\Id \defsym \Ar \uplus \Fn$ are
mapped to addresses at runtime, as governed by a layout randomization
scheme.  The \emph{size} (length) of an array $\ar$ is denoted by
$\size{\ar}$ and is fixed for simplicity, i.e., we do not model
dynamic allocation and deallocation.

A command $\cmd \in \Cmd$ is a sequence of instructions, evaluated in-order.
The instruction $\vx \ass \expr$ stores the result of evaluating $\expr$ within register $\vx \in \Reg$.
To keep the semantics brief, expressions neither read nor write to
memory. Specifically, addresses are dereferenced explicitly.  To this
end, the instruction $\cmemread \vx \expr$ performs a memory read from
the address given by $\expr$, and stores the corresponding value in
register $\vx$. Dually, the instruction $\cmemass \expr \exprtwo$
stores the value of $\exprtwo$ at the address given by $\expr$.  The
instruction $\ccall \exprtwo{\expr_1,\dots,\expr_n}$ invokes the
procedure residing at address $\exprtwo$ in memory, supplying
arguments $\expr_1,\dots,\expr_n$.  Likewise,
$\csyscall \syscall {\expr_1,\dots,\expr_n}$ invokes a system call
$\syscall$ with arguments $\expr_1,\dots,\expr_n$ from a finite set of
system calls $\Sys$.  The execution of a system call engages the
\emph{privileged execution mode} and thereby the accessible address
space changes.  To this end, the address space $\Add$ is
  partitioned into $\kappa_\um$ \emph{user-space} addresses
  $\Addu = \{0, \ldots, \kappa_\um-1\}$, visible in unprivileged mode,
  and $\kappa_{\km}$ \emph{kernel-space} addresses
  $\Addk = \{ \kappa_{\um}, \ldots, \kappa_{\um}+\kappa_{\km}-1\}$,
  visible in unprivileged mode.  The remaining constructs are
standard.

\paragraph*{Stores} 
Let $\Arr$ denote the set of arrays, i.e., finite sequences of values
$\vec{\val}$ of fixed length $|\vec{\val}|$.  A \emph{store} is a
(well-sorted) mapping $\rfs : \Id \to \Arr \cup \Cmd$, mapping array
identifiers $\ar$ to arrays $\rfs(\ar) \in\Arr$ of length $\size{\ar}$
and procedure identifiers $\fn$ to their implementation
$\rfs(\fn) \in \Cmd$.  Let $\Idu \uplus \Idk = \Id$ be a partitioning
of identifiers int \emph{user-space} and \emph{kernel-space}
identifiers, respectively. This distinction will signify the intended
location of the corresponding objects within the memory address space.
We write $\Fnk \subseteq \Idk$ and $\Fnu \subseteq \Idu$ for the
\emph{kernel-space} and \emph{user-space procedure idenitifiers};
similar for array identifiers we use $\Ark \subseteq \Idk$ and
$\Aru \subseteq \Idu$ to denote \emph{kernel-space} and
\emph{user-space array idenitifiers} respectively.  Given a store
$\rfs$, we always assume that the address space is sufficiently large
to hold $\rfs$; that is, for every $b \in \{\um,\km\}$,
$\kappa_{b} \geq \sum_{\id \in \Id[b]} \size{\id}$.  Here, by
convention, $\size{\fn} \defsym 1$. In the following, we often work
with pairs of stores that associate a set of identifiers to the same
values---e.g., containing identical procedures. We write
$\rfs \eqon{Id} \rfs'$ if $\rfs$ and $\rfs'$ coincide on
$Id \subseteq \Id$.

\paragraph*{\Tcaps} To model safety, each system call $\syscall$
is associated with a fixed set of identifiers defining the memory
regions it can access. We refer to this set as the \emph{\tcaps} of
$\syscall$.

\paragraph*{Systems} 
In our model, kernels are modeled as triples defining system calls,
the content of kernel space memory, and the capabilities associated to every system call.  
Let $\Sys$ denote a (finite) set of \emph{system call identifiers}.
A \emph{system} for $\Sys$ is a tuple $\system = (\rfs,\syss,\caps)$, consisting of:
\begin{itemize}
\item an \emph{initial store} $\rfs : \Id \to \Arr \cup \Cmd$, relating identifiers to their initial
  value;
  \item a \emph{system call map} $\syss : \Sys \to \Cmd$ associating system calls to
  their implementation; and
\item a \emph{capability map} $\caps : \Sys \to \parts{\Idk}$ associating system calls with their capabilities.
\end{itemize}
We require that the code $\rfs(\fn)$ associated to user space
identifiers $\fn \in \Fnu$ is \emph{unprivileged}, i.e.
$\ids(\rfs(\fn)) \subseteq \Idu$, where $\ids(\cmd) \subseteq \Id$ is
the set of identifiers literally occurring in $\cmd$.  As our notion
of safety will be defined in terms of system calls' capabilities, we
furthermore require that the set of capabilities $\caps(\syscall)$ of
a system call $\syscall$ includes all identifiers that $\syscall$
refers to, directly in its body, or indirectly through procedure or
system calls.
To define this latter requirement formally, let
$\refs_\system(\cmd) \subseteq \Id \cup \Sys$ denote the set of
identifiers and system calls that the command $\cmd$ refers to,
directly or indirectly.  Specifically, we call $\mathsf{syscalls}(\cmd)$ is
the set of system calls in $\cmd$, and we define $\refs_\system(\cmd)$
as the smallest set such that:
(a)~$\ids(\cmd) \cup \mathsf{syscalls}(\cmd) \subseteq
\refs_\system(\cmd)$, where (b)~if $\fn \in \refs_\system(\cmd)$ then
$\ids(\rfs(\fn)) \subseteq \refs_\system(\cmd)$, and likewise, (c)~if
$\syscall \in \refs_\system(\cmd)$ then
$\ids(\syss(\syscall)) \subseteq \refs_\system(\cmd)$.  The
requirement can now be stated as
$\refs_\system(\syss(\syscall)) \setminus \Sys\subseteq
\caps(\syscall)$.
%

\subsection{Semantics}

We now endow our language with an operational semantics. To define our
notion of safety, we directly define an \emph{instrumented semantics},
which signals capability violations.  Semantics cannot be defined
directly on stores, which directly relate identifiers to their
values. Instead, to model address-based memory accesses, we introduce
\emph{memories}.

\paragraph*{Memories}
A \emph{memory} is a function
$\mem: \Add \to \Val \cup \Cmd \cup \{\none\}$ associating addresses
with their content, or to the special symbol $\none \not\in \Val$ if
the location is not occupied.  We denote by $\Mem$ the set of all
memories.  Arrays will be represented by sequences of values in
  contiguous memory locations. We use $\update \mem \add \val$ to
denotes the memory that is pointwise identical to $\mem \in \Mem$,
except for the address $\add\in \Add$ that is mapped to
$\val\in \Val$.  Note that updates are restricted values, in
particular updating a memory location with a command is forbidden.
Thereby we model a W\^{}X memory protection policy, separating
writable from executable memory regions.

\paragraph{Layouts} A \emph{(memory) layout} is a function $\lay : \Id \to \Add$ that
describes where objects are placed in memory.
As we mentioned, an array $\ar$ is stored as continuous block at addresses
$\underline \lay(\ar) \defsym \{ \lay(\ar), \ldots, \lay(\ar)+\size\ar - 1 \}$
within memory. For procedure identifiers $\fn$, we set 
$\underline \lay(\fn) \defsym \{ \lay(\fn)\}$.
We overload this notation to sets of identifiers in the obvious way.
In particular, $\underline \lay(\Ar)$ and $\underline \lay(\Fn)$ refer
to the address-spaces of arrays and procedures, under the given layout.
We regard only layouts that associate identifiers with non-overlapping blocks
($\underline \lay(\id_1)\cap \underline \lay(\id_2) = \emptyset$ for all
$\id_1 \neq \id_2$) and that respect address space separation ($\underline \lay(\id) \subseteq \Add[b]$
for $\id \in \Id[b]$, $b \in \{\um,\km\}$).
The set of all such layouts is denoted by $\Lay$.
Note that, by the assumptions on the sizes
$\kappa_\um$ and $\kappa_{\km}$ of address spaces, layouts always exist.


A layout $\lay : \Id \to \Add$ now defines how a store
$\rfs : \Id \to \Arr \cup \Cmd$ is placed in memory. This memory,
denoted by $\lay \lcomp \rfs$, is defined as follows:
\[
  (\lay \lcomp \rfs)(\add) \defsym
  \begin{cases}
    \rfs(\fn) & \text{if $\add = \lay(\fn)$ for some $\fn \in\Fn$,} \\
    \vec \val[k] & \parbox{35em}{if $\add = \lay(\ar)+ k$, for some $\ar \in \Ar$ and $0 \leq k < \size{\ar}$ s.t. $\rfs(\ar) = \vec \val$,} \\
    \none & \text{otherwise,}
  \end{cases}
\]
where $\vec \val[k]$ denotes the $k$-th element of the tuple $\vec \val$, indexed
starting from $0$.

\paragraph{Randomization scheme}
Abstracting from details, we model an address space randomization
scheme through a probability distribution over layouts.  A specific
layout $\lay$ is selected at random prior to system execution.  For a
given system $\system = (\rfs,\syss,\caps)$, this choice then dictates
the initial memory configuration $\lay \lcomp \rfs$.  Although the
semantics is itself deterministic, computation can be viewed as a
probabilistic process.  In our language, instructions are
layout-sensitive: for example, the outcome of a memory load operation
at a specific address depends on whether $\lay$ places an object at
that address. Therefore, kernel's safety should be construed as a
property that holds in a probabilistic sense.


\paragraph*{Register maps} Besides memory locations, our program can
also manipulate the content of registers. In our semantics, we model
registers through functions $\regmap: \Reg \to \Val$, associating register
identifiers to their value. As for memories,
the notation $\update \regmap \vx \val$ denotes the register map that is pointwise
identical to $\regmap$, except for the register $\vx$, which is mapped to the value $\val$.
Again, this operation is only defined when $\val$
is a value, i.e., registers cannot store arrays or procedures.

\paragraph*{Semantics of expressions} To define the semantics of
expressions, we assume for each $n$-ary operator $\op$ an
interpretation $\widehat \op:\Val^n \to \Val$. The semantics of
an expression depends, besides registers, on a layout, in order to resolve identifiers.
The semantics is now defined by:
\begin{align*}
\sem\val_{\regmap, \lay} &\defsym \val & \sem\vx_{\regmap, \lay} &\defsym \regmap(\vx) & \sem \id_{\regmap, \lay} &\defsym \lay(\id) &  \sem{\op (\expr_1,\dots,\expr_n)}_{\regmap, \lay} \defsym \widehat \op (\sem{\expr_1}_{\regmap, \lay}, \ldots, \sem{\expr_n}_{\regmap, \lay} ).
\end{align*}
Let $(\cdot)^{\Add{}} : \Val \to \Add$ and
$(\cdot)^{\cBool{}} : \Val \to \cBool$ be functions that cast any
value to an address, or a Boolean, respectively. In particular,
$\sem {\expr}^{\Add}_{\regmap, \lay}$ and
$\sem {\expr}^{\cBool}_{\regmap, \lay}$ evaluate expression to
addresses and Boolean.

\paragraph*{Configurations}
Due to the presence of (possible recursive) procedures, configurations make use
of a stack of frames. 
Each such frame records the command under evaluation, the register contents
and the execution mode.
Formally, configurations are drawn from the following BNF:

{\hfill
\begin{minipage}[l]{0.25\linewidth}
  \begin{cbnf}
    \opt & \um \bnfmid \km[\syscall]
  \end{cbnf}
\end{minipage}
\begin{minipage}[l]{0.25\linewidth}
  \begin{cbnf}
    \st & \varepsilon \bnfmid \frame{\cmd}{\regmap}{\opt} : \st 
  \end{cbnf}
\end{minipage}
\begin{minipage}[l]{0.25\linewidth}
\begin{cbnf}
    \confone, \conftwo & \conf{\st,\mem} \bnfmid \err \bnfmid \unsafe. 
  \end{cbnf}
\end{minipage}
\hfill}

A \emph{configuration} of the form $\conf{\st,\mem}$, with
top-frame $\frame{\cmd}{\regmap}{\opt}$, indicates that $\cmd$ is
executed with allocated registers $\regmap$ in \emph{execution mode}
$\opt$ on memory $\mem$.  In particular, $\opt = \km[\syscall]$ indicates that
execution proceeds in privileged kernel-mode, triggered by system call
$\syscall$. The annotation of the \emph{kernel-mode} flag by a system
call name facilitates the instrumentation of the semantics.  Indeed,
every time an access to the memory is made, the semantics enforces
that address is in the capabilities of the system call that is running
(if any).  If the address can be rightfully accessed, the execution
Finally, $\err$ signals abnormal termination (for instance, when
dereferencing a pointer to kernel memory from user-mode or vice
versa).


\ifdefined\conference{\input{rules1}}\fi
\ifdefined\arxiv{\begin{figure*}[t]
  \small
  \centering
  \columnwidth=\linewidth
  \begin{framed}
    \[
      \Infer[WL][Pop]{
        \step
        {\ntc{\cnil}{\regmap}{\opt} {\frame {\cmd} {\regmap'} {\opt'}\cons \st}{\mem}}
        {\ntc{\cmd}{\update{\regmap'}{\ret} {\regmap(\ret)}}{\opt'}{\st}{\mem}}
      }
      {}
    \]
    \\[-3mm]
    \[
      \Infer[WL][Skip]
      { \step
        {\ntc{\cskip\sep\cmd}{\regmap}{\opt}{\st}{\mem}}
        {\ntc{\cmd}{\regmap}{\opt}{\st}{\mem}}
      }
      {
      }
    \]
    \\[-3mm]
    \[
      \Infer[WL][Op]
      { \step
        {\ntc{\vx \ass \expr\sep\cmd}{\regmap}{\opt}{\st}{\mem}}
        {\ntc{\cmd}{\update{\regmap}{\vx}{\sem \expr_{\regmap, \lay}}}{\opt}{\st}{\mem}}
      }
      {
      }
    \]
    \\[-3mm]
    \[
      \Infer[WL][If]
      { \step
        {\ntc{\cif \expr {\cmd_\ctrue} {\cmd_\cfalse}\sep\cmdtwo}{\regmap}{\opt}{\st}{\mem}}
        {\ntc{\cmd_{\toBool{\sem \expr_{\regmap, \lay}}}\sep \cmdtwo}{\regmap}{\opt}{\st}{\mem}}
      }
      {
      }
    \]  
    \\[-3mm]
    \[
      \Infer[WL][While]
      { \step
        {\ntc{\cwhile {\expr} {\cmd}\sep\cmdtwo}{\regmap}{\opt}{\st}{\mem}}
        {\confone_{\toBool{\sem \expr_{\regmap, \lay}}}}
      }
      {
        \confone_{\ctrue} = {\ntc{\cmd\sep\cwhile {\expr} {\cmd}\sep \cmdtwo}{\regmap}{\opt}{\st}{\mem}} &
        \confone_{\cfalse} = {\ntc{\cmdtwo}{\regmap}{\opt}{\st}{\mem}}     }
    \]
  \end{framed}
  \caption{Semantics w.r.t. system $\system=(\rfs,\syss,\caps)$, first part.}
  \label{fig:stepexcerpt1}
\end{figure*}

\begin{figure*}[t]
  \small
  \centering
  \columnwidth=\linewidth
  \begin{framed}
    \[
      \Infer[WL][Load]
      { \step
        {\ntc{\cmemread \vx \expr\sep\cmd}{\regmap}{\opt}{\st}{\mem}}
        {\ntc{\cmd}{\update{\regmap}{x}{\mem(\add)}}{\opt}{\st}{\mem}}
      }
      {
        \toAdd{\sem{\expr}_{\regmap, \lay}} = \add &
        \add \in \underline \lay(\Ar[\opt]) &
        \fbox{$\opt = \km[\syscall] \Rightarrow \add \in \underline \lay(\caps(\syscall))$}
      }
    \]
    \\[-3mm]
    \[
      \Infer[WL][Load-Error]
      {\step
        {\ntc{\cmemread \vx \expr\sep\cmd}{\regmap}{\opt}{\st}{\mem}}
        {\err }
      }
      {
        \toAdd{\sem{\expr}_{\regmap, \lay}} = \add &
        \add \not\in \underline \lay(\Ar[\opt])
      }
      \quad
      \Infer[WL][Load-Unsafe]
      {\step
        {\ntc{\cmemread \vx \expr\sep\cmd}{\regmap}{\km[\syscall]}{\st}{\mem}}
        {\unsafe}
      }
      {
        \toAdd{\sem{\expr}_{\regmap, \lay}} = \add &
        \add \in \underline \lay(\Ar[\km]) &
        \fbox{$\add \not\in \underline \lay(\caps(\syscall))$}
      }
    \]
    \\[-3mm]
    \[
      \Infer[WL][Store]
      { \step
        {\ntc{\cmemass \expr \exprtwo\sep\cmd}{\regmap}{\opt}{\st}{\mem}}
        {\ntc{\cmd}{\regmap}{\opt}{\st}{\update{\mem}{\add}{\sem \exprtwo_{\regmap, \lay}}}}
      }
      {
        \toAdd{\sem{\expr}_{\regmap, \lay}} = \add &
        \add \in \underline \lay(\Ar[\opt]) &
        \fbox{$\opt = \km[\syscall] \Rightarrow \add \in \underline \lay(\caps(\syscall))$}
      }
    \]
    \\[-3mm]
    \[
      \Infer[WL][Store-Error]
      {\step
        {\ntc{\cmemass \expr \exprtwo\sep\cmd}{\regmap}{\opt}{\st}{\mem}}
        {\err}
      }
      {\toAdd{\sem\expr_{\regmap, \lay}} = \add &
        \add \notin \underline \lay(\Ar[\opt])
      }
      \quad
      \Infer[WL][Store-Unsafe]
      {\step
        {\ntc{\cmemass \expr \exprtwo\sep\cmd}{\regmap}{\km[\syscall]}{\st}{\mem}}
        {\unsafe}
      }
      {
        \toAdd{\sem{\expr}_{\regmap, \lay}} = \add &
        \add \in \underline \lay(\Ar[\km]) &
        \fbox{$\add \not\in \underline \lay(\caps(\syscall))$}
      }
    \]
    \\[-3mm]
    \[
      \Infer[WL][Call]{
        \step
        {\ntc{\ccall{\expr}{\vec \exprtwo}\sep\cmd}{\regmap}{\opt}{\st}{\mem}}
        {
          \ntc
          {\mem(\add)}
          {\regmap_0[\vec \vx \upd \sem{\vec \exprtwo}_{\regmap,\lay}]}
          {\opt}
          {\frame{\cmd}{\regmap}{\opt} : \st}
          {\mem}
        }
      }
      {\toAdd{\sem{\expr}_{\regmap, \lay}}=\add &
        \add \in \underline \lay(\Fn[\opt]) &
        \fbox{$\opt = \km[\syscall] \Rightarrow \add \in \underline \lay(\caps(\syscall))$}
      }
    \]
    \\[-3mm]
    \[
      \Infer[WL][Call-Error]{
        \step
        {\ntc{\ccall{\expr}{\vec\exprtwo}\sep\cmd}{\regmap}{\opt}{\st}{\mem}}
        {\err}
      }
      {
        \toAdd{\sem{\expr}_{\regmap, \lay}} = \add &
        \add \not\in \underline \lay(\Fn[\opt])
      }
      \quad
      \Infer[WL][Call-Unsafe]{
        \step
        {\ntc{\ccall{\expr}{\vec \exprtwo}\sep\cmd}{\regmap}{\km[\syscall]}{\st}{\mem}}
        {\unsafe}
      }
      {
        \toAdd{\sem{\expr}_{\regmap, \lay}} = \add &
        \add \in \underline \lay(\Fn[\km]) &
        \fbox{$\add \not\in \underline \lay(\caps(\syscall))$}
      }
    \]
    \\[-3mm]
    \[
      \Infer[WL][SystemCall][\textsc{SC}]{
        \step
        {\ntc{\csyscall{\syscall}{\vec \exprtwo}\sep\cmd}{\regmap}{\opt}{\st}{\mem}}
        {
          \ntc
          {\syss(\syscall)}
          {\update {\regmap_0}{\vec \vx}{\sem{\vec\exprtwo}_{\regmap,\lay}}}
          {b'}
          {\frame{\cmd}{\regmap}{\opt} : \st}
          {\mem}
        }
      }
      {
        b = \um \Rightarrow b' = \km[\syscall] &
        b = \km[\syscalltwo] \Rightarrow b' = \km[\syscalltwo] 
      }
    \]
  \end{framed}
  \caption{Semantics w.r.t. system $\system=(\rfs,\syss,\caps)$, second part.}
  \label{fig:stepexcerpt2}
\end{figure*}}\fi

\paragraph{Small step operational semantics} Transitions in our semantics take the form
\[
  \step{\confone}{\conftwo},
\]
indicating that, w.r.t. system $\system$, configuration $\confone$
reduces to $\conftwo$ in one step, under layout $\lay$.
\ifdefined\conference{The most important reduction rules are defined
    in \Cref{fig:stepexcerpt}}\fi\ifdefined\arxiv{The reduction rules
    are defined in \Cref{fig:stepexcerpt1,fig:stepexcerpt2}}\fi.
%
Rule~\ref{WL:Load} implements a successful memory load $\cmemread \vx \expr$.
Expression $\expr$ is evaluated to an address $\add = \toAdd{\sem{\expr}_{\regmap, \lay}}$,
and the content of the register $\vx$ is updated with the value $\mem(\add)$. 
The side-condition $\add \in \underline \lay(\Ar[\opt])$ enforces
that $\add$ refers to a value accessible in the current execution mode $\opt$ (by slight abuse of notation, we disregard
the system call label in kernel-mode),
otherwise the instruction leads to $\err$ (see Rule~\ref{WL:Load-Error}).
As such, we are modeling unprivileged execution and SMAP protection, preventing
the access of kernel-space addresses when in user-mode, and vice versa.
The final, boxed, side-condition refers to the safety instrumentation.
In kernel-mode, triggered by system call $\syscall$ ($\opt =\km[\syscall]$), the rule
ensures that $\add$ refers to an object within the capabilities of $\syscall$ ($\add \in \underline w(\caps(\syscall))$).
When this condition is violated, unsafe execution is signaled (see Rule~\ref{WL:Load-Unsafe}).
In a similar fashion, the rules for memory writes and procedure calls
are defined.


Rule~\ref{WL:Call} deals with procedure calls.  It opens a new frame
and, places the $n$ evaluated arguments
$\exprtwo_1, \ldots, \exprtwo_n=\vec\exprtwo$ at registers
$\vx_1,\dots,\vx_n$ in an initial register environment $\regmap_0$,
summarized by the notation
$\regmap_0[\vec \vx \upd \sem{\vec \exprtwo}_{\regmap,\lay}]$. Notice
that our choice does not exclude stack-based inter-procedural
communication from our model: procedures can use a dedicated array as
stack, and pass the stack and return pointers as arguments.  System
calls, modeled by Rule~\ref{WL:SystemCall}, follow the same calling
convention.  Note that, in the newly created frame, the execution flag
is set to kernel-mode.  Once a procedure or system call finished
evaluation, Rule~\ref{WL:Pop} removes the introduced frame from the
stack. Observe how the rule permits return values through a designated
register $\ret$.  The remaining rules are standard.

Let us denote by $\lay \red \confone \to^* \conftwo$ that configuration $\confone$ reduces in zero or more steps to configuration $\conftwo$,
and by $\diverge \confone$ that $\confone$ \emph{diverges}.
In our semantics, under layout $\lay$, any non-diverging computation either
halts in a terminal configuration of the form $\conf{\frame{\cnil}{\regmap}{\opt}, \lay \lcomp \store'}$,
or abnormally terminates through an error $\err$, or safety violation $\unsafe$.
This motivates the following definition of an evaluation function:

\[
  \Eval {\cmd,\regmap,\opt,\store} \defsym
  \begin{cases}
    (\val,\store')
    & \!\!\!\text{if } \lay \calign{\red \conf{\frame{\cmd}{\regmap}{\opt},\lay \lcomp \store} \to^* \conf{\frame{\cnil}{\regmap'[\ret \mapsto \val]}{\opt}, \lay \lcomp \store'},} \\
    \err
    & \!\!\!\text{if } \lay \red \conf{\frame{\cmd}{\regmap}{\opt},\lay \lcomp \store} \to^* \err, \\
    \unsafe
    & \!\!\!\text{if } \lay \red \conf{\frame{\cmd}{\regmap}{\opt},\lay \lcomp \store} \to^* \unsafe\\
      \Omega
    & \!\!\!\text{if $\diverge{\conf{\frame{\cmd}{\regmap}{\opt},\lay \lcomp \store}}$}.
  \end{cases}
\]

Note how, in the case of normal termination, a computation produces a pair of a return value and a store.


\section{Threat Model}
\label{sec:threatmodel}
In our threat model, attackers are unprivileged user-space
programs that execute on a machine supporting two
privilege rings: user-mode and kernel-mode.
The victim is the host operating system which
runs in kernel mode and has exclusive access to its private memory.
In particular, the operating system exposes a set of procedures,
the system calls, that can be invoked by the attacker and
that have access to kernel's memory. The attacker's goal is
to trigger a system call to perform an unsafe memory
access.

In \Cref{sec:safety1}, attackers are ordinary
programs that do not control speculative execution and do not have access
to side-channel info-leaks. However, the target machine implements
standard mitigations against this kind of attacks.
In particular, it supports data execution protection mechanisms (DEP),
SMAP~\cite{LWNSMAP} that prevents kernel-mode access
to user-space data, and SMEP~\cite{SMEP} that prevents the execution
of user-space functions when running in kernel-mode.
More precisely, the above-mentioned protection mechanisms are
modeled in our semantics by the preconditions of the rules
\ref{WL:Call}, \ref{WL:Load}, \ref{WL:Store} that prevent the system from:
(i) overwriting functions, (ii) execute values, (iii)
accessing user-space data and functions when the system is in
kernel-mode. Most importantly, the system adopts kernel address
space layout randomization, that is modeled by
sampling the memory layout from a probability distribution.

In \Cref{sec:safety2}, we then consider a stronger threat model
where, in addition, attackers have access to side-channel observations and
control PHT and STL predictions, related to Spectre v1 and
v4 vulnerabilities~\cite{Spectre}. In addition to the above-mentioned
mechanisms against speculative attacks, the machine
supports PTI~\cite{PTI} to prevent the speculative
access of kernel-space memory from user-space;
this is modeled by using the same preconditions of rules
\ref{WL:Call}, \ref{WL:Load}, \ref{WL:Store} for their
speculative counterparts, see~\Cref{sec:internalsemantics}.


\section{Classic Threat Model}
\label{sec:safety1}

In this section we show how the result of Abadi et. al. \cite{Abadi,Abadi2,Abadi3} scales
to the model introduced in Section~\ref{sec:language}.
The safety property that we aim at is defined in terms of our instrumented semantics, as follows:

\begin{definition}[Kernel safety]\label{def:ks}
  We say that a system $\system = (\rfs,\syss,\caps)$ is \emph{kernel safe},
  if for every layout $\lay$, 
  \emph{unprivileged} attacker $\cmd \in \Cmd$, 
  and registers $\regmap$, we have: 
  \[
    \lnot \left(\lay \red \conf{\frame{\cmd}{\regmap}{\um}, \lay \lcomp \rfs} \to^* \unsafe \right).
  \]
\end{definition}
\def\ex{\cc[v,s1,s2,f,s]}
Thus, safety is broken if an attacker $\cmd$, executing in unprivileged user mode, is able to trigger a system call
in such a way that it accesses, or invokes, a kernel-space object outside its capabilities.
The source of such a safety violation can be twofold:
\begin{varenumerate}
  \item\label{issue1} \textbf{Scope extrusion.}
  An obvious reason why kernel-safety may fail is due to apparent
  communication channels, specifically through the memory and procedure returns.
  As an example, consider a system $\system =(\rfs, \syss, \caps)$,
  where:
  \begin{align*}
    \syss(\syscall_1) &\defsym \cmemass \ar \fn & \syss(\syscall_2) &\defsym \cmemread \vx \ar ;\ccall \vx {}   & \caps (\syscall_1)=\caps(\syscall_2)=\{\ar\}
  \end{align*}
  A malicious program can use $\syscall_1$ to store the address of
$\fn$ at $\ar[0]$, which is a shared capability. A consecutive call
to $\syscall_2$ then breaks safety if $\fn$ is not within the capabilities of $\syscall_2$.
\item\label{issue2} \textbf{Probing.} Another counterexample is given
  by a system call accessing memory based on its input, such as the
  system call that only contains the instruction $\ccall {\vx_1}{}$,
  which directly invokes the procedure stored at the kernel-address
  corresponding to the value of its first argument $\vx_1$.  This
  system call can potentially be used as a gadget to invoke an
  arbitrary kernel-space procedure from user-space. Since an attacker
  lacks knowledge of the kernel-space layout, such an invocation needs
  to happen effectively through probing.  As any probe of an unused
  memory address leads to an unrecoverable error,\footnote{%
    This is not always the case for \emph{user-space} software
    protected with layout randomization, as some programs (e.g. web
    servers) may automatically restart after a crash to ensure
    availability.  This behavior can be exploited by attackers to
    probe the entire memory space of the victim program, thus
    compromising the protection offered by layout
    randomization~\cite{ApacheAttack}.  } the likelihood of an unsafe
  memory access is, albeit not zero, diminishingly small when the
  address-space is reasonably large.
\end{varenumerate}
To overcome Issue~\ref{issue1}, we impose a form of (layout) \emph{non-interference}
on system calls.
\begin{definition}[Layout non-interference]
  \label{def:lni}
  Given $\system=(\rfs,\syss,\caps)$,
  a system call $\syscall$ is \emph{layout non-interfering}, if,
  \[
    \Eval[\system][\lay_1]{\syss(\syscall), \regmap, \km[\syscall], \store'}
    \evaleq
    \Eval[\system][\lay_2]{\syss(\syscall), \regmap, \km[\syscall], \store'}
  \]
  for all layouts $\lay_{1},\lay_{2}$, registers $\regmap$ and stores $\store' \eqon{\Fn} \rfs$.
  Here, the equivalence $\evaleq$ extends equality with $\err\evaleq \unsafe$ and $\unsafe\evaleq \err$.
  The system $\system$ is non-interfering if all its system calls are.
\end{definition}
In effect, layout non-interfering systems do not expose layout information,
neither through the memory nor through return values.
In particular, observe how non-interference rules out Issue~\ref{issue1},
as witnessed by two layouts placing $\fn$ at different addresses in kernel-memory.

Concerning Issue~\ref{issue2}, it is well known that layout
randomization provides in general safety not in an absolute sense, but
\emph{probabilistically}~\cite{Abadi,DieHard,PaXASLR}.  Indeed, the
chance for a probe to be successful depends on the randomization
scheme.  Following \citet{Abadi}, let $\mu$ be a \emph{probability
  distribution} of layouts, i.e., a function $\mu : \Lay \to [0,1]$
assigning to each layout $\lay \in \Lay$ a probability $\mu(\lay)$
(where $\sum_{\lay \in \Lay} \mu(\lay) = 1$).  Without loss of
generality, we assume that the layout of public, i.e. user-space,
addresses is fixed.  That is, we require for each $\lay_{1},\lay_{2}$
with non-zero probability in $\mu$, that $\lay_1(\id)=\lay_2(\id)$ for
all $\id \in \Idu$.  For a distribution of layouts $\mu$ and a system
$\system = (\rfs, \syss, \caps)$, the value of $\delta_{\mu, \system}$
quantifies the smallest probability that a probe for an address
$\add \in \Addk$ fails.  To formally define $\delta_{\mu, \system}$, given a
system call $\syscall$ we denote with
$\id_1^\syscall, \dots,\id_k^\syscall$ the enumeration of its
references $\refs_\system(\syss(\syscall))$.  The value $\delta_{\mu, \system}$
can then be defined as follows:
\begin{multline*}
  \delta_{\mu, \system} \defsym
  \min \bigl \{
  \displaystyle{\Pr_{\lay\leftarrow \mu}}
  [ \add \notin \underline\lay(\Idk) \mid \lay(\id_i^\syscall)=p_i,
  \text{ for } 1\le i \le h]
  \mid \syscall \in \Sys, p,p_1,\dots,p_h \in \Addk \land {}\\
  \phantom{{} \mid} p \notin \{p_i, \ldots, p_i +\size{\id_i^\syscall}-1\}, \text{ for } 1\le i \le h
  \bigr \}.
\end{multline*}
In practice, $\delta_{\mu, \system}$ bounds the probability that,
during the execution of a system call $\syscall$, a fixed kernel
address $\add$ is not allocated, given that it does not store any
object that is in the references of that system call. Notably, if an
attacker controls the value of a kernel address $\add$,
$\delta_{\mu, \system}$ is a lower bound to the probability that its
probe for $\add$ does not hit any memory content.
This property
is reflected by the cases \ref{case:lemma12} and \ref{case:lemma2b} of
\Cref{lemma:onsyscall,lemma:onsyscallterm} below.  More precisely,
\Cref{lemma:onsyscall} proves this property for fixed-length reductions,
and \Cref{lemma:onsyscallterm} lifts it to full evaluations.
By considering the
complementary event, $\delta_{\mu, \system}$ gives an upper bound to
the probability of performing a safety violation in the presence of
layout randomization, as expressed by \Cref{thm:scenario1} below.

We now give the formal statements of
\Cref{lemma:onsyscall,lemma:onsyscallterm} together with their proofs.
To this aim, we extend $\refs_\sigma$ to frame stacks as follows:
\begin{align*}
  \refs_\sigma(\nil) &\defsym \emptyset&
  \refs_\sigma(f\cons \st) &\defsym \refs_\sigma(f) \cup \refs_\sigma(\st) &
  \refs_\sigma(\frame \cmd \regmap \opt) &\defsym \refs_\sigma(\cmd).
\end{align*}

\begin{restatable}{lemma}{onsyscall}
  \label{lemma:onsyscall}
  Let $\syscall$ be a system call of a \emph{layout non-interfering} system $\system=(\rfs, \syss, \caps)$, and let $\refs_\system(\syss(\syscall))\setminus \Sys=\{\id_1, \dots, \id_h\}$ be identifiers within the references of $\syscall$. Given a sequence of addresses $\add_1, \dots, \add_h$, an initial frame $\frame{\syss(\syscall)}{\regmap}{\km[\syscall]}$, and a store $\rfs'$ such that $\rfs'\eqon{\Fn}\rfs$,
  for every reduction length $\nat\in \Nat$ and distribution of layouts $\mu$, one of the following statements holds:
  
  \begin{enumerate}[label={(\arabic*)}]
    \item For every layout $\lay$ such that $\forall 1\le i\le h$, $\lay(\id_i)=\add_i$, we have \(
      \nstep{\nat}
      {\conf{\frame{\syss(\syscall)}{\regmap}{\km[\syscall]}, \lay \lcomp \rfs'}}
      {\conf{\overline \st, \lay \lcomp \overline \rfs}}.
      \)
      For some non-empty stack $\overline \st$
      such that
      $\refs_\system(\overline\st)\subseteq \refs_\system(\syss(\syscall))$, and a store
      $\overline \rfs\eqon{\Fn}\rfs'$.
    \item\label{case:lemma12}
    \(
      \Pr_{\lay \leftarrow \mu}\Big[ \exists \nat' \le \nat. \nstep {\nat'} {\conf{\frame{\syss(\syscall)}{\regmap}{\km[\syscall]}, \lay\lcomp \rfs'}}\err \,\,\Big|\,\,
      \forall 1\le i\le h.\lay(\id_i)=\add_i\Big]\ge \delta_{\mu, \system}.
    \)

  \item For every layout $\lay$ such that $\forall 1\le i\le h$, $\lay(\id_i)=\add_i$, we have   \(
      \nstep{\nat'}
      {\conf{\frame{\syss(\syscall)}{\regmap}{\km[\syscall]}, \lay \lcomp \rfs'}}
      {\conf{\frame{\cnil}{\overline\regmap}{\km[\syscall]}, \lay \lcomp \overline \rfs}}
      \)
      for some $\nat' \le\nat$,  $\overline \regmap$ and
    store $\overline \rfs\eqon{\Fn}\rfs'$.
  \end{enumerate}
\end{restatable}

\begin{proof}[Proof sketch of \Cref{lemma:onsyscall}]
  The proof proceeds by induction on $\nat$ and case analysis on the transition rules. In the base case, it is trivial to establish claim (A). In the inductive case, by applying the IH to the initial configuration. Three cases arise:
  \begin{proofcases}
    \proofcase{A} Let $\system$ be a system, $\syscall$ a system call
    and consider the initial configuration
    $ {\conf{\frame{\syss(\syscall)}{\regmap}{\km[\syscall]}, \lay
        \lcomp \rfs'}}$.  Suppose
    $\refs_\system(\syss(\syscall))\setminus \Sys= \{\id_1, \dots,
    \id_h\}$ and fix addresses $\add_1, \dots, \add_h$.  By the IH,
    there exists a stack $\st$ and a
    store $\rfs'' \eqon{\Fn} \rfs$ satisfying the following property:
    \[
      \forall \lay. \left( \forall 1 \leq i \leq h, \lay(\id_i) = \add_i\right)  \Rightarrow
      \nstep{\nat}
      {\conf{\frame{\syss(\syscall)}{\regmap}{\km[\syscall]}, \lay \lcomp \rfs'}}
      {\conf{\st, \lay \lcomp \rfs''}}.
    \]
    Since $\st$ is non-empty, we can assume that
    $\st = \frame \cmd {\regmap'} {\km[\syscall]}:\st'$.  Furthermore,
    from the IH, we know that
    $\refs_\system(\st)\subseteq \refs_\system(\syss(\syscall))$
    (H). The proof proceeds with a case analysis on $\cmd$. We only
    focus on the representative case of procedure invocation.
    
    \begin{proofcases}
      \proofcase{$\ccall \expr {\exprtwo_1, \dots, \exprtwo_k}\sep
        \cmdtwo$} We start by observing that by (H) there exists a
      unique address $\add$ such that for every $\lay$ that satisfies
      the precondition (that stores the references of $\syscall$ at
      addresses $p_1, \ldots, p_h$), we have
      $\toAdd{\sem \expr_{\regmap,\lay}}=\add$. Similarly, we
      introduce the values $\val_1, \dots,\val_k$, which correspond to
      the semantics of $\exprtwo_1,\dots,\exprtwo_k$ evaluated under
      $\regmap$ and every layout that satisfies the precondition.
      Again for the same reason, we observe that there exists a set
      $P$ such that for each layout under consideration, it holds that
      $\underline{\lay} (\refs_\system(\syss(\syscall))\setminus
      \Sys)=P$. The proof proceeds by cases on whether $\add \in P$.
      \begin{proofcases}
        \proofcase{$\add\in P$} In this case, there is a unique identifier $\id_j$ such that for every layout $\lay$ that satisfies the precondition, we have $\add \in \underline{\lay}(\id_j)$. We analyze two cases based on whether $\id_j$ is a function identifier $ \fn$.
        \begin{proofcases}
          \proofcase{$\id_j= \fn$} In this case, from the definition
          of $\lcomp$, we deduce that for each of these layouts we
          have $\lay\lcomp \rfs''(\add) = \rfs''(\fn) = \rfs(\fn)$,
          where the last step follows from the assumption
          $\rfs''\eqon{\Fn}\rfs$.  Since $\add \in P$, we conclude
          that, independently of the specific layout, if the
          preconditions hold, then:
          \begin{equation*}
            \step
            {\conf{\frame{\ccall  \expr {\exprtwo_1, \dots, \exprtwo_k}\sep \cmdtwo}{\regmap}{\km[\syscall]}:\st', \lay\lcomp \rfs''}}{}\\
            {\conf{\frame{\rfs(\fn)}{\regmap_0'}{\km[\syscall]}:\frame{\cmdtwo}{\regmap}{\km[\syscall]}:\st', \lay\lcomp \rfs''}},
          \end{equation*}
          where $\regmap_0'= \regmap_0[\vx_1,\dots,\vx_k\upd \val_1,\dots,\val_k]$.
          Finally, we observe that, by the definition of
          $\refs_\system$, $\refs_\system(\syss(\syscall))$
          contains all the identifiers within
          $\rfs(\fn) = \rfs(\id_j)$ because $\id_j \in \refs_\system(\syss(\syscall))$ and
          $\refs$ is closed under procedure calls.
          This shows that (A) holds.
          \proofcase{$\id_j =\ar$} In this case, since the set of
          array identifiers and that of functions are disjoint, we
          conclude that for every layout $\lay$ that satisfies the
          preconditions, we have that
          $\add \notin \underline \lay(\Fn[{\km}])$.  This means that
          for each of these layouts, we can show:
          \begin{equation*}
            \step
            {\conf{\frame{\ccall  \expr {\exprtwo_1, \dots, \exprtwo_k}\sep \cmdtwo}{\regmap}{\km[\syscall]}:\st', \lay\lcomp \rfs''}}
            {\err},
          \end{equation*}
          and this means that (B) holds with probability 1.
        \end{proofcases}
        \proofcase{$\add\notin P$}
        %
        Observe that, for every layout $\lay$ such that
        $\add \notin \underline{\lay}(\Idk)$,
        only rule \ref{WL:Call-Error} applies, which shows the following transition:
        \begin{equation*}
          \step
          {\conf{\frame{\ccall  \expr {\exprtwo_1, \dots, \exprtwo_k}\sep \cmdtwo}{\regmap}{\km[\syscall]}:\st', \lay\lcomp \rfs''}}
          {\err},
        \end{equation*}
        Thus, we observe that:
        \begin{equation*}
          \Pr_{\lay \leftarrow\mu}\big[
          \step
          {\conf{\frame{\ccall  \expr {\exprtwo_1, \dots, \exprtwo_k}\sep \cmdtwo}{\regmap}{\km[\syscall]}:\st', \lay\lcomp \rfs''}}
          \err\,\, \big|\\
          \forall 1\le i\le h.\lay(\id_i)=\add_i\big]
        \end{equation*}
        is greater than
        \[
          \Pr_{\lay \leftarrow\mu}\big[ \add \notin \underline\lay(\Idk) \mid
          \forall 1\le i\le h.\lay(\id_i)=\add_i\big]
        \]
        which, by definition, is greater than $\delta_{\mu}$.  
        This shows that (B) holds.
      \end{proofcases}
    \end{proofcases}
  \end{proofcases}
\end{proof}

Observe that, in the proof of \Cref{lemma:onsyscall}, the case where
$\add \notin P$ also covers situations where memory is accessed via a
\emph{raw reference} during a system call, such as when dereferencing
a constant pointer. In this scenario, \Cref{lemma:onsyscall}
establishes claim (B), emphasizing that such practices should be
avoided in kernel code, as they are highly likely to result in memory
violations.

\begin{restatable}{lemma}{onsyscallterm}
  \label{lemma:onsyscallterm}
    Let $\syscall$ be a system call of a \emph{layout non-interfering} system $\system=(\rfs, \syss, \caps)$. Given an initial frame $\frame{\syss(\syscall)}{\regmap}{\km[\syscall]}$, and a store $\rfs'$ such that $\rfs'\eqon{\Fn}\rfs$,
  for every distribution of layouts $\mu$, one of the following statements holds:

  
  \begin{enumerate}[label={(\Alph*)}]
    \item For every layout $\lay$, we have
    \(
      \Eval{\syss(\syscall), \regmap, \rfs', \km[\syscall]} = (\overline \rfs, \overline v),
    \)
    for some $\overline \val$ and $\overline\rfs\eqon{\Fn}\rfs$.
    \item\label{case:lemma2b}
    \(
      \Pr_{\lay\leftarrow \mu}\left[\Eval{\syss(\syscall), \regmap, \km[\syscall], \rfs'} = \err\right] \ge \delta_{\mu, \system}.
    \)

    \item For every layout $\lay$, we have:
    \(
      \Eval{\syss(\syscall), \regmap, \km[\syscall], \rfs'} = \Omega.
    \)
  \end{enumerate}
\end{restatable}

\begin{proof}[Proof Sketch of \Cref{lemma:onsyscallterm}]
  The proof is by case analysis: if for some layout $\lay$, $\Eval{\syss(\syscall), \regmap, \km[\syscall], \rfs'} \notin \{\err, \unsafe\}$, claims (A) or (C) must hold because of layout non-interference. On the other hand, if all layouts lead to $\err$ or to $\unsafe$, due to the finiteness of $\Lay$, there exists an upper bound $t$ on the number of steps needed to reach a terminal configuration. 
  We call $\refs(\syss(\syscall))\setminus \Sys=\{\id_1, \dots, \id_h\}$ and apply \Cref{lemma:onsyscall} with $\nat=t$. Observe that for every choice of $p_1, \ldots, p_h \in \Addk$, we can refuse conclusions (1) and (3) because they are contradictory with $\forall \lay. \Eval{\syss(\syscall), \regmap, \km[\syscall], \rfs'} \in \{\err, \unsafe\}$, so (2) must hold.
  Observe that:
  \[
    \Pr_{\lay\leftarrow \mu}\left[\Eval{\syss(\syscall), \regmap, \km[\syscall], \rfs'} = \err\right] = 
    \Pr_{\lay \leftarrow \mu}\Big[ \exists \nat' \le \nat. \nstep {\nat'} {\conf{\frame{\syss(\syscall)}{\regmap}{\km[\syscall]}, \lay\lcomp \rfs'}}\err].
  \]
  The right-hand-side, in turn, is equal to:
    \begin{multline*}
    \sum_{\add_1, \ldots, \add_k}\Pr_{\lay \leftarrow \mu}\Big[ \exists \nat' \le \nat. \nstep {\nat'} {\conf{\frame{\syss(\syscall)}{\regmap}{\km[\syscall]}, \lay\lcomp \rfs'}}\err \,\,\Big|\,\,
    \forall 1\le i\le h.\lay(\id_i)=\add_i\Big]\cdot\\
    \Pr_{\lay \leftarrow \mu}\Big[\forall 1\le i\le h.\lay(\id_i)=\add_i\Big]. 
  \end{multline*}
  Since all the factors on the left
  are bounded by $\delta_{\mu, \system}$ (2), their
  convex combination is also bounded, proving (B).
\end{proof}

We arrive now at the main result of this section, where we show the effectiveness of
layout randomization:

\begin{theorem}
  \label{thm:scenario1}
  Let $\system=(\rfs, \syss, \caps)$ be \emph{layout non-interfering}.
  Then, for any \emph{unprivileged} attacker $\cmd \in \Cmd$ and register map $\regmap$,
  \(
    \Prob{\lay \leftarrow \mu}\left[\Eval{\cmd, \regmap, \um, \rfs} =\unsafe\right] \leq 1 - \delta_{\mu, \system}
  \).
\end{theorem}

\begin{proof}
  The $\unsafe$ state can only be reached during the execution of a system call. \Cref{lemma:onsyscallterm} provides a lower bound on the probability of reaching $\err$, which we can use to bound on the probability of reaching $\unsafe$.
\end{proof}

\Cref{thm:scenario1}
extends the results of \cite{Abadi,Abadi2,Abadi3} by showing
that layout randomization guarantees \emph{kernel safety} probabilistically to
operating systems; in contrast with \cite{Abadi,Abadi2,Abadi3}, this holds even
when victim's code contains unsafe programming constructs
such as arbitrary pointer arithmetic and indirect jumps. This is achieved by
replacing \Citet{Abadi}'s restrictions
on the syntax of the victims with a \emph{weaker} dynamic property: \emph{layout
  non-interference}. 
Notice that the strength of the security guarantee provided by
\Cref{thm:scenario1} depends on the distribution of the layouts $\mu$.
Therefore, in practice, it is important to determine a randomization
scheme that provides a good bound.  This can be done quite easily: for
instance, if we assume that
(i)~$\kappa_{\km} > \sum_{\id \in \Id[\km]} \size \id$ and that
(ii)~$\theta \defsym \max_{\id \in \Id[\km]} (\size \id)$ divides
$\kappa_{\km}$, we can think of the kernel space address range as
divided in $\frac {\kappa_{\km}} \theta$ slots, each one large enough
to store any procedure or array.  In this setting, we can define the
distribution $\nu$ as the uniform distribution of all the layouts that
store each \emph{memory object} within a \emph{slot} starting from the
beginning of that slot.  For this simple scheme, we can approximate
the bound $\delta_\nu$ as the ratio between unallocated slots and all
the slots that do not store any object that is in the reference of a
system call:
\begin{equation*}
  \label{eq:deltanu}
  \delta_\nu \ge 
  \min_{\syscall \in \Sys}\frac{\kappa_{\km}/\theta-|{\Idk}|} {\kappa_{\km}/\theta-|{\refs_\system(\syss(\syscall))\setminus \Sys}|}.
\end{equation*}
In particular, the fraction in the right-hand side is the probability that
by choosing a slot that is not storing any object referenced by $\syscall$,
we end up with a fully unoccupied slot.
Observe that this lower-bound approaches
$1$ when $\kappa_{\km}$ goes to infinity.

\paragraph{Relation of kernel safety with security} Kernel safety
encompasses some form of \emph{spatial memory safety} and of
\emph{control flow integrity}, which are among the most critical
security properties for operating systems' kernels. This importance is
reflected in the numerous measures developed to enforce such
properties~\cite{kCFI,RustInLinux,towardkernelsafety,HyperSafe,FineGrainedkCFI,CompProtKer}.
Often, definitions of \emph{spatial memory safety} associate a
software component (a program, an instruction, or even a variable)
with a fixed memory area, that this component can access
rightfully~\cite{HighAssurance,SoftBound,MSWasm,PierceMS}. In this
realm, any load or store operation that does not fall within this area
is considered a violation of \emph{spatial memory safety}.
Our notion of \emph{kernel safety} encompasses a form of \emph{spatial
  memory safety}: if a system enjoys \emph{kernel safety}, then no
system call can access a memory region that does not appear within its
\tcaps.
In addition, \emph{kernel safety} also implies a form of \emph{control
  flow integrity}: a property which requires that the control transfer
operations performed by a program can reach only specific statically
determined targets~\cite{CFI}.  More precisely, in our semantics it is
unsafe to execute of a procedure if its address does not belong to the
set of \tcaps of the current system call. This means that, if a system
$\system = (\rfs, \syss, \caps)$ is \emph{kernel safe}, by executing
that system call, the control flow will flow across the procedures in
$\caps(\syscall)$.

\paragraph{Final remarks}
Kernel-space layout randomization provides \emph{kernel safety} in a
probabilistic sense. In particular, the probability of violating
safety depends on the randomization scheme, and for certain
randomization schemes, it approaches 0 when the size of the address
space goes to infinity.  In turn, \emph{kernel safety} implies critical
security guarantees, including forms \emph{spatial memory safety} and
\emph{control flow integrity}.


\section{Speculative Threat Model}
\label{sec:safety2}
In this section we establish to which extent
a system enjoys \emph{kernel safety} in presence of speculative attackers.
To this aim, in \Cref{subsec:specmodel}, we extend the model of
\Cref{sec:language} for this new scenario. More precisely,
we endow the semantics of \Cref{sec:language}
with speculative execution and side-channel observations
that reveal the accessed addresses and the value of conditional branches~\cite{HighAssurance,CTFundations,Spectector,Contracts}.
In \Cref{subsec:specsafety},
we refine the notion of \emph{kernel safety} for this model,
by defining \emph{speculative kernel safety}.

\subsection{The Speculative Execution Model}
\label{subsec:specmodel}

A substantial difference between our model and previous models~\cite{HighAssurance,CTFundations,
  Spectector,Contracts}
lies in the possibility to explicitly model attackers. More precisely, an attacker is not
given as a mere sequence of microarchitectural directives,
but becomes a fully-fledged program that can directly
interact with the system.
This permits us to naturally extend the notion of \emph{kernel safety}
to the new scenario. Besides, we believe that modeling an attack explicitly can be
interesting on its own.  Feasibility of an attack is witnessed explicitly
through a program. In this setting, for instance, assumptions on
the attacker's computational capabilities can be imposed seamlessly. 



\DD{Relate with crypto's CPA property??}

\subsubsection{Victim Language and Semantics}
\label{sec:internalsemantics}

The victims' language remains identical to the classic model.
To permit attackers to influence the speculative execution of specific instructions,
we assume load and branch instructions are tagged by unique labels $\lbl \in \Lbl$,
We enrich the language with a \emph{fence instruction} found
in modern CPUs~\cite{IntelManual}:
\begin{cbnf}
  \Instr \ni \stat & \dots \mid \cfence .
\end{cbnf}
Architecturally, this instruction is a no-op,
on the microarchitecture level
it commits all buffered writes to memory.
Following~\citet{HighAssurance}, the speculative semantics is instrumented through \emph{directives}, modeling the choice made by
prediction units of the processor. Directives take the form
\begin{cbnf}
  \dir \ni \Dir
  & \dbranch[\lbl] \bool
  \bnfmid \dload[\lbl]{i}
  \bnfmid \dbt
  \bnfmid \dstep,
\end{cbnf}
where $i \in \Nat$ and $\bool \in \cBool$.
The $\dbranch[\lbl] \bool$ directive causes a branch
instruction to be evaluated as if the guard resolved to
$b$. The $\dload[\lbl]{i}$ causes the load instruction to load the
$i$-th most recent value that is associated to an address in a
(buffered) memory. The $\dbt$ directive is used to direct speculations,
either backtracking the most recent mis-speculation or committing
the microarchitectural state.
Finally, the $\dstep$ directive evaluates an instruction
without engaging into speculation, in correspondence to the semantics
we have given in \Cref{sec:language}.

The semantics is also instrumented with observations to model timing side-channel leakage:
\begin{cbnf}
  \obs, \obstwo \ni \Obs
  & \onone
  \bnfmid \obranch \bool
  \bnfmid \omem \add
  \bnfmid \ojump \add
  \bnfmid \obt \bool,
\end{cbnf}
where $n \in \Nat$ $b\in \cBool$, and $\add \in \Add$.
We use $\onone$ to label transitions that do not leak observations.
The $\obranch \bool$ observation is caused
by branching instructions, with $\bool$ reflecting the taken branch.
The $\omem \add$ observation is caused by memory access, through loads or stores,
and contains the address of the accessed location, thus modeling
instruction-cache leaks.
Likewise, the $\ojump \add$ observation is caused by calls to
procedures residing at address $\add$ in memory.
Finally, the $\obt \bool$ observation signals a
backtracking step during speculative execution.
Notice that we leak full addresses on memory accesses, and
the value of the branching instructions, i.e. we adopt
the \emph{baseline leakage model} that is
widely employed in the literature to model side-channel
info-leaks~\cite{CTPolicies, HighAssurance,
  CTFundations, SLNonInt}.

\ifdefined\conference{\input{rules2}}\fi
\ifdefined\arxiv{
\begin{figure*}[t]
  \small
  \centering
  \begin{framed}
    \vspace{0.5ex}
    
  \resizebox{\textwidth}{!}{\(
      \Infer[SIE][SLoad][\textsc{SLoad-Load}]
      {\sstep
        {\confone \cons\cfstack}
        {\sframe{\frame{\cmd}{\update \regmap x \val}{\opt}\cons\st}{\bm\buf\mem}{\boolms\lor f}\cons
          \confone \cons \cfstack}  {{{\dload[\lbl]{i}}}} {\omem \add}}
      {\confone = \sframe{\frame{\cmemread[\lbl] \vx  \expr\sep\cmd}{\regmap}{\opt}\cons \st} {\bm\buf\mem}{\boolms} &
        \toAdd{\sem\expr_{\regmap, \lay}} = \add &
        \bufread {\bm\buf\mem} \add i =(\val, f) &
        \add \in \underline \lay(\Ar[\opt]) &
        \fbox{$\opt = \km[\syscall] \Rightarrow \add \in \underline \lay(\caps(\syscall))$}
      }
    \)}

    \[
      \Infer[SIE][SLoad-Error][\textsc{SLoad-}\\\textsc{Error}]
      {
        \lay \red
        \sframe{\frame{\cmemread[\lbl] \vx \expr\sep\cmd}{\regmap}{\opt}\cons \st}{\bm\buf\mem}{\boolms}\cons\cfstack
        \sto{\dload[\lbl]{i}}{\onone}
        \sconf{\err, \boolms}\cons\cfstack
      }{
        \toAdd{\sem\expr_{\regmap, \lay}} = \add &
        \add \notin \underline \lay(\Ar[\opt]) &
      }
    \]

    \[
      \Infer[SIE][SLoad-Unsafe][\textsc{SLoad-}\\\textsc{Unsafe}]
      {\sstep
        {\sframe{\frame{\cmemread[\lbl] \vx \expr\sep\cmd}{\regmap}{\km[\syscall]}\cons \st}{\bm\buf\mem}{\boolms}\cons\cfstack}
        {{\unsafe}}   {\dir} {\omem \add}}
      {\toAdd{\sem\expr_{\regmap, \lay}} = \add &
        \add \in \underline \lay(\Ar[\km]) &
        \fbox{$\add \notin \underline \lay(\caps(\syscall))$} &
      }
    \]

    \[
      \Infer[SIE][Load-Step][\textsc{SLoad-Step}]
      {
        \lay \red \sframe{\frame{\cmemread[\lbl] \vx \expr\sep\cmd}{\regmap}{\opt}\cons\st}{\bmem}{\boolms}\cons\cfstack
        \sto{\dstep}{\omem \add} \sframe{\frame{\cmd}{\update \regmap x \val}{\opt} \cons \st}{\bmem}{\boolms}\cons\cfstack
      }
      {\toAdd{\sem\expr_{\regmap, \lay}} = \add &
        \bufread {\bm\buf\mem} \add 0 =\val, \bot &
        \add \in \underline \lay(\Ar[\opt]) &
        \fbox{$\opt = \km[\syscall] \Rightarrow \add \in \underline \lay(\caps(\syscall))$}
      }
    \]
    
    \[
      \Infer[SIE][If-Branch]{
      \lay \red
      \specconfone\cons\cfstack
      \sto{\dbranch[\lbl]{d}}{\obranch{d}}
      \sframe{\frame{\cmdtwo_{d}\sep\cmdtwo}{\regmap}{\opt}\cons\st}{\bm\buf\mem}{\boolms\lor d\neq\toBool{\sem \expr_{\regmap, \lay}}}
      \cons \specconfone \cons \cfstack
    }
    {\specconfone=\sframe {\frame{\cif[\lbl] \expr {\cmdtwo_\ctrue} {\cmdtwo_\cfalse}\sep\cmdtwo}{\regmap}{\opt}\cons\st}{\bm\buf\mem}{\boolms}}
  \]

  \[
    \Infer[SIE][Backtrack-Top][\textsc{Bt}_{\top}]{
      \lay \red \specconfone \cons\cfstack \sto{\dbt}{\obt \top} \cfstack
    }{
      \specconfone = \sframe{\st}{\bm\buf\mem}{\top} \lor \specconfone = \sconf{\err,\top}
    }
    \quad
    \Infer[SIE][Backtrack-Bot][\textsc{Bt}_{\bot}]{
      \lay \red \specconfone \cons\cfstack \sto{\dbt}{\obt \bot} \specconfone \cons \nil
    }{
      \specconfone = \sframe{\st}{\bm\buf\mem}{\bot} \lor \specconfone = \sconf{\err,\bot}
      & \cfstack \neq \nil
    }
  \]

    \[
      \Infer[SIE][Fence]{
        \lay \red
        \sframe{\frame{\cfence\sep\cmd}{\regmap}{\opt} \cons\st}{\bm\buf\mem}{\bot} \cons \cfstack
        \sto{\dstep}{\onone} \sframe{\frame{\cmd}{\regmap}{\opt}\cons\st}{\overline {\bm\buf\mem}}{\bot} \cons \cfstack
      }{}
    \]
  \end{framed}
  \caption{Speculative semantics, excerpt.}
  \label{fig:scen2specsemiexcerpt}
\end{figure*}
}\fi

A reduction step now takes the form
\[
  \lay \red \cfstack \sto{\dir}{\obs} \cfstack',
\]
indicating that for a given system $\system$, under layout $\lay \in \Lay$,
the system evolves from state $\cfstack$ with directive $\dir \in \Dir$ to $\cfstack'$ in one step,
producing the side-channel observation $\obs$.
The state of a system $\cfstack$ is now modeled as a stack of backtrackable configurations.
Specifically, configurations follow the following BNF:
\begin{cbnf}
  \specconfone, \specconftwo & \sframe{\st}{\bm{\buf}{\mem}}{\boolms} \bnfmid \sconf
  {\err, \boolms} \bnfmid \unsafe
\end{cbnf}
In a configuration $\sframe{\st}{\bm{\buf}{\mem}}{\boolms}$,
$\st$ is a call-stack as in \Cref{sec:language}, $\bm{\buf}{\mem}$ is a memory equipped with a \emph{write buffer} $\buf$, and $\boolms$ the \emph{mis-speculation} flag.
Buffered memories $\bm{\buf}{\mem}$ permit out-of-order, speculative memory operations.
Specifically, writing a value $\val$ at address $\add$  results in a delayed write $[\add \mapsto \val]\bm{\buf}{\mem}$, and $\bufread {\bm{\buf}{\mem}} \add k$ yields the
$k$th-last buffered entry $\val$ at address $\add$, together with a boolean flag $f$ that it $\bot$ if and only if $\val$ is the most recent one associated to address $\add$.
This operation is formally described by the following function:
\begin{align*}
  \bufread {\bm{[]}\mem} \add k &\defsym \mem(a), \bot &  &\\
  \bufread {\bm{\bitem \add \nat \cons \buf}\mem} \add 0 &\defsym \nat, \bot &  &\\
  \bufread {\bm{\bitem \add \nat \cons \buf}\mem} \add {i+1} &\defsym \nat', \top &&\text{if }\bufread {\bm\buf\mem} \add i = \nat', b \\
  \bufread {\bm{\bitem {\add'} \nat \cons \buf}\mem} \add {i} &\defsym \bufread {\bm\buf\mem} \add i &&\text{if }\add\neq\add'.
\end{align*}
In a configuration, the mis-speculation flag $\boolms$ records whether a past
step of computation led to a mis-speculation.
It is employed when backtracking from a speculative state.
As errors are recoverable under mis-speculation,
error configurations $\err$ carry also a mis-speculation flag. Finally,
as in \Cref{sec:language}, $\unsafe$ indicates
a safety violation.

Some illustrative rules of the semantics are given in \Cref{fig:scen2specsemiexcerpt},
\ifdefined\arxiv{the complete set of rules is relegated to the appendix, see~\Cref{fig:scen2sem1,fig:scen2sem1bis}.}\fi
The rules for a load instruction are very similar to the ones we give in
\Cref{sec:language}, but attackers can take advantage of
the store-to-load dependency speculation by issuing a $\dload[\lbl] i$
directive. When this happens, the $i$-th most recent value
associated to the address $\add = \toAdd{\sem \expr_{\regmap, \lay}}$
is retrieved from the buffered memory, see rule \ref{SIE:SLoad}.
This value may not correspond to that of the most recent
store to the address $\add$, and this is signaled by the flag $f$
that is returned after the buffer lookup. If $f=\top$,
this operation may be engaging mis-speculation and, for this reason,
the semantics keeps track of the current configuration in the stack.
A successful load produces the observation $\omem \add$
that leaks the  address to the attacker.
The rules for erroneous and unsafe loads are
\ref{SIE:SLoad-Error} and \ref{SIE:SLoad-Unsafe}
and they are analogous to their non-speculative counterparts.
In our semantics, every command supports the $\dstep$ directive,
which evaluates the configuration without speculating. For instance,
the \ref{SIE:Load-Step} rule evaluates the $\cmemread \vx \expr$ command
by fetching the most recent value from the write buffer, instead of an
arbitrary one.

Even branch instructions can be executed speculatively by issuing
the directive $\dbranch d$ by means of the rule \ref{SIE:If-Branch}.
This causes the evaluation to continue as if the guard resolved to
$d$. This operation leaks which branch is being executed
by means of the observation $\obranch d$. Even in this case,
the rule may be mis-speculating; for this reason, the current configuration
is book-kept in the stack and the mis-speculation flag is updated.

When the topmost configuration of a stack carries
the mis-speculation flag $\top$,
the configuration can be is discarded with the
rules \ref{SIE:Backtrack-Top}.
If it is $\bot$, the current
state is not mis-speculating, so the whole stack
of book-kept configurations can be discarded with
the rule  \ref{SIE:Backtrack-Bot}.

The \ref{SIE:Fence} rule commits all the entries in the
write buffer to the memory. Precisely, this operation is defined as follows:
\begin{align*}
  \overline {\bm{[]} \mem} &\defsym \mem &  \overline {\bm{\bitem \add \val : \buf}\mem} &\defsym \update{\overline {\bm{\buf}\mem}}\add \val,
\end{align*}
where, by $\update \mem \add\val$, we denote the memory
obtained by updating the value at address $\add$ with $\val$.
In particular, for consistency, a potentially mis-speculative state must
be resolved. This is why this rule requires the mis-speculation flag to
be $\bot$. This means that, if this configuration is reached
when the flag is $\top$, the semantics must backtrack
with the rule \ref{SIE:Backtrack-Top}.

\ifdefined\conference{\input{rules3}}\fi
\ifdefined\arxiv{\begin{figure*}
  \small
  \centering
  \begin{framed}
    \[
      \Infer[ALE][Poison]{
        \lay \red \aconf{\frame{\cpoison \dir\sep\speccmd}{\regmap}{\opt}\cons\st }{\mem}{\Ds}{\Os}
        \ato \aconf{\frame{\speccmd}{\regmap}{\opt}\cons\st}{\mem}{\dir:\Ds}{\Os}
      }{\strut}
    \]

    \[
      \Infer[ALE][Observe]{
        \lay \red \aconf{\frame{\vx \ass \eobs\sep\speccmd}{\regmap}{\opt}\cons\st }{\mem}{\Ds}{\obs:\Os}
        \ato \aconf{\frame{\speccmd}{\update \regmap\vx \obs}{\opt}\cons\st}{\mem}{\Ds}{\Os}
      }{}
    \]

    \[
      \Infer[ALE][Spec-Init]{
        \lay \red \aconf{\frame{\cspec\cmd\sep\speccmd}{\regmap}{\opt}\cons\st }{\mem}{\Ds}{\Os}
        \ato \hconf{\sframe{\frame{\cmd}{\regmap}{\opt}}{\mem}{\bot}}{\frame{\speccmd}{\regmap}{\opt}\cons\st}{\Ds}{\Os}
      }{}
    \]

    \[
      \Infer[ALE][Spec-D]{
        \lay \red \hconf{\cfstack}{\st}{\dir{\cons}\Ds}{\Os} \ato \hconf{\cfstack'}{\st}{\Ds}{\obs{\cons}\Os}
      }{
        \lay \red \cfstack \sto{\dir}{\obs} \cfstack'
      }
    \]

    \[
      \Infer[ALE][Spec-S]{
        \lay \red \hconf{\cfstack}{\st}{\Ds}{\Os} \ato \hconf{\cfstack'}{\st}{\Ds}{\obs{\cons}\Os}
      }{
        \nf \cfstack \Ds &
        \lay \red \cfstack \sto{\dstep}{\obs} \cfstack'
      }
      \
      \Infer[ALE][Spec-BT]{
        \lay \red \hconf{\cfstack}{\st}{\Ds}{\Os} \ato \hconf{\cfstack'}{\st}{\Ds}{\obs{\cons}\Os}
      }{
        \nf \cfstack \Ds &
        \nf \cfstack \dstep &
        \lay \red \cfstack \sto{\dbt}{\obs} \cfstack'
      }
    \]

    \[
      \Infer[ALE][Spec-Term]{
        \lay \red
        \hconf{\sframe{\frame{\cnil}{\regmap}{\opt}}{\bm\buf\mem}{\bot}}
              {\frame{\speccmd}{\regmap'}{\opt'}\cons\st}{\Ds}{\Os}
        \ato \aconf{\frame{\speccmd}{\regmap'}{\opt'}\cons\st}{\overline{\bm\buf\mem}}{\Ds}{\Os}
      }{}
    \]

  \[
    \Infer[ALE][Spec-Error]{
        \lay \red \hconf{\sconf{\err,\bot}} {\st}{\Ds}{\Os} \ato \err
    }{}
    \quad
    \Infer[ALE][Spec-Unsafe]{
      \lay \red \hconf{\unsafe} {\st}{\Ds}{\Os} \ato \unsafe
    }{}
  \]
  \end{framed}
  \caption{Semantics for speculative attackers, excerpt.}%
  \label{fig:scen2semiexcerpt}
\end{figure*}
}\fi


\subsubsection{Attacker's Language and Semantics}
\label{sec:adversarialsemantics}

To give a definition of kernel safety w.r.t.\ speculative semantics,
we endow an attacker with the ability to engage in speculative executions,
by issuing directives, and by the ability to read side-channel information.
To this end, we extend the instructions from \Cref{sec:language}
as follows:
\begin{cbnf}
  \SpInstr \ni \specstat & \dots &  \\
  & \cspec \cmd & speculation on victim $\cmd$\\
  & \cpoison \dir & speculative poisoning \\
  & \vx \ass \eobs & side-channel observation \\
  \SpCmd \ni \speccmd & \cnil \bnfmid \specstat\sep\speccmd
\end{cbnf}
The instruction $\cspec \cmd$ is used to execute a victim code $\cmd$,
w.r.t.\ the speculative semantics defined just above.
By using the instruction
$\cpoison \dir$, the attacker is able to to mistrain microarchitectural
predictors and to control the speculative execution.
Issued directives control the evaluation of victim code under speculative semantics.
Dual, the instruction $\vx \ass \eobs$ is used to extract side-channel info-leaks, collected
during speculative execution of the victim's code. To model this operation, in the following,
we assume $\Obs \subseteq\Val$.
As an example, the snippet
\begin{equation}
  \label{spec-prob}
  \cblock{
    \cpoison{\dbranch[\lbl]{\top}}\sep
    \cinline{\cspec{\cwhen[\lbl]{\expr}{\csyscall{\syscall}{\add}}}} \sep
    \vx \ass \eobs
  }
  \tag{\dag}
\end{equation}
forces the mis-speculative execution of $\csyscall{\syscall}{\add}$, independently of
the value of $\expr$. The register $\vx$ will hold the final observation leaked through executing the system call.

The attacker's semantics is defined in terms of a relation
\[
  \lay \red \specconfone \ato \specconfone' .
\]
In essence, the attacker executes under the standard semantics given in \Cref{sec:language},
the speculative semantics defined above play a role only when execution of the
victim is triggered by the directive $\cspec \cmd$.
Consequently, 
configurations are identical in structure to
the ones underlying the standard semantics, but carry however
additionally stacks $\Ds$ and $\Os$ of directives and observations,
in order to model the new constructs.
In addition, hybrid configurations $\hconf{\cfstack}{\st}{\Ds}{\Os}$ are used
to model the system when executing the victim under speculative semantics.
Here $\cfstack$ is a stack of speculative configurations concerning the victim,
and $\st$ the attacker's call stack
up to the invocation of speculation. Again, $\Ds$ gives the
directives (to be processed) and $\Os$ the observations (collected from executing victim's code).
In summary, configurations are drawn from the following BNF:
\begin{cbnf}
  \specconfone
  & \aconf{\st}{\mem}{\Ds}{\Os}
  \bnfmid \hconf{\cfstack}{\st}{\Ds}{\Os}
  \bnfmid \err
  \bnfmid \unsafe
\end{cbnf}

\Cref{fig:scen2semiexcerpt} shows the evaluation rules for the new constructs.
Rules~\ref{ALE:Poison} and~\ref{ALE:Observe} define the semantics for poisoning
and side-channel observations, by pushing and redacting elements of the corresponding stacks.
Rule~\ref{ALE:Spec-Init} deals with the initialization $\cspec \cmd$ of speculative execution,
starting from the corresponding initial configuration of the victim $\cmd$ in an empty speculation context.
A frame for the continuation of the attacker $\speccmd$ is pushed on the call stack $\cfstack$.
This frame is used to resume execution of the attacker, once the victim has been
fully evaluated.
The victim itself is evaluated
via the speculative semantics through rules~\ref{ALE:Spec-D}--\ref{ALE:Spec-BT}.
Note how execution of the victim is directed through the directive stack $\Ds$ (rule~\ref{ALE:Spec-D}).
Should the current directive be inapplicable, a non-speculative rewrite step (rule~\ref{ALE:Spec-S})
or backtracking (rule~\ref{ALE:Spec-BT}) is performed. Here,
the premise $\nf \cfstack \dir$ signifies that $\cfstack$ is irreducible
w.r.t.\ the directive $\dir$. Likewise, $\nf \cfstack \Ds$ means that
$\cfstack$ is irreducible w.r.t.\ the topmost directive of $\Ds$, or that $\Ds$ is empty. 
Note also how side-channel leakage, modeled through observations, is
collected in the configuration via these rules.
Upon normal termination, resuming of evaluation of the attacker is
governed by rule~\ref{ALE:Spec-Term} in the case of normal termination.
Finally, rules~\ref{ALE:Spec-Error} and~\ref{ALE:Spec-Unsafe} deal with abnormal termination.

We write $\sto{}{}^{*}$ for the multistep reduction relation induced by $\sto{}{}$, i.e,
$\cfstack \sto{\nil}{\nil} \cfstack$ and $\cfstack \sto{\dir : \Ds}{\obs : \Os}^{*} \cfstack'$ if
$\cfstack \sto{\dir}{\obs} \cdot \sto{\Ds}{\Os}^{*} \cfstack'$.

\subsection{Speculative Kernel Safety}
\label{subsec:specsafety}
We are now ready to extend the definition of kernel safety (\Cref{def:ks})
to the speculative semantics.
\begin{definition}[Speculative kernel safety]\label{def:sks}
  We say that a system $\system = (\rfs,\syss,\caps)$ is \emph{speculative kernel safe}
  if for every unprivileged attacker $\speccmd \in \SpCmd$,
  every layout $\lay$,
  and register map $\regmap$, we have: 
  \[
    \lnot \left(\lay \red \aconf{\frame{\speccmd}{\regmap}{\um}}{\lay \lcomp \rfs}{\nil}{\nil} \ato^* \unsafe \right).
  \]
\end{definition}
It is important to note that this safety notion captures violations that
occur during transient execution. This is in line with what happens, for instance,
for Spectre and Meltdown~\cite{Spectre,Meltdown}, both exploiting unsafe
memory access under transient execution in order to
reveal confidential information.

\subsection{The Demise of Layout Randomization in the Spectre Era}
\label{sec:breakingaslr}

A direct consequence of \Cref{def:sks} is that every system that is
\emph{speculative kernel safe} is also \emph{kernel safe}. 
The inverse, of course, does not hold in general. Most importantly,
the probabilistic form of safety provided by layout randomization
in \Cref{sec:safety1} does not scale to this extended threat model. 
This happens because \Cref{def:ks} does
not take side-channel leakage into account. As a simple example, gadgets like
\[
  \cif{\fn = \add}{\cmd}{\cmdtwo},
\]
can be exploited by an attacker to infer information about the address of a kernel-space
procedure $\fn$, through side-channel leakage distinguishing the execution of $\cmd$ and $\cmdtwo$.
In our model, this is reflected as executing this instruction allows the attacker to observe $\obranch \bool$, with $\bool$ being true precisely when $\fn$ resides at address $\add$.
Secondly, speculative execution undermines a fundamental premise crucially leveraged in \Cref{thm:scenario1} and, more widely, in the majority of studies demonstrating the efficacy of layout randomization as a defense against attacks (e.g.,~\cite{Abadi,Abadi2,Abadi3}): the notion that an unsuccessful memory probe leads to abnormal termination, thus thwarting the attack.
Indeed, within transient executions, memory access violations are recoverable.

This happens, for instance, if the system call $\syscall$
of \eqref{spec-prob} tries to load the content of the
address $\add$ form the memory to a register.
If $\add$ is not allocated, the system performs a
memory access violation under transient execution,
that does not terminate the execution. Conversely, if $\add$ is allocated,
its content is loaded into the cache, producing the observation $\omem \add$
before the execution of the branch and the system call are backtracked.
By reading side-channel observations, the attacker can thus
distinguish allocated kernel-addresses from those
that are not allocated.
This last example, in particular, is not at all
fictitious: the BlindSide attack~\cite{BlindSide}
uses the same idea to break Linux's KASLR
and locate the position of kernel's executable code
and data.

\subsection{Speculative Layout Non-Interference}
\label{sec:slni}

As this revised model significantly enhances the attackers' strength,
we will need to implement more stringent countermeasures in order to restore kernel safety.
To counter side-channel info-leaks, we can impose a form of
side-channel non-interference that is in line with the notion
of \emph{speculative constant-time} from~\cite{CTFundations}.

\begin{definition}[Speculative layout non-interference]\label{def:slni}
  Given $\system=(\rfs,\syss,\caps)$,
  a system call $\syscall$ is \emph{speculative layout non-interfering}, if,
  \[
    \lay_1 \red \sframe{\frame{\syss(\syscall)}{\regmap}{\km[\syscall]}}{\bm{\buf}{(\lay_1 \lcomp \store')}}{\boolms} \sto{\Ds}{\Os}^* \cfstack_1
  \]
  implies
  \[
    \lay_2 \red \sframe{\frame{\syss(\syscall)}{\regmap}{\km[\syscall]}}{\bm{\buf}{(\lay_2 \lcomp \store')}}{\boolms} \sto{\Ds}{\Os}^* \cfstack_2,
  \]
  for all layouts $\lay_{1},\lay_{2}$, configurations over stores $\store'\eqon\Fn\rfs$ coinciding on procedures ($\store'(\fn) = \rfs(\fn)$ for all $\fn \in \Fn$), directives $\Ds$, observations $\Os$ and register map $\regmap$.
\end{definition}
Speculative layout non-interference effectively prevents side-channel-related attacks, even during transient executions. Importantly, it ensures the non-leakage of layout information throughout the side-channels by requiring the identity of the two sequences of observations produced by the two reductions. This, however, implies severe restrictions on memory interactions --- effectively prohibiting the use of random memory layouts! Unsurprisingly, this form of non-interference directly establishes kernel safety of system calls:

\begin{lemma}
  \label{lemma:cttosafe}
  For every system $\system=(\rfs,\syss,\caps)$, if
  \[
    \kappa_{\km}\ge \sum_{\id \in \Idk} \size\id + 2\cdot\max_{\id \in \Idk} \size\id,
  \]
  and if $\syscall$ is speculative layout non-interfering, then
  \[
    \lnot \left(\lay \red \sframe{\frame{\syss(\syscall)}{\regmap}{\km[\syscall]}}{\bm{\buf}{(\lay \lcomp \store')}}{\boolms} \sto{\Ds}{\Os}^* \unsafe\right)
  \]
  for all layouts $\lay$ and initial configurations over stores $\store'$ coinciding with $\rfs$ on $\Fn$.
\end{lemma}
\ifdefined\arxiv{The proof of this lemma is in \Cref{sec:proofs2}.}\fi

Intuitively, this statement holds because if an invocation of a system call $\syscall$ 
performs an unsafe memory access when executing under a layout $\lay$, 
the address $\add$ of the accessed resource is leaked;
but the same address cannot leak if the resource is moved to another location ---
this is why we impose the condition on the size of the memory.
Thus, if a system call is speculative layout non-interfering,
it cannot be speculative non-interferent because
different memory layouts produce different observations.

In general, it is not always the case that a non-interference
property has as consequence memory safety.
For instance, being \emph{non-interferent} with respect to a set
of secrets does not prevent a victim program from breaking
memory safety. In our case, this property holds because the layouts
are not only used as the inputs for a computation,
but they also determine \emph{where} objects are placed in memory.

\begin{theorem}
  \label{thm:scenario2}
  Under the assumption $\kappa_{\km}\ge \sum_{\id \in \Idk} \size\id + 2\cdot\max_{\id \in \Idk} \size\id$,
  if a system $\system$ is speculative layout non-interfering, then it is \emph{speculative kernel safe}.
\end{theorem}
\ifdefined\arxiv{This result is demonstrated in  \Cref{sec:proofs2}.}\fi 

Observe that the safety guarantee
provided by \emph{speculative layout non-interference} is not probabilistic.
Although its effectiveness, layout randomization is unlikely to be restored
at the software level without imposing \emph{speculative layout non-interference},
in presence of this assumption, layout randomization is a redundant protection measure.
Also notice that
\emph{speculative layout non-interference} is not a necessary condition for
\emph{speculative kernel safety}. For instance, we can take in consideration
the following system:

\begin{code}[emph={s,f}]
  void f() { skip; };
  void s() { f(); };
\end{code}

This system does not enjoy \emph{speculative non-interference}, because by executing
\ex!s!, the address of \ex!f! leaks, and this address changes under different layouts.
However, this system is speculatively safe if we assume that \ex!f! belongs
to the capabilities of \ex!s!.


\section{Enforcement of Speculative Kernel Safety}
\label{sec:sksenforcement}
Although by requiring \emph{speculative layout non-interference},
we would be able to restore \emph{speculative kernel safety}, this would
impose important limitations on the system. For this reason,
we believe it is worth investigating whether \emph{speculative kernel safety}
can be enforced without imposing \emph{speculative layout non-interference}.

Nevertheless, directly enforcing \emph{speculative kernel safety}
is non-trivial, because it requires the developers to constantly
take in account a large variety of microarchitectural behaviors that
their system may run into. On the other hand, in the last decades, plenty of effort
has been put in developing safe code in the classic model,~\cite{kCFI,RustInLinux,
  towardkernelsafety,HyperSafe,FineGrainedkCFI,CompProtKer}.
So, our last question is to determine to which extent we can
establish a link between kernel safety and \emph{speculative}
kernel safety.
Following other works in this direction~\cite{ProSpeCT,Blade},
our main idea is to nullify the gap between \emph{kernel safety}
and \emph{speculative kernel} safety, by making the latter
property a consequence of the former. 

This can be achieved by finding a transformation
$\zeta$ that turns 
any \emph{kernel safe} system $\system$ into another system 
$\zeta(\system)$ which is architecturally equivalent to $\system$
but enjoys \emph{speculatively kernel safety}.
The semantic requirement on the transformation $\zeta$
is expressed by \Cref{def:sempres}, by which we ask that
no user-space program may show different behaviors by
executing in the two systems.

\begin{definition}[Semantics preservation]
  \label{def:sempres}
  A system transformation $\systrans$ is \emph{user-space semantics preserving}
  if, for any system $\system=(\rfs, \syss, \caps)$,
  \[
  \Eval[\systrans(\system)][\lay]{\cmd, \regmap, \um, \rfs'}
  \simeq
  \Eval[\system][\lay]{\cmd, \regmap, \um, \rfs}
  \]
  for every layout $\lay$, unprivileged command $\cmd$, and registers $\regmap$.
  Here, $\rfs'$ is the store underlying $\systrans(\system)$.
  The equivalence is given by $(v,\rfs[1]) \simeq (v,\rfs[2])$ if
  $\rfs[1] \eqon{\Idu} \rfs[2]$, and coincides with equality otherwise.
\end{definition}

Notice that in the previous definition we require $\rfs[1] \eqon{\Idu} \rfs[2]$ instead
of $\rfs[1] = \rfs[2]$ in order to allow the transformation $\systrans$
to modify kernel-space procedures.

Thanks to \emph{semantics preservation}, the second requirement on $\zeta$
can be fulfilled by asking that the system $\zeta(\system)$
can violate \emph{speculative kernel safety} only if it violates \emph{kernel safety}, as
captured by \Cref{def:cbu} below.

\begin{definition}
  \label{def:cbu}
  We say that $\zeta$  \emph{imposes speculative kernel safety}
  if, for every system $\system$ such that $\zeta(\system)=(\rfs, \syss, \caps)$,
  every buffer $\buf$ with $\dom(\buf)\subseteq\underline \lay(\Ar)$
  and store $\rfs' \eqon{\Fun} \rfs$,
  if
  \[
    \lay \red[\zeta(\system)] \conf{\conf{\syss(\syscall), \regmap, \km[\syscall]}, \bm\buf{(\lay \lcomp \rfs')}, \boolms} \sto{\Ds}{\Os}^* \unsafe,
  \]
  then
  \[
    \lay \red[\zeta(\system)] \conf{\conf{\syss(\syscall), \regmap, \km[\syscall]}, \overline{\bm\buf{(\lay \lcomp \rfs')}}}
    \to^* \unsafe.
  \]
\end{definition}

Observe that, by means of this property, we can easily
show that if a system $\zeta(\system)$ is \emph{kernel safe}, then
it is also \emph{speculative kernel safe}. Finally,
by combining \Cref{def:cbu} and \Cref{def:sempres},
we obtain the following conclusion:

\begin{proposition}
  \label{prop:mitigation}
  If a system $\system$ is \emph{kernel-safe},
  and the transformation $\zeta$ (i) \emph{imposes speculative kernel
    safety} and (ii) is \emph{user-space semantics preserving}, then
  (i) $\zeta(\system)$ is \emph{speculative kernel safe},
  and (ii) $\zeta(\system)$ is semantically equivalent to $\system$.
\end{proposition}

This result states that every \emph{kernel safe} system can be transformed
into another system that is equivalent to it
from the user's perspective and enjoys
stronger security guarantees.
Notice that \emph{kernel safety} cannot be provided
solely by the adoption of layout randomization:
by \Cref{thm:scenario1}, we know that layout randomization
provides \emph{kernel safety}
only modulo a small probability of failure.

Just as an example, we observe that a simple transformation that satisfies
the requirements of \Cref{prop:mitigation} 
can be implemented by placing a $\cfence$ instruction before
\emph{all} the potentially unsafe operations. This
instrumentation stops any ongoing speculation before
executing potentially unsafe operations, and prevents their transient
execution, yet leaving the program's semantics
unaltered at the architectural level. 
This \emph{fencing} transformation is
expressed by $\fencetrans : \Instr \to \cmd$
on the level of instructions where, in particular:
\begin{align*}
  \fencetrans( \cmemass \expr\exprtwo)&\defsym \cfence\sep\cmemass \expr\exprtwo\\
  \fencetrans( \cmemread \vx\expr)&\defsym \cfence\sep\cmemread \vx\expr\\ 
  \fencetrans( \ccall \expr {\exprtwo_1,\dots,\exprtwo_k})& \defsym \cfence\sep \ccall \expr {\exprtwo_1,\dots,\exprtwo_k}\\
  \fencetrans( \cwhile{\expr}{\cmd} )& \defsym \cwhile{\expr}{\fencetrans( \cmd) } \\
  \fencetrans( \cif{\expr}{\cmd}{\cmdtwo} )& \defsym \cif{\expr}{\fencetrans(\cmd)}{\fencetrans(\cmdtwo)},
\end{align*}
and it is the identity on the remaining instructions. Here,
the transformation is extended homomorphically to a
transformation $\fencetrans : \cmd \to \cmd$. It is lifted
to a system $\system$ by systematically applying it to system calls
and procedures $\rfs(\fn)$ in kernel-memory ($\fn \in \Fnk$).

Notice that the transformation $\fencetrans$, does not stop
completely speculation as, for instance,
speculation on conditional instructions is still allowed. This is
not in contrast with \Cref{def:cbu} because even in transient
execution, a conditional instruction cannot perform any safety
violation. However, as their branches can contain unsafe operations,
the transformation visits them.

By observing that $\fencetrans$ enjoys both the properties in
\Cref{def:sempres,def:cbu}, we can draw the following conclusion:

\begin{theorem}
  \label{thm:mitigation}
  If a system $\system$ is \emph{kernel-safe}, then
  $\eta(\system)$ is \emph{speculative kernel safe},
  and $\eta(\system)$ is semantically equivalent to $\system$.
\end{theorem}
\ifdefined\arxiv{This result is proven in \Cref{sec:proofs3}.}\fi

In addition to $\fencetrans$, other program
transformations that fit the requirements of
\Cref{prop:mitigation} can be identified:
for instance, the variation of $\fencetrans'$
that places a single $\cfence$ instruction before
sequences of loads --- or of stores --- is a good choice.
Similarly, $\cfence$ instructions can be omitted before
direct calls. Finally, in presence of an external proof
that shows that a system call $\syscall$ enjoys
\emph{speculative kernel safety}, the instrumentation
may decide to leave that system call unchanged, yet
preserving \Cref{def:cbu}.



\section{Related Work}
\label{sec:relatedwork}
\paragraph{On Layout Randomization. }
The first work that provided  a formal account of layout randomization was by \Citet{Abadi}, later extended in \cite{Abadi2,Abadi3}. In these works, the authors show that layout randomization prevents, with high probability, malicious programs from accessing the memory of a victim in an execution context with  shared address space. We have already discussed this in the body of the
paper how these results do not
model speculative execution or side-channel observations.


\paragraph{Spatial Memory Safety and Non-Interference} Spatial memory safety is typically defined by associating a software component with a memory area and requiring that, at runtime, it only accesses that area~\cite{HighAssurance,SoftBound,MSWasm}. \Citet{PierceMS} demonstrated that memory safety can be expressed in terms of non-interference; this property, in turn, stipulates that the final output of a computation is not influenced by secret data that a program must keep confidential~\cite{NonInt}. Both of these properties have been extended to the speculative model. The definition of \emph{speculative memory safety} from~\cite{HighAssurance} closely aligns with ours, while \emph{speculative non-interference} was initially introduced in the context of the \textsc{Spectector} symbolic analyzer~\cite{Spectector}. \textsc{Spectector}'s property captures information flows to side-channels that occur with speculative execution but not in sequential execution. In contrast to \textsc{Spectector}'s approach, our definition aligns with \emph{speculative constant-time}~\cite{CTFundations}, as it specifically targets information leaks that occur with speculative semantics.



\paragraph{Formal Analysis of Security Properties of Privileged Execution Environments. } \Citet{SLNonInt} deploy a model with side-channel leaks and privileged execution mode, without specualtive execution. In particular, they are interested in studying the preservation of constant-time in virtualization platforms. They also model privilege-raising procedures \emph{hypercalls}, similar to our system calls. They show that if one of the hosts is constant-time then the system enjoys a form of non-interference with respect to that host's secret memory. For this reason, although the two models are similar, the purposes of \Citet{SLNonInt} and our work are different: in~\cite{SLNonInt} the victim and the attacker have the same levels of privilege and the role of the hypervisor is to ensure their separation whilst, in our work, the privileged code base is itself the victim.  

\paragraph{Attacks to Kernel Layout Randomization} Attacks that aim at leaking information on the kernel's layout are very popular and can rely on implementation bugs that reveal information the kernel's layout~\cite{Uncontained,KernelBugs,KernelElasticObj} or on side-channel info-leaks~\cite{cacheKASLR,TagBleed, EntryBleed, kaslrfbr, kaslrfbr}. In particular attacks such as EchoLoad, TagBleed and EntryBleed~\cite{TagBleed, EntryBleed, EchoLoad} are successful even in presence of state-of-art mitigations such as Intel's Page Table Isolation (PTI)~\cite{PTI}. These attacks motivate our decision to take into account side-channel info-leaks. Due to address-space separation between kernel and user space programs, an attacker cannot easily use a pointer to a kernel address to access the victim's memory. So, in general, if the attacker does not control the value of a pointer that is used by the victim, this kind of leak is not harmful.

The Meltdown attack~\cite{Meltdown} uses speculative execution to overcome this limitation on operating systems running on Intel processors that do not adopt KAISER~\cite{Kaiser} or PTI~\cite{PTI}. In particular, the hardware can speculatively access an address before checking its permissions. The attack uses this small time window to access kernel memory content and leak it by using a side-channel info-leak gadget. These attacks can also be used to leak information on the layout: by dereferencing pointers under transient execution, the whole kernel's address space can be brute-forced without crashing the system. Due to the adoption of PTI~\cite{PTI}, this kind of attack is mitigated by removing most of the kernel-space addresses from the page tables of user-space programs. The BlindSide attack~\cite{BlindSide} overcomes this issue by probing directly from kernel-space. Similar attacks can be mounted by triggering different forms of mispredictions~\cite{BHInjection,SPEAR}. Arm's Pointer Authentication~\cite{ARMpa} is a technique that prevents forging pointers by extending them with an authentication code and raising an error if the code is violated. This can be used to deploy protections similar to layout randomization, but \Citet{PACMAN} showed that by leveraging speculative execution, it is possible to brute-force the authentication codes.

\paragraph{Relation Between Security in the Speculative- and Classic-model} 
Blade~\cite{Blade} is a protection mechanism which is aimed at preventing speculative data-flows by selectively stopping speculations. The authors show that, with this mechanism, all those program that are constant-time in the sequential model, are constant-time in the speculative model too. This is similar to what we do in \Cref{sec:sksenforcement}, by imposing \emph{speculative kernel safety} on a system that enjoys \emph{kernel safety}. \textsc{ProSpeCT}~\cite{ProSpeCT} is an open-source RISC-V processor that ensures a similar guarantee: each program that is constant-time in the classical model remains \emph{constant-time} even when executed on that processor. This protection relies on taint-tracking and requires explicit annotations on the security level of programs' data.  

\paragraph{Protections against Speculative Data Leaks} Commonly, speculative attacks are
aimed at leaking its victim's secret data~\cite{Spectre, Meltdown, ProSpeCT, Blade, Spectector, Contracts, HighAssurance, CTFundations}. As a consequence,
many of the conventional mitigations against speculative attacks are aimed at
preventing secret data from leaking during speculative execution. For instance,
Speculative Load Hardening~\cite{SLH} is a software protection measure which,
in its simplest form, sets each value that is loaded from memory during transient execution
to a constant value. By doing so, this mechanism does not prevent these value to be loaded --- i.e.
its application to the kernel would not prevent the attacker from breaking speculative kernel safety.
Together with the above-mentioned \textsc{ProSpeCT}, other hardware-level
taint-tracking based mechanisms have been deployed to prevent speculative leaks~\cite{STT,NDA}. These mechanisms limit the speculative execution of
load instructions with different levels of strictness, ranging from completely preventing
the execution of these instructions (\cite{NDA}, strict propagation and load restriction mode),
from just prohibiting the propagation of the loaded value~\cite{STT}.
Although this approach is promising, the above-mentioned mechanisms do
not impose limitations on the speculative execution of indirect branches that may be used
by attackers in practice to break speculative kernel safety. 







%

\section{Conclusion}
\label{sec:conclusion}
We have formally demonstrated that kernel's layout randomization probabilistically ensures kernel safety for a classic model, where
an attacker cannot compromise the 
system via speculative execution or side-channels, and 
 users of an operating system execute without privileges, but 
victims can feature pointer arithmetic, introspection, and indirect jumps.

We have also shown that the protection offered by layout randomization does not naturally
scale against attackers that can control speculative execution and side-channels, and stipulate a sufficient condition 
to enforce kernel safety in the Spectre era.
We also propose mechanisms based on program transformations that provably enforce
 {speculative kernel safety} on a system, provided that this system
already enjoys {kernel safety} in the classic model.  
To the best of our knowledge, our work is the first to formally investigate and provide ways to achieve kernel safety in the presence of speculative and side-channel vulnerabilities. 
 
This work prepares the ground for future developments such as
modeling more expressive attacker models, e.g.
attackers speculating on the branch target buffer, related to Spectre v2~\cite{Spectre},
optimizing the instrumentation we presented in \Cref{sec:sksenforcement},
and assessing its overhead on a real operating system.

An orthogonal research direction is to study more fine-grained
safety properties instead of \emph{(speculative) kernel safety} ---
e.g. by distinguishing violations of CFI from violations
of spatial safety, and data integrity from confidentiality, akin to what happens in~\cite{FormalizingSafety}. In this direction, it would also be interesting to model
a call stack in order to determine whether the safety of kernel's
stack can be granted under other, possibly weaker, conditions.  

Finally, a valuable future development is to extend our execution model with other features
that are often used to undermine operating systems' security, such as dynamic memory
allocation and dynamic module loading, with the aim to establish their impact on
system's safety in presence of speculative attackers.


\section*{Acknowledgments}
\label{sec:acks}

We are grateful to Gilles Barthe, M\'arton Bogn\'ar, Ugo Dal Lago, Lesly-Ann Daniel, Benjamin Gr\'egoire, Jean-Pierre Lozi, and Frank Piessens for their comments on an early draft of the paper. This work was partially supported through the projects PPS ANR-19-C48-0014 and UCA DS4H ANR-17-EURE-0004, and by the Wallenberg AI, Autonomous Systems and Software Program (WASP) funded by the Knut and Alice Wallenberg Foundation.
  
\bibliographystyle{plainnat}
\bibliography{bibliography}

\pagebreak

\appendix
\section{Appendix}
\label{sec:appendix}
\subsection{Appendix for \Cref{sec:safety1}}
\label{sec:appsafety1}

\subsubsection{Omitted Proofs and Results}
\label{sec:proofs1}

We begin by introducing additional notation that will be used throughout this section. We then present the proof of \Cref{thm:scenario1}, followed by the proofs of the intermediate results upon which this proof relies.

\paragraph*{Semantics' Notation}
We write $\update{{\rfs}}{(\ar, i)} \val$ to denote the store that is identical to $\rfs$ except for the array $\ar$, which remains pointwise equal to $\rfs(\ar)$ except at index $i$, where it is updated to $\val$.

We use the notation $\nf \confone{}$ when $\confone$ does not reduce
in system $\system$ and using layout $\lay$.
$\nstep{!} \confone \conftwo$ to represent the formula
$\exists \nat'. \nstep{\nat'} \confone \conftwo \land \nf \conftwo {}$. Similarly, we use $\nstep{!\nat} \confone \conftwo$ to
denote the formula
$\exists \nat' \leq \nat. \nstep{\nat'} \confone \conftwo \land
\nf \conftwo {}$.

Additionally, we generalize the function $\Eval\cdot$ to take full configurations as arguments. This allows us to use it in the following way:
\[
  \Eval {\conf{\st, \lay\lcomp\rfs}} \defsym
  \begin{cases}
    \Omega
    & \text{if $\diverge{\conf{\st, \lay\lcomp\rfs}}$},\\
    (\val,\store')
    & \text{if } \lay \red \conf{\st, \lay\lcomp\rfs}  \to^* \conf{\frame{\cnil}{\regmap[\ret \mapsto \val]}{\opt}, \lay \lcomp \store'}, \\
    \err
    & \text{if } \lay \red \conf{\st, \lay\lcomp\rfs} \to^* \err, \\
    \unsafe
    & \text{if } \lay \red \conf{\st, \lay\lcomp\rfs} \to^* \unsafe.
  \end{cases}
\]
\noindent
We also extend the relation $\evaleq$ by stating that:
\[
(\val,\store) \evaleq \conf{\frame{\cnil}{\regmap[\ret \mapsto \val]}{\opt}, \lay \lcomp \store}
\]
holds for every layout $\lay$, register map $\regmap$, and every flag $\opt$. Furthermore, we take its symmetric and transitive closure.

\paragraph*{User- and Kernel-mode stack}
We write $\um(\cmd)$ as a shorthand for $\ids(\cmd) \subseteq \Idu$. The predicate  $\ids(\st)$ is defined inductively as follows:
\[
  \um(\nil)\defsym \top \quad \um(\fr:\st)\defsym \um(\fr) \land \um(\st)\quad \um(\frame{\cmd}{\regmap}{\opt}) \defsym \um(\cmd) \land \opt=\um. 
\]
The predicate $\km[\syscall]$ is defined analogously:
\[
  \km[\syscall](\nil)\defsym \top \quad \km[\syscall](\fr:\st)\defsym \km[\syscalltwo](\fr) \land \km[\syscall](\st)\quad \km[\syscall](\frame{\cmd}{\regmap}{\opt}) \defsym \km(\cmd) \land \opt=\km[\syscall]. 
\]
We write $\km(\cdot)$ as a shorthand for $\exists \syscall \in \Sys. \km[\syscall](\cdot)$.

\begin{proof}[Proof of \Cref{thm:scenario1}]
  Let $\system=(\rfs, \syss, \caps)$ be a system. We aim to show that for every \emph{unprivileged} command $\cmd$,
register map $\regmap$, and layout distribution $\mu$, the following holds:
\[
    \Pr_{\lay \leftarrow\mu}\left[\Eval {\conf{\frame{\cmd}{\regmap}{\um}, \lay \lcomp \rfs}} = \unsafe \right]=
      \Pr_{\lay \leftarrow\mu}\left[\nstep! {\conf{\frame{\cmd}{\regmap}{\um}, \lay \lcomp \rfs}} \unsafe \right] \le
    1-\delta_\mu.
  \]
  If the probability is not 0, by \Cref{lemma:premiseApreserv} (that we can apply because all the system calls in $\system$ are \emph{layout non-interferent}, and because all the layouts are identical with respect to user-space identifiers), we deduce that:
  \[
    \forall \lay. \Eval[\system][\lay]{\conf{\frame{\cmd}{\regmap}{\um}, \lay \lcomp \rfs}} \evaleq \unsafe.
  \]
  This proposition can be rewritten as:
  \[
    \forall \lay. \Eval[\system][\lay]{\conf{\frame{\cmd}{\regmap}{\um}, \lay \lcomp \rfs}} \in \{\unsafe, \err\},
  \]
  which allows us to deduce that:
  \[
    \forall \lay. \exists \conftwo_\lay, \nat_\lay. \nstep{!\nat_\lay} {\conf{\frame{\cmd}{\regmap}{\um}, \lay \lcomp \rfs}} {\conftwo_\lay}.
  \]
  Since memory size is finite, also $\Lay$ is finite, so there is a natural number $\overline \nat$ such that $\overline \nat \ge \max_{\lay \in \Lay}{\nat_\lay}$.
  Thus, we obtain:
  \[
    \forall \lay. \exists \conftwo_\lay. \nstep{!\overline \nat} {\conf{\frame{\cmd}{\regmap}{\um}, \lay \lcomp \rfs}} {\conftwo_\lay}.
  \]
  Applying \Cref{lemma:mainlemmapms} to the system $\system$, the
  bound $\overline \nat$, the unprivileged attacker $\cmd$, the
  initial register map $\regmap$, the distribution $\mu$, and to the last
  intermediate claim, we deduce that one of the following statements
  holds:
  \begin{varitemize}
  \item There are a register map $\regmap'$ and a store $\rfs'\eqon{\Fun}\rfs$
    such that for every layout $\lay$:
    \[
      \nstep {!\overline \nat}{\conf{\frame{\cmd}{\regmap}{\um}, \lay \lcomp \rfs}}{\conf{\frame{\cnil}{\regmap}{\um}, \lay \lcomp \rfs'}},
    \]
  \item
    $
      \Pr_{\lay \leftarrow \mu}\left[\nstep {!\overline \nat} {\conf{\frame{\cmd}{\regmap}{\um}, \lay \lcomp \rfs}}\err \right] \ge \delta_\mu.
    $
  \end{varitemize}
  We go by cases on these two statements. If the first one holds,
  then for every layout $\lay$: 
  \[
    \nstep {!}{\conf{\frame{\cmd}{\regmap}{\um}, \lay \lcomp \rfs}}{\conf{\frame{\cnil}{\regmap'}{\um}, \lay \lcomp \rfs'}}.
  \]
  Consequently, we have
  \[
    \Pr_{\lay \leftarrow \mu}\left[    \nstep {!}{\conf{\frame{\cmd}{\regmap}{\um}, \lay \lcomp \rfs}}{\conf{\frame{\cnil}{\regmap'}{\um}, \lay \lcomp \rfs'}} \right] = 1.
  \]
  Since the final configuration cannot be contemporary $\unsafe$
  and ${\conf{\frame{\cnil}{\regmap'}{\um}, \lay \lcomp \rfs'}}$, we conclude:
  \begin{equation*}
    \Pr_{\lay \leftarrow \mu}\left[    \nstep {!}{\conf{\frame{\cmd}{\regmap}{\um}, \lay \lcomp \rfs}}{\unsafe} \right] = 0.
  \end{equation*}
  This established the claim. If the second proposition holds then:
  \begin{equation*}
    \Pr_{\lay \leftarrow \mu}\left[\nstep {!} {\conf{\frame{\cmd}{\regmap}{\um}, \lay \lcomp \rfs}}\err \right] \ge\\
    \Pr_{\lay \leftarrow \mu}\left[\nstep {!\overline \nat} {\conf{\frame{\cmd}{\regmap}{\um}, \lay \lcomp \rfs}}\err \right] \ge \delta_\mu,
  \end{equation*}
  and therefore we conclude that:
  \begin{equation*}
    \Pr_{\lay \leftarrow \mu}\left[\nstep {!} {\conf{\frame{\cmd}{\regmap}{\um}, \lay \lcomp \rfs}}\unsafe \right] \le\\
    1- \Pr_{\lay \leftarrow \mu}\left[\nstep {!} {\conf{\frame{\cmd}{\regmap}{\um}, \lay \lcomp \rfs}}\err \right] \le 1-\delta_\mu.
  \end{equation*}
\end{proof}

\begin{lemma}[Preservation of \emph{layout non-interference}]
  \label{lemma:premiseApreserv}
  If every system call in $\system$ is \emph{layout non-interferent},
  then the entire system is also \emph{layout non-interferent} with
  respect to \emph{unprivileged} attackers, in the following sense:
  for every \emph{unprivileged} attacker $\cmd$, register map
  $\regmap$, pair of layouts $\lay_1, \lay_2$, and configuration
  $\confone_1$ we have that:
  \[
    \Eval[\system][\lay_1]{\conf{\frame{\cmd}{\regmap}{\um}, \lay_1\lcomp \rfs}} \evaleq \Eval[\system][\lay_2]{\conf{\frame{\cmd}{\regmap}{\um}, \lay_2\lcomp \rfs}}
  \]  
\end{lemma}
\begin{proof}
  The main claim is a consequence of an auxiliary claim, namely:
    \begin{equation*}
    \nstep[\system][\lay_1]!{\conf{\frame{\cmd}{\regmap}{\um}, \lay_1\lcomp \rfs}}{\confone_1} \Rightarrow
    \Eval[\system][\lay_2]{\conf{\frame{\cmd}{\regmap}{\um}, \lay_2\lcomp \rfs}} \evaleq \confone_1
    \tag{C}
  \end{equation*}
  To prove (C), we establish a slightly stronger statement: instead of
  quantifying over $\cmd$ and $\regmap$, we generalize to a non-empty
  stack $\st$ such that $\um(\st)$, and we rewrite the claim in a more
  convenient shape, where all the meta-variables that are not
  quantified explicitly are quantified universally.
  \begin{equation*}
    \forall \nat. \forall \rfs'\eqon{\Fun} \rfs. \left(\nstep[\system][\lay_1]\nat{\conf{\st, \lay_1\lcomp \rfs'}}{\confone_1} \land \final[\system][\lay_1] {\confone_1}\right) \Rightarrow\\
    \Eval[\system][\lay_2]{\conf{\st, \lay_2\lcomp \rfs'}} \evaleq \confone_1
    \tag{C'}
  \end{equation*}
  Once (C') is established, we derive (C) by instantiating $\rfs'$ as
  $\rfs$ and $\st$ as
  ${\conf{\frame{\cmd}{\regmap}{\um}, \lay_1\lcomp \rfs}}$. We fix
  $\nat$ as the number of steps taken by the reduction from the
  initial configuration to the terminal configuration
  $\confone_1$. From (C), we can deduce the main claim by case
  analysis:
  \begin{proofcases}
    \proofcase{$\Eval[\system][\lay_1]{\conf{\frame{\cmd}{\regmap}{\um},
          \lay_1\lcomp \rfs}} = \Omega$} In this case, the value of
    $\Eval[\system][\lay_2]{\conf{\frame{\cmd}{\regmap}{\um},
        \lay_2\lcomp \rfs}}$ must also be $\Omega$. If this was not
    the case, we could apply (C) with $\lay_1$ and $\lay_2$ swapped to
    prove that
    $\Eval[\system][\lay_1]{\conf{\frame{\cmd}{\regmap}{\um},
        \lay_1\lcomp \rfs}}\neq \Omega$.
    \proofcase{$\Eval[\system][\lay_1]{\conf{\frame{\cmd}{\regmap}{\um},
          \lay_1\lcomp \rfs}} \neq \Omega$} In this case, the
    conclusion is a direct consequence of (C).
  \end{proofcases}

  Now that we showed how the main claim can be deduced from (C'), we
  can focus on the proof of (C').  In the proof we will extensively
  use the assumption on the layouts:
  \[
    \forall \id \in \Idu.\lay_1(\id)=\lay_2(\id),
    \tag{\dag}
  \]
  For this reason, we fix it on top and name it (\dag). The proof goes by induction on $\nat$.
  \begin{proofcases}
    \proofcase{0} We assume
    \[
      \nstep[\system][\lay_1]0{\conf{\st, \lay_1\lcomp \rfs'}}{\confone_1} \land \final[\system][\lay_1] {\confone_1}.
    \]
    This implies that $\confone_1 = \conf{\st, \lay_1\lcomp
      \rfs'}$. From this last observation, and since $\confone_1$ is
    terminal, we deduce that $\st = \frame{\cnil}{\regmap'}{\um}$.  By
    examining the semantics, we also establish that
    \[
      \final[\system][\lay_2]{\conf{\st, \lay_2\lcomp \rfs'}},
    \]
    so we conclude that
    $\Eval[\system][\lay_2]{\conf{\st, \lay_2\lcomp \rfs'}} =
    (\regmap'(\ret),\rfs') \evaleq \confone_1$.
    \proofcase{$\nat+1$}
    Observe that $\st \neq \nil$, as it would contradict the assumption
    \[
      \nstep[\system][\lay_1]{\nat+1}{\conf{\st, \lay_1\lcomp \rfs'}}{\confone_1} \land \final[\system][\lay_1] {\confone_1}.
      \tag{H}
    \]
    Thus, we assume $\st = \frame \cmd \regmap \um :\st'$, and proceed by case analysis on $\cmd$.  
    Note that the induction hypothesis coincides syntactically with (C').
    \begin{proofcases}
      \proofcase{$\cnil$} In this case, we rewrite the first part of (H) as follows:
      \begin{equation*}
        \step[\system][\lay_1] {\conf{\frame{\cnil}{ \regmap}{\um}:\st', \lay_1\lcomp \rfs'}}{}
        \conf{\st'', \lay_1\lcomp \rfs'}\to^\nat{\confone_1},
      \end{equation*}
      where $\st''$ is obtained by updating the topmost register map
      in $\st'$ with the return value from $\regmap$.  Since
      $\um(\st')$ holds, also $\um(st'')$ does, so we can apply the IH
      and conclude that
      $ \Eval[\system][\lay_2]{\conf{\st'', \lay_2\lcomp \rfs'}}
      \evaleq \confone_1 $.  Thus, to conclude the proof, it suffices
      to observe that
      $ \step[\system][\lay_2]{\conf{\frame{\cnil}{\regmap}{\um}:\st',
          \lay_2\lcomp \rfs'}}{} {\conf{\st'', \lay_2\lcomp \rfs'}}, $
      which follows by introspection of the semantics.
      \proofcase{$\vx \ass\expr\sep \cmdtwo$} In this case, we rewrite
      the first part of (H) as follows:
      \begin{equation*}
        \step[\system][\lay_1] {\conf{\frame{\vx \ass\expr\sep\cmdtwo}{ \regmap}{\um}:\st', \lay_1\lcomp \rfs'}}{}\\
        {\conf{\frame{\cmdtwo}{\update \regmap\vx {\sem \expr_{\regmap, \lay_1}}}{\um}:\st', \lay_1\lcomp \rfs'}}\to^\nat{\confone_1}
      \end{equation*}
      observe that
      $\um(\frame{\cmdtwo}{\update \regmap\vx {\sem \expr_{\regmap,
            \lay_1}}}{\um}:\st')$ holds, so we can apply the IH and
      conclude that
      \[
        \Eval[\system][\lay_2]{\conf{\frame{\cmdtwo}{\update \regmap\vx {\sem \expr_{\regmap, \lay_1}}, }{\um}:\st', \lay_2\lcomp \rfs'}} \evaleq \confone_1
      \]
      This means that, in order to conclude the proof, it suffices to observe that
      \begin{equation*}
        \step[\system][\lay_2]{\conf{\frame{\vx \ass\expr\sep\cmdtwo}{\regmap}{\um}:\st', \lay_2\lcomp \rfs'}}{}\\
        {\conf{\frame{\cmdtwo}{\update \regmap\vx {\sem \expr_{\regmap, \lay_1}}}{\um}:\st', \lay_2\lcomp \rfs'}},
      \end{equation*}
      which, in turn, reduces to showing that
      $\sem \expr_{\regmap, \lay_1}=\sem \expr_{\regmap, \lay_2}$. By
      definition of $\um(\st)$, we have that
      $\vx \ass\expr\sep\cmdtwo$ is an unprivileged command, so, in
      particular all identifiers in $\expr$ belong to $\Idu$, hence,
      the $\sem \expr_{\regmap, \lay_1}=\sem \expr_{\regmap, \lay_2}$
      follows from \Cref{rem:expreval}, and (\dag).
      \proofcase{$\cskip\sep \cmdtwo$} Analogous to the case of
      assignments.
      \proofcase{$\cif \expr {\cmd_\bot}{\cmd_\top}\sep\cmdtwo$}
      Analogous to the case of assignments.
      \proofcase{$\cwhile \expr {\cmd}\sep\cmdtwo$} Analogous to the
      case of assignments.
      \proofcase{$\cmemass \expr \exprtwo\sep \cmdtwo$} In this case,
      we start by observing that
      $\sem \expr_{\regmap, \lay_1}=\sem \expr_{\regmap, \lay_2}$ by
      \Cref{rem:expreval}. This also implies
      $\toAdd{\sem \expr_{\regmap, \lay_1}}=\toAdd{\sem
        \expr_{\regmap, \lay_2}}$. Let
      $\toAdd{\sem \expr_{\regmap, \lay_1}}= \add$. By
      \Cref{rem:expreval}, we also know that
      $\sem \exprtwo_{\regmap, \lay_1}=\sem \exprtwo_{\regmap,
        \lay_2}$, and we denote this value $\val$.
      From the assumption (\dag) we deduce that
      $\underline \lay_1(\Ar[\um])=\underline \lay_2(\Ar[\um])$, and
      we denote this set $P$.  We go by cases on $\add \in P$.
      \begin{proofcases}
        \proofcase{$\add \in P$} In this case, the rule \ref{WL:Store}
        applies to both the configurations
        \[
        {\conf{\frame{\cmemass  \expr \exprtwo\sep \cmdtwo}{\regmap}{\um}:\st', \lay_1\lcomp \rfs'}}
        \]
        and
        \[
        {\conf{\frame{\cmemass  \expr \exprtwo\sep \cmdtwo}{\regmap}{\um}:\st', \lay_2\lcomp \rfs'}}
        \]
        under the layouts $\lay_1$ and $\lay_2$ respectively, obtaining:
        \[
          {\conf{\frame{\cmdtwo}{\regmap}{\um}:\st', \update{\lay_1\lcomp \rfs'}{\add}{\val}}}
        \]
        and
        \[
          {\conf{\frame{\cmdtwo}{\regmap}{\um}:\st', \update{\lay_2\lcomp \rfs'}{\add}{\val}}}.
        \]
        Since $\add \in \underline \lay_1(\Ar[\um])$
        and by definition of $\underline \lay_1$, we deduce that
        there exist $\ar \in \Ar[\um]$
        and $0\le i\le\size\ar$ such that $\lay_1(\ar)+i = \add$.
        By (\dag), we also deduce that $\lay_2(\ar)+i = \add$,
        so we can apply \Cref{rem:memupdtostupd}
        in order to show that:
        \[
          \update{\lay_1\lcomp \rfs'}{\add}{\val} =
          \lay_1\lcomp \update{\rfs'}{(\ar, i)}{\val}
        \]
        and that 
        \[
          \update{\lay_2\lcomp \rfs'}{\add}{\val} =
          \lay_1\lcomp \update{\rfs'}{(\ar, i)}{\val}.
        \]
        Finally, since $\um({\frame{\cmdtwo}{\regmap}{\um}:\st'})$ holds
        and we have:
        \[
          \update{\rfs'}{(\ar, i)}\val\eqon{\Fn}{\rfs'}\eqon{\Fn}\rfs,
        \]
        we can apply the IH, and conclude the proof of this
        sub-derivation.  \proofcase{$\add \notin P$} In this case,
        from (\dag), we deduce that
        $\add \notin \underline \lay_2(\Ar[\um])$.  Therefore, the
        rule \ref{WL:Store-Error} applies, showing both:
        \[
          \step[\system][\lay_1]
          {\conf{\frame{\cmemass  \expr \exprtwo\sep \cmdtwo}{\regmap}{\um}:\st', \lay_1\lcomp \rfs'}}
          \err,
        \]
        and
        \[
          \step[\system][\lay_2]
          {\conf{\frame{\cmemass  \expr \exprtwo\sep \cmdtwo}{\regmap}{\um}:\st', \lay_2\lcomp \rfs'}}
          \err.
        \]
        This establishes the claim. 
      \end{proofcases}
      \proofcase{$\cmemread \vx \expr\sep \cmdtwo$} Analogous to the
      case of memory store operations.
      \proofcase{$\ccall \expr {\exprtwo_1, \dots, \exprtwo_k}\sep
        \cmdtwo$} This case is also similar to that of memory loads,
      but requires some non-trivial observations.  As before, we
      define $\add = \toAdd{\sem \expr_{\regmap,\lay_1}}$ and observe
      that $\add = \toAdd{\sem \expr_{\regmap,\lay_2}}$ because of
      \Cref{rem:expreval}. Similarly, we introduce the values
      $\val_1, \dots,\val_k$ which correspond to the semantics of
      $\exprtwo_1,\dots,\exprtwo_k$ evaluated under $\regmap$ and both
      layouts $\lay_1, \lay_2$. By (\dag), we have
      $\lay_1(\Fn[\um]) = \lay_2(\Fn[\um])$, so we denote this set
      $P_{\Fn[\um]}$, and we go by case analysis on
      $\add \in P_{\Fn[\um]}$.
      \begin{proofcases}
        \proofcase{$\add\in P_{\Fn[\um]}$} In this case, there exists
        $\fn \in \Fn[\um]$ such that $\lay_1(\fn)=\add$, and from
        (\dag) we deduce $\add=\lay_2(\fn)$. Using these intermediate
        conclusions and the definition of $\lcomp$, we deduce that
        $(\lay_1\lcomp \rfs')(\add) = \rfs'(\fn) = \rfs(\fn)$, (the
        last step follows from the assumption $\rfs'\eqon{\Fn}\rfs$),
        and similarly for $\lay_1\lcomp \rfs'(\add)$. Since
        $\add\in P_{\Fn[\um]}$, we deduce that rule \ref{WL:Call} can
        be applied, proving:
        \begin{equation*}
          \step[\system][\lay_1]
          {\conf{\frame{\ccall  \expr {\exprtwo_1, \dots, \exprtwo_k}\sep \cmdtwo}{\regmap}{\um}:\st', \lay_1\lcomp \rfs'}}{}\\
          {\conf{\frame{\rfs(\fn)}{\regmap_0'}{\um}:\frame{\cmdtwo}{\regmap}{\um}:\st', \lay_1\lcomp \rfs'}},
        \end{equation*}
        and
        \begin{equation*}
          \step[\system][\lay_2]
          {\conf{\frame{\ccall  \expr {\exprtwo_1, \dots, \exprtwo_k}\sep \cmdtwo}{\regmap}{\um}:\st', \lay_2\lcomp \rfs'}}{}\\
          {\conf{\frame{\rfs(\fn)}{\regmap_0'}{\um}:\frame{\cmdtwo}{\regmap}{\um}:\st', \lay_2\lcomp \rfs'}},
        \end{equation*}
        where $\regmap_0'= \regmap_0[\vx_1,\dots,\vx_k\upd \val_1,\dots,\val_k]$.
        Finally, by assumption on $\rfs$, we know that $\rfs(\fn)$ is an unprivileged program. So, in particular,
        \[
          \um({\frame{\rfs(\fn)}{\regmap_0'}{\um}:\frame{\cmdtwo}{\regmap}{\um}:\st'})
        \]
        holds. This allows us to apply the IH and conclude the proof.
        \proofcase{$p \notin P_{\Fn[\um]}$} Analogous to the corresponding case for store operations.
      \end{proofcases}
      \proofcase{$\csyscall \syscall{\expr_1, \ldots,
          \expr_k}\sep\cmdtwo$} This is perhaps the most interesting
      case of this proof. We begin by introducing the values
      $\val_1, \dots,\val_k$, which correspond to the
      semantics of $\exprtwo_1,\dots,\exprtwo_k$ evaluated under
      $\regmap$ as well as any of layouts $\lay_1, \lay_2$.These
      values do not depend on the layout because of \Cref{rem:expreval}. By
      introspection of the rule \ref{WL:SystemCall}, we deduce that
      both the following statements hold:
      \begin{equation*}
          \step[\system][\lay_1]
          {\conf{\frame{\csyscall  \syscall {\exprtwo_1, \dots, \exprtwo_k}\sep \cmdtwo}{\regmap}{\um}:\st', \lay_1\lcomp \rfs'}}{}\\
          {\conf{\frame{\syss(\syscall)}{\regmap_0'}{\km[\syscall]}:\frame{\cmdtwo}{\regmap}{\um}:\st', \lay_1\lcomp \rfs'}},
        \end{equation*}
        and
        \begin{equation*}
          \step[\system][\lay_2]
          {\conf{\frame{\csyscall \syscall {\exprtwo_1, \dots, \exprtwo_k}\sep \cmdtwo}{\regmap}{\um}:\st', \lay_2\lcomp \rfs'}}{}\\
          {\conf{\frame{\syss(\syscall)}{\regmap_0'}{\km[\syscall]}:\frame{\cmdtwo}{\regmap}{\um}:\st', \lay_2\lcomp \rfs'}},
        \end{equation*}
        where
        $\regmap_0'= \regmap_0[\vx_1,\dots,\vx_k\upd
        \val_1,\dots,\val_k]$. Unfortunately, we cannot apply the IH
        to the configuration
        \[
          {\conf{\frame{\syss(\syscall)}{\regmap_0'}{\km[\syscall]}:\frame{\cmdtwo}{\regmap}{\um}:\st', \lay_1\lcomp \rfs'}}
        \]
        because  $\um(\frame{\syss(\syscall)}{\regmap_0'}{\km[\syscall]}:\frame{\cmdtwo}{\regmap}{\um}:\st')$ does not hold. However, we can go by cases on
        \[
          \Eval[\system][\lay_1]{{{\syss(\syscall)},{\regmap_0'},{\km[\syscall]},\rfs'}}
        \]
        \begin{proofcases}
          \proofcase{$\Omega$} By \Cref{lemma:contextplugging}, we deduce that 
        \[
          \Eval[\system][\lay_1]{\conf{\frame{\syss(\syscall)}{\regmap_0'}{\km[\syscall]}:\frame{\cmdtwo}{\regmap}{\um}:\st', \lay_1\lcomp \rfs'}} = \Omega,
        \]
        which leads to a contradiction because we assumed
        \small
        \[
          \nstep[\system][\lay_1] {\nat+1} {\conf{\frame{\syss(\syscall)}{\regmap_0'}{\km[\syscall]}:\frame{\cmdtwo}{\regmap}{\um}:\st', \lay_1\lcomp \rfs'}}{\confone_1}\land \nf[\system][\lay_1] {\confone_1}.
        \]
        \normalsize
        \proofcase{$\unsafe$} From the \emph{layout non-intreference property} (\Cref{def:lni}), we deduce that
        \[
          \Eval[\system][\lay_2]{{\syss(\syscall), \regmap_0', \km[\syscall], \rfs'}} \in\{\err,\unsafe\}.
        \]
        We take $\err$ as an example. This means that 
        \[
          \exists h.\nstep[\system][\lay_2] h {\conf{\frame{\syss(\syscall)}{\regmap_0'}{\km[\syscall]}, \lay_2\lcomp \rfs'}}{\err}.
        \]
        By applying \Cref{lemma:contextplugging} to the reduction with $\lay_1$,
        we deduce that
        \[
          \nstep[\system][\lay_1] {\nat} {\conf{\frame{\syss(\syscall)}{\regmap_0'}{\km[\syscall]}:\frame{\cmdtwo}{\regmap}{\um}:\st', \lay_1\lcomp \rfs'}}{\unsafe}.
        \]
        By applying \Cref{lemma:contextplugging} to the
        reduction with $\lay_2$, we deduce that:
        \[
          \nstep[\system][\lay_2] {!} {\conf{\frame{\syss(\syscall)}{\regmap_0'}{\opt}:\frame{\cmdtwo}{\regmap}{\um}:\st', \lay_2\lcomp \rfs'}} {\err}.
        \]
        We conclude by observing that $\unsafe\evaleq \err$. The case with $\unsafe$ follows analogously.
        \proofcase{$\err$} Analogous to the previous case.
        \proofcase{$(\val, \mem)$}
        We start by observing that $\mem = \lay_1\lcomp \rfs''$ for some $\rfs''$ due to \Cref{rem:simispreserved}. Form the \emph{layout non-interference}
        property (\Cref{def:lni}),
        we deduce that
        \begin{equation*}
          \Eval[\system][\lay_2]{{\syss(\syscall)},{\regmap_0'},{\km[\syscall]},\rfs'} =
          {(\val,\rfs'')}.
        \end{equation*}
        By applying  \Cref{lemma:contextplugging} twice, we deduce that:
        \begin{multline*}
          \step[\system][\lay_1]
          {\conf{\frame{\csyscall  \syscall {\exprtwo_1, \dots, \exprtwo_k}\sep \cmdtwo}{\regmap}{\um}:\st', \lay_1\lcomp \rfs'}}{}\\
          {\conf{\frame{\syss(\syscall)}{\regmap_0'}{\km[\syscall]}:\frame{\cmdtwo}{\regmap}{\um}:\st', \lay_1\lcomp \rfs'}} \to^*\\
          {\conf{\frame{\cmdtwo}{\update\regmap\ret\val}{\um}:\st', \lay_1\lcomp \rfs''}}
        \end{multline*}
        and 
        \begin{multline*}
          \step[\system][\lay_2]
          {\conf{\frame{\csyscall  \syscall {\exprtwo_1, \dots, \exprtwo_k}\sep \cmdtwo}{\regmap}{\um}:\st', \lay_2\lcomp \rfs'}}{}\\
          {\conf{\frame{\syss(\syscall)}{\regmap_0'}{\km[\syscall]}:\frame{\cmdtwo}{\regmap}{\um}:\st', \lay_2\lcomp \rfs'}} \to^*\\
          {\conf{\frame{\cmdtwo}{\update\regmap\ret\val}{\um}:\st', \lay_2\lcomp \rfs''}}
        \end{multline*}
        Next, we observe that we assumed
        \begin{equation*}
          \nstep[\system][\lay_1]{n+1}
          {\conf{\frame{\csyscall  \syscall {\exprtwo_1, \dots, \exprtwo_k}\sep \cmdtwo}{\regmap}{\um}:\st', \lay_1\lcomp \rfs'}}{}
          {\confone_1},
        \end{equation*}
        and $\confone_1$ is terminal, so it must be the case that
        $\nstep[\system][\lay_1]!{\conf{\frame{\cmdtwo}{\update\regmap\ret\val}{\um}:\st',
            \lay_1\lcomp \rfs''}}{\confone_1}$. In particular, the
        reduction requires less than $\nat+1$ steps. From the
        assumption $\um(\st)$, we also deduce $\um(\st')$, and
        $\um(\frame{\cmdtwo}{\update\regmap\ret\val}{\um}:\st')$ and
        from \Cref{rem:simispreserved} we conclude that
        $\rfs''\eqon{\Fn}\rfs'\eqon{\Fn}\rfs$, so we can apply the IH
        to
        \[
          {\conf{\frame{\cmdtwo}{\update\regmap\ret\val}{\um}:\st', \lay_1\lcomp \rfs''}}
        \]
        and deduce that
        \[
          \Eval[\system][\lay_2]{\conf{\frame{\cnil}{\regmap'}{\opt_2}:\st', \lay_2\lcomp \rfs''}}\evaleq
          \confone_1,
        \]
        which concludes the proof.
      \end{proofcases}
    \end{proofcases}
  \end{proofcases}
\end{proof}

\begin{lemma}[Main Lemma]
  \label{lemma:mainlemmapms}
  Let $\system=(\rfs, \syss, \caps)$ be a \emph{layout non-interferent} system.
  For every $\nat \in \Nat$,
  \emph{unprivileged} program $\cmd$,
  register map $\regmap$, and distribution $\mu$,
  if
  \[
    \forall \lay\in \Lay. \exists \conftwo_\lay. \nstep{!\nat}{\conf{\frame{\cmd}{\regmap}{\um}, \lay \lcomp \rfs}} \conftwo_\lay,
  \]
  then one of the two following statements must hold: 
  \begin{varitemize}
  \item There are a register map $\regmap'$ and a store $\rfs'\eqon{\Fun}\rfs$
    such that, for every layout $\lay$, we have the reduction
    \[
      \nstep {! \nat}{\conf{\frame{\cmd}{\regmap}{\um}, \lay \lcomp \rfs}}{\conf{\frame{\cnil}{\regmap'}{\um}, \lay \lcomp \rfs'}}.
    \]
  \item
    $
      \Pr_{\lay \leftarrow \mu}\left[\nstep {! \nat} {\conf{\frame{\cmd}{\regmap}{\um}, \lay \lcomp \rfs}}\err \right] \ge \delta_\mu.
    $
  \end{varitemize}  
\end{lemma}

\begin{proof}
  To facilitate the induction, we prove a slightly stronger statement. Specifically, we replace the quantification over $\cmd$ and $\regmap$ with a quantification over a non-empty stack $\st$ such that $\um(\st)$ holds. Additionally, we quantify over a general $\rfs' \eqon{\Fn} \rfs$ instead of $\rfs$, and we rewrite the claim in a more convenient form. In the following claim, all meta-variables that are not explicitly quantified are globally and universally quantified. As a result, we obtain the following claim:
  if $\forall \lay\in \Lay. \exists \conftwo_\lay. \nstep{!\nat}{\conf{\st, \lay \lcomp \rfs'}} \conftwo_\lay$, then:
  \begin{varitemize}
  \item $\exists \regmap', \rfs''\eqon{\Fn}\rfs'. \forall \lay\in \Lay.\nstep {! \nat}{\conf{\st, \lay \lcomp \rfs'}}{} {\conf{\frame{\cnil}{\regmap'}{\um}, \lay \lcomp \rfs''}}$, or
  \item $\Pr_{\lay \leftarrow \mu}\left[\nstep {! \nat} {\conf{\st, \lay \lcomp \rfs'}}\err \right] \ge \delta_\mu$.
  \end{varitemize}
  We call this auxiliary claim (C') and we proceed by induction on $\nat$.
  \begin{proofcases}
    %
    %
    \proofcase{0} In this case, we assume
    \[
      \forall \lay\in \Lay. \exists \conftwo_\lay. \nstep{!0}{\conf{\st, \lay \lcomp \rfs'}} \conftwo_\lay,
    \]
    and from this premise, we deduce that
    $\conf{\st, \lay \lcomp \rfs'}$ is terminal. Using the definition
    of the notation $\to^!$, the fact that $\um(\st)$ holds, and the
    assumption on the non-emptiness of $\st$, we conclude that
    $\st = \frame \cnil \regmap \um : \nil$ for some $\regmap$. Hence, the claim
    holds trivially by choosing the register map $\regmap$ and the
    store $\rfs'$.
    %
    %
    \proofcase{$\nat+1$} In this case, we start by assuming
    \[
      \forall \lay\in \Lay. \exists \conftwo_\lay. \nstep{!\nat+1}{\conf{\st, \lay \lcomp \rfs'}} \conftwo_\lay
      \tag{H}
    \]
    and we go by cases on $\st$, excluding the case where $\st =\nil$ due to our assumption on the non-emptiness of $\st$. 
    Therefore, in the following, we assume that $\st = \frame \cmd \regmap \um :\st'$,
    and we proceed by cases on $\cmd$. Notice that the induction
    hypothesis coincides syntactically with (C').
    \begin{proofcases}
      \proofcase{$\vx \ass\expr\sep \cmdtwo$} In this case, we observe
      that there exists a value $\val$ such that, for every layout
      $\lay$, we have $\sem \expr_{\regmap, \lay} = \val$. More
      specifically, this follows from \Cref{rem:expreval}, since for
      every $\id \in \Idu$ and every pair of layouts $\lay_1, \lay_2$,
      it holds that $\lay_1(\id) = \lay_2(\id)$. In particular, all
      identifiers appearing in $\expr$ belong to $\Idu$.
      By introspection of the semantics, we deduce that for every layout $\lay$:
      \begin{equation*}
        \step {\conf{\frame{\vx \ass\expr\sep\cmdtwo}{\regmap}{\um}:\st', \lay\lcomp \rfs'}}{}\\
        {\conf{\frame{\cmdtwo}{\update \regmap\vx {\val}}{\um}:\st', \lay\lcomp \rfs'}}.
        \tag{$*$}
      \end{equation*}
      Notice that
      $\um(\frame{\cmdtwo}{\update \regmap\vx
        {\val}}{\um}:\st')$. Because of (H) and this observation, we
      can apply the IH to the stack in the target configuration and
      $\rfs'$ to conclude that one among (A) and (B) below holds.
      \begin{equation*}
        \exists \regmap', \rfs''\eqon{\Fn}\rfs'. \forall \lay\in \Lay.\\
        \nstep {! \nat}{\conf{\frame{\cmdtwo}{\update \regmap\vx {\val}}{\um}:\st', \lay \lcomp \rfs'}}{}\\
        {\conf{\frame{\cnil}{\regmap'}{\um}, \lay \lcomp \rfs''}}
        \tag{A}
      \end{equation*}
      \[
        \Pr_{\lay \leftarrow \mu}\left[\nstep {! \nat} {\conf{\frame{\cmdtwo}{\update \regmap\vx {\val}}{\um}:\st', \lay \lcomp \rfs'}}\err \right] \ge \delta_\mu.
        \tag{B}
      \]
      We go by cases on this disjunction.
      \begin{proofcases}
        \proofcase{A} In this case, we introduce $\regmap'$ and $\rfs''$ from the IH,
        we assume $\rfs''\eqon{\Fn}\rfs'$, we fix a layout $\lay$ and from
        ($*$) and (A), we conclude:
        \begin{equation*}
          \nstep{\nat+1} {\conf{\frame{\vx \ass\expr\sep\cmdtwo}{\regmap}{\um}:\st', \lay\lcomp \rfs'}}{}
          {\conf{\frame{\cnil}{\regmap'}{\um}, \lay \lcomp \rfs''}}.
        \end{equation*}
        Due to the generality of $\lay$, we can
        introduce universal quantification over all layouts, thereby
        completing the proof.

        \proofcase{B} We call $E$ the set of all the layouts $\overline \lay$ such that
        \[
          \nstep[\system][\overline \lay] {! \nat} {\conf{\frame{\cmdtwo}{\update \regmap\vx {\val}}{\um}:\st', \lay \lcomp \rfs'}}\err
        \]
        Notice that we can apply the assumption ($*$)
        to each of these layouts, showing that
        \[
          \forall \overline \lay\in E.\nstep[\system][\overline \lay]{!\nat+1} {\conf{\frame{\vx \ass\expr\sep\cmdtwo}{\regmap}{\um}:\st', \overline\lay\lcomp \rfs'}}{\err}.
        \]
        From (B), we know that the probability associated to the event $E$ is bigger than $\delta_\mu$,
        the last observation establishes the claim.
      \end{proofcases}
      \proofcase{$\cnil$} Analogous to the previous case.
      \proofcase{$\cskip\sep \cmdtwo$} Analogous to the case of
      assignments.
      \proofcase{$\cif \expr {\cmd_\bot}{\cmd_\top}\sep\cmdtwo$}
      Analogous to the case of assignments.
      \proofcase{$\cwhile \expr {\cmd}\sep\cmdtwo$} Analogous to the
      case of assignments.
      \proofcase{$\cmemass \expr \exprtwo\sep \cmdtwo$} In this case,
      we begin by observing that there exists a value $\val_\expr$
      such that for every $\lay \in \Lay$, we have
      $\sem\expr_{\regmap,\lay} = \val_\expr$. This follows from
      \Cref{rem:expreval}, since we assume that
      $\um(\cmemass \expr \exprtwo \sep \cmdtwo)$, which implies that
      all identifiers within $\expr$ belong to $\Idu$.
      For the same reason, there exists a unique value $\val_\exprtwo$
      such that for every layout $\lay$, we have
      $\sem\exprtwo_{\regmap,\lay} = \val_\exprtwo$. Consequently,
      there is a unique address $\add \in \Add$ such that
      $\add = \toAdd{\val}$.
      Similarly, we deduce that there exists a unique set of addresses
      $P_{\Aru}$ such that, for every layout $\lay$, we have
      $\underline{\lay}(\Aru) = P_{\Aru}$. Moreover, due to the
      assumptions on the set of layouts, we conclude that
      $P_{\Aru} \subseteq \Addu$.
      Finally, we proceed by cases on whether $\add \in P_{\Aru}$.
      \begin{proofcases}
        \proofcase{$\add \in P_{\Aru}$} The rule \ref{WL:Store}
        can be applied to each configuration
        \[
        {\conf{\frame{\cmemass  \expr \exprtwo\sep \cmdtwo}{\regmap}{\um}:\st', \lay\lcomp \rfs'}}
        \]
        for every layout $\lay$, obtaining:
        \[
          {\conf{\frame{\cmdtwo}{\regmap}{\um}:\st', \update{\lay\lcomp \rfs'}{\add}{\val_\expr}}};
        \]
        from the assumption $\add \in \underline \lay(\Ar[\um])$
        and the definition of $\underline \lay$, we deduce that
        there is $\ar \in \Ar[\um]$
        and $0\le i\le\size\ar$ such that $\lay(\ar)+i = \add$
        and that for every other layout $\overline \lay$
        it must hold that $\overline \lay(\ar)+i = \add$,
        so we can apply \Cref{rem:memupdtostupd}
        in order to show that
        \[
          \forall \lay \in \Lay. \update{\lay\lcomp \rfs'}{\add}{\val_\expr} =
          \lay\lcomp \update{\rfs'}{(\ar, i)}{\val_\expr}.
        \]
        Finally, we observe that $\um({\frame{\cmdtwo}{\regmap}{\um}:\st'})$ holds
        and that
        \[
          \update{\rfs'}{(\ar, i)}{\val_\expr}\eqon{\Fn}{\rfs'}\eqon{\Fn}\rfs
        \]
        holds as well, so we can apply the IH
        and conclude the proof of this sub-derivation as we did in the case of assignments.
        \proofcase{$p \notin P_{\Aru}$}
        In this case,
        we just observe that
        the rule \ref{WL:Store-Error}
        can be applied to show both
        \[
          \step[\system][\lay_1]
          {\conf{\frame{\cmemass  \expr \exprtwo\sep \cmdtwo}{\regmap}{\um}:\st', \lay_1\lcomp \rfs'}}
          \err,
        \]
        and
        \[
          \step[\system][\lay_2]
          {\conf{\frame{\cmemass  \expr \exprtwo\sep \cmdtwo}{\regmap}{\um}:\st', \lay_2\lcomp \rfs'}}
          \err.
        \]
        This establishes the claim. 
      \end{proofcases}
      \proofcase{$\cmemread \vx \expr\sep \cmdtwo$} Analogous to the case of store operations.
      \proofcase{$\ccall \expr {\exprtwo_1, \dots, \exprtwo_k}\sep \cmdtwo$} 
      This case is similar to that of stores but requires additional
considerations regarding the stacks of the target configurations.
      We begin by observing that there exists a unique address $\add$
such that for every $\lay \in \Lay$, we have $\toAdd{\sem
\expr_{\regmap,\lay}} = \add$. Similarly, we define the values
$\val_1, \dots, \val_k$, which correspond to the semantics of
$\exprtwo_1, \dots, \exprtwo_k$ evaluated under $\regmap$ for every
layout $\lay \in \Lay$.
      Next, we note the existence of a set $P_{\Fn[\um]}$ such that
for all $\lay \in \Lay$, we have $\underline{\lay} (\Fn[\um]) =
P_{\Fn[\um]}$.
      The proof proceeds by cases on $\add \in P_{\Fn[\um]}$.

      \begin{proofcases}
        \proofcase{$\add\in P_{\Fn[\um]}$} In this case, there exists
        a unique $\fn \in \Fn[\um]$ such that for every
        $\lay \in \Lay$, we have $\lay(\fn) = \add$. From this
        observation and the definition of $\lcomp$, we deduce that for
        every $\lay \in \Lay$, $(\lay\lcomp \rfs')(\add) = \rfs'(\fn) = \rfs(\fn)$,
        where the last step follows from the assumption $\rfs' \eqon{\Fn} \rfs$. Since
        $\add\in P_{\Fn[\um]}$, we deduce that:
        \begin{multline*}
          \forall \add\in P_{\Fn[\um]}.
          \step
          {\conf{\frame{\ccall  \expr {\exprtwo_1, \dots, \exprtwo_k}\sep \cmdtwo}{\regmap}{\um}:\st', \lay\lcomp \rfs'}}{}\\
          {\conf{\frame{\rfs(\fn)}{\regmap_0'}{\um}:\frame{\cmdtwo}{\regmap}{\um}:\st', \lay\lcomp \rfs'}},
        \end{multline*}
        where
        $\regmap_0'= \regmap_0[\vx_1,\dots,\vx_k\upd
        \val_1,\dots,\val_k]$.  Finally, we observe that
        $\um(\frame{\rfs(\fn)}{\regmap_0'}{\um}:\frame{\cmdtwo}{\regmap}{\um}:\st')$
        holds since $\rfs$ maps unprivileged commands to identifiers
        in $\Fnu$.  This allows us to apply the IH. The proof then
        follows the same structure as in the case of assignments.

        \proofcase{$p \notin P_{\Fn[\um]}$} Analogous to the
        corresponding case for store operations.
      \end{proofcases}
      \proofcase{$\csyscall \syscall{\expr_1, \ldots, \expr_k}\sep\cmdtwo$}
      We begin by introducing the values $\val_1, \dots, \val_k$, which correspond to the semantics of $\exprtwo_1, \dots, \exprtwo_k$ evaluated under $\regmap$ and all layouts $\lay \in \Lay$. By introspection of the rule \ref{WL:SystemCall}, we deduce that for every $\lay \in \Lay$, the following holds:
      \begin{equation*}
        \step
        {\conf{\frame{\csyscall  \syscall {\exprtwo_1, \dots, \exprtwo_k}\sep \cmdtwo}{\regmap}{\um}:\st', \lay\lcomp \rfs'}}{}\\
        {\conf{\frame{\syss(\syscall)}{\regmap_0'}{\km[\syscall]}:\frame{\cmdtwo}{\regmap}{\um}:\st', \lay\lcomp \rfs'}},
        \tag{\dag}
      \end{equation*}
      where $\regmap_0'= \regmap_0[\vx_1,\dots,\vx_k\upd \val_1,\dots,\val_k]$. However, we cannot apply the IH on the configuration
      \[
        {\conf{\frame{\syss(\syscall)}{\regmap_0'}{\km[\syscall]}:\frame{\cmdtwo}{\regmap}{\um}:\st', \lay\lcomp \rfs'}}
      \]
      because
      $\um(\frame{\syss(\syscall)}{\regmap_0'}{\km[\syscall]}:\frame{\cmdtwo}{\regmap}{\um}:\st')$
      does not hold. Nevertheless, we can apply
      \Cref{lemma:onsyscallterm} and deduce that
      one of the following statements must hold:
      \begin{equation*}
        \exists \overline \regmap, \overline \rfs\eqon{\Fn} \rfs'.\forall \lay \in \Lay.\\
        \nstep !{\conf{\frame{\syss(\syscall)}{\regmap_0'}{\km[\syscall]}, {\lay\lcomp \rfs'}}}  {\conf{\frame{\cnil}{\overline \regmap}{\km[\syscall]}, \lay\lcomp \rfs''}},
        \tag{A}
      \end{equation*}
      \[
        \Pr_{\lay\leftarrow\mu}\left[\nstep !{\conf{\frame{\syss(\syscall)}{\regmap_0'}{\km[\syscall]}, \lay\lcomp \rfs'}}\err\right] \ge \delta_\mu,
        \tag{B}
      \]
      \[
        \forall \lay \in \Lay.\Eval{\conf{\frame{\syss(\syscall)}{\regmap_0'}{\km[\syscall]}, \lay\lcomp \rfs'}} = \Omega.
        \tag{C}
      \]
      We go by cases on the valid statement.
      \begin{proofcases}
        \proofcase{A} From (A) and \Cref{lemma:contextplugging}, we conclude that
        there exist a register map $\overline \regmap$ and a store
        $\overline \rfs$ such that for every
        layout $\lay$, we have:
        \begin{equation*}
          \nstep !        {\conf{\frame{\syss(\syscall)}{\regmap_0'}{\km[\syscall]}:\frame{\cmdtwo}{\regmap}{\um}:\st', \lay\lcomp \rfs'}}{}\\
          {\conf{\frame{\cmdtwo}{\update\regmap \ret{\overline \regmap(\ret)}}{\um}:\st', \lay\lcomp \overline \rfs}}.
        \end{equation*}
        We also observe that $\rfs \eqon{\Fn} \rfs'$ by \Cref{rem:simispreserved}. Combining this with (\dag), we conclude:
        \begin{equation*}
          \nstep ! {\conf{\frame{\csyscall  \syscall {\exprtwo_1, \dots, \exprtwo_k}\sep \cmdtwo}{\regmap}{\um}:\st', \lay\lcomp \rfs'}}{}\\
          {\conf{\frame{\cmdtwo}{\update\regmap \ret{\overline \regmap(\ret)}}{\um}:\st', \lay\lcomp \overline \rfs}}
        \end{equation*}
        holds as well.
        From (H), we also deduce that the number of steps must be at most $\nat+1$, while from (\dag), we infer that it must be greater than 1. This allows us to apply the IH to 
        ${\conf{\frame{\cmdtwo}{\update\regmap \ret{\overline \regmap(\ret)}}{\um}:\st', \lay\lcomp \overline \rfs}}$. The proof then proceeds similarly to the case of assignments.
        \proofcase{B} Define $E$ as the set of all layouts $\overline \lay \in \Lay$ such that:
        \[
          \nstep[\system][\overline \lay] {!} {\conf{\frame{\syss(\syscall)}{\regmap_0'}{\km[\syscall]}, \overline \lay\lcomp \rfs'}}\err
        \]
        Since (\dag) holds for each of these layouts, it follows that
        for every $\overline \lay \in E$, we have:
        \[
          \nstep[\system][\overline \lay]{!} {\conf{\frame{\csyscall  \syscall {\exprtwo_1, \dots, \exprtwo_k}\sep \cmdtwo}{\regmap}{\um}:\st', \overline \lay\lcomp \rfs'}}{\err}.
        \]
        From (B), we know that the probability measure associated with
        this set is greater than $\delta_\mu$. Furthermore, from (H),
        we deduce:
        \[
          \nstep[\system][\overline \lay]{!\nat+1} {\conf{\frame{\csyscall  \syscall {\exprtwo_1, \dots, \exprtwo_k}\sep \cmdtwo}{\regmap, \um}:\st', \overline \lay\lcomp \rfs'}}{\err},
        \]
        which completes the claim.
        \proofcase{C} Let $\lay\in \Lay$ be any layout.
        From \Cref{lemma:contextplugging}, we deduce that 
        \[
          \Eval{\conf{\frame{\syss(\syscall)}{\regmap_0'}{\km[\syscall]}:\frame{\cmdtwo}{\regmap}{\um}:\st', \lay\lcomp \rfs'}} = \Omega,
        \]
        which contradicts our assumption (H). This completes the proof.
      \end{proofcases}
    \end{proofcases}
  \end{proofcases}
  The main claim is a particular case of C' where $\st$ is $\frame{\cmd}{\regmap}{\um}$, and $\rfs'$ is $\rfs$.
\end{proof}


\onsyscallterm*

\begin{proof}
  We begin by fixing all the universally quantified variables in the
  statement using the same meta-variable as before. We also fix some
  $\overline{\lay} \in \Lay$ and proceed by case analysis on
  \[
    \Eval[\system][\overline \lay]{\conf{\frame{\syss(\syscall)}{\regmap}{\km[\syscall]}, \lay\lcomp \rfs'}}.
  \]
  \begin{proofcases}
    \proofcase{$(\val, \rfs'')$}
    In this case, by the \emph{layout non-interference} property, we conclude that (A) holds for $\overline{\val} = \val$ and $\overline{\rfs} = \rfs''$. Furthermore, by \Cref{rem:simispreserved}, we know that $\rfs'' \eqon{\Fn} \rfs' \eqon{\Fn} \rfs$.
    \proofcase{$\unsafe, \err$}
    Here, by the \emph{layout non-interference} property, we deduce that for every $\lay \in \Lay$,
    \[
      \Eval{\conf{\frame{\syss(\syscall)}{\regmap}{\km[\syscall]}, \lay\lcomp \rfs'}} \in \{\err, \unsafe\}.
    \]
    Consequently, for every $\lay \in \Lay$, there exists some $\nat_\lay$ such that
    \[
      \exists \nat_\lay.\nstep {\nat_\lay} {\conf{\frame{\syss(\syscall)}{\regmap}{\km[\syscall]}, \lay\lcomp \rfs'}} \conftwo \in\{\err, \unsafe\}.
    \]
    Since $\Lay$ is finite, there exists a uniform bound $\overline{\nat}$ such that
    \[
      \forall \lay\in \Lay.\nstep {!\overline \nat} {\conf{\frame{\syss(\syscall)}{\regmap}{\km[\syscall]}, \lay\lcomp \rfs'}} \conftwo \in\{\err, \unsafe\}.
      \tag{\dag}
    \]
    From this, we conclude that  
    \[
      \Pr_{\lay \leftarrow \mu}\left[ \nstep {!\overline \nat} {\conf{\frame{\syss(\syscall)}{\regmap}{\km[\syscall]}, \lay\lcomp \rfs'}}\err\right]\ge \delta_\mu.
    \]
    This follows by considering an enumeration $\id_0, \dots, \id_h$ of the identifiers in $\refs_\system(\syss(\syscall))$ and noting that the probability above can be rewritten as:
    \begin{multline*}
      \sum_{\add_1, \dots, \add_h \in \Addk}
      \Pr_{\lay \leftarrow \mu}\Big[ \nstep {!\overline \nat} {\conf{\frame{\syss(\syscall)}{\regmap}{\km[\syscall]}, \lay\lcomp \rfs'}}\err \,\,\Big|\,\,
      \forall 1\le i\le h.\lay(\id_i)=\add_i\Big] \cdot\\
      \Pr_{\lay \leftarrow \mu}[\lay(\id_1)=\add_1,\dots, \lay(\id_h)=\add_h]
    \end{multline*}
    From (\dag) and \Cref{lemma:onsyscall}, we deduce that for every
    choice of $\add_1, \dots, \add_h \in \Addk$, the terms on the left
    in the expression above are bounded by $\delta_{\mu}$ making
    their convex combination bounded by $\delta_{\mu}$. moreover, since:
    \[
      \nstep {!\overline \nat} {\conf{\frame{\syss(\syscall)}{\regmap}{\km[\syscall]}, \lay\lcomp \rfs'}}\err
      \Rightarrow
      \nstep {!} {\conf{\frame{\syss(\syscall)}{\regmap}{\km[\syscall]}, \lay\lcomp \rfs'}}\err
    \]
    for the definition of $\to^{!m}$, and this means that 
    \[
      \Pr_{\lay \leftarrow \mu}\left[ \nstep {!} {\conf{\frame{\syss(\syscall)}{\regmap}{\km[\syscall]}, \lay\lcomp \rfs'}}\err\right]
      \ge
      \Pr_{\lay \leftarrow \mu}\left[ \nstep {!\overline \nat} {\conf{\frame{\syss(\syscall)}{\regmap}{\km[\syscall]}, \lay\lcomp \rfs'}}\err\right],
    \]
    which establishes the claim.
    \proofcase{$\Omega$} In this case, due to the \emph{layout non-interference} property, we conclude that (C) holds.
  \end{proofcases}
\end{proof}

For the next lemma, we extend $\refs_\sigma$ to frame stacks as follows:
\begin{align*}
  \refs_\sigma(\nil) &\defsym \emptyset&
  \refs_\sigma(f\cons \st) &\defsym \refs_\sigma(f) \cup \refs_\sigma(\st) &
  \refs_\sigma(\frame \cmd \regmap \opt) &\defsym \refs_\sigma(\cmd).
\end{align*}

\onsyscall*

\begin{proof}
  We proceed by induction on $\nat$.
  \begin{proofcases}
    \proofcase{0}
    In this case, we conclude that either (A) or (C) holds trivially, depending on whether $\syss(\syscall) \neq \cnil$ or not.
    \proofcase{$\nat+1$}
    We begin by applying the IH to the initial configuration. Three cases arise:
  \begin{proofcases}
    \proofcase{A} Let $\system$ be a system, $\syscall$ a system call
    and consider the initial configuration
    $ {\conf{\frame{\syss(\syscall)}{\regmap}{\km[\syscall]}, \lay
        \lcomp \rfs'}}$.  Suppose
    $\refs_\system(\syss(\syscall))\setminus \Sys= \{\id_1, \dots,
    \id_h\}$ and fix addresses $\add_1, \dots, \add_h$.  By the IH,
    there exists a stack $\st$ and a
    store $\rfs'' \eqon{\Fn} \rfs$ satisfying the following property:
    \[
      \forall \lay. \left( \forall 1 \leq i \leq h, \lay(\id_i) = \add_i\right)  \Rightarrow
      \nstep{\nat}
      {\conf{\frame{\syss(\syscall)}{\regmap}{\km[\syscall]}, \lay \lcomp \rfs'}}
      {\conf{\st, \lay \lcomp \rfs''}}.
    \]
    Since $\st$ is non-empty, we can assume that $\st = \frame \cmd {\regmap'} {\km[\syscall]}:\st'$.
    Furthermore, from the IH, we know that that $\refs_\system(\st)\subseteq \refs_\system(\syss(\syscall))$ (H). We now perform a case analysis on $\cmd$.
    \begin{proofcases}
      \proofcase{$\vx \ass\expr\sep \cmdtwo$} We first observe that
      there exists a value $\val$ such that, for every layout in our
      quantification $\lay$, $\sem \expr_{\regmap, \lay} = \val$.
      This follows from \Cref{rem:expreval}, given that for every
      $\id \in \refs_\system(\syss(\syscall))$ and any two layouts
      $\lay_1, \lay_2$, we assume that $\lay_1(\id) = \lay_2(\id)$. In
      particular, this applies to identifiers appearing in $\expr$
      that belong to $\refs_\system(\syss(\syscall))$ by (H).  For all
      layouts if our quantification layouts, we have:
      \begin{equation*}
        \step {\conf{\frame{\vx \ass\expr\sep\cmdtwo}{\regmap'}{\km[\syscall]}:\st', \lay\lcomp \rfs''}}{}\\
        {\conf{\frame{\cmdtwo}{\update {\regmap'}\vx {\val}}{\km}:\st', \lay\lcomp \rfs''}}.
        \tag{$*$}
      \end{equation*}
      From (H), we deduce
      \[
        \refs_\system(\frame{\cmdtwo}{\update {\regmap'}\vx {\val}}{\km}:\st')\subseteq\refs_\system(\syss(\syscall))
      \]
      Therefore, we conclude that (A) holds if either
      $\cmdtwo \neq \cnil$ or $\st'$ is non-empty. Otherwise, (C)
      holds.  \proofcase{$\cnil$} Analogous to the case above.
      \proofcase{$\cskip\sep \cmdtwo$} Analogous to the case of
      assignments.
      \proofcase{$\cif \expr {\cmd_\bot}{\cmd_\top}\sep\cmdtwo$}
      Analogous to the case of assignments.
      \proofcase{$\cwhile \expr {\cmd}\sep\cmdtwo$} Analogous to the
      case of assignments.
      \proofcase{$\cmemass \expr \exprtwo\sep \cmdtwo$} In this case,
      we start by observing that there exist a value $\val_\expr$ such
      that for each layout $\lay$ that satisfies the premise, we have
      $\sem\expr_{\regmap,\lay} = \val_\expr$. This follows from
      \Cref{rem:expreval}, (H2), and the definition of $\refs$, which
      ensures that all the identifiers within $\expr$ are in
      $\refs(\frame{\cmemass\expr\exprtwo\sep\cmdtwo}{\regmap'}{\km[\syscall]}:\st')$.
      For the same reason, there is a unique value $\val_\exprtwo$
      such that for each of the layout in our quantification, we have
      that $\sem\exprtwo_{\regmaptwo,\lay} = \val_\exprtwo$.
      Therefore, there is a unique $\add\in \Add = \toAdd{\val_\expr}$ which
      the store instruction attempts to write at.  Finally, we observe
      that there is a unique set $P$ such that, for every layout
      $\lay$ satisfying the assumption, we have
      $P = \underline \lay(\refs_\system(\syss(\syscall))\setminus \Sys)$.  We go by
      case analysis on $\add \in P$.
      \begin{proofcases}
        \proofcase{$\add\in P$} In this case, for every layout $\lay$
        that we are quantifying over, there exists some index $i$ such
        that $\add \in \underline \lay(\id_i)$.  If $\id_i$ is a
        procedure identifier, then we have
        $\add = \underline{\lay}(\id_i)$, and thus, independently of $\lay$, we have:
        \[
          \step
          {\conf{\frame{\cmemass\expr\exprtwo\sep\cmdtwo}{\regmap'}{\km[\syscall]}:\st',
              \lay\lcomp \rfs''}} {\err}.
        \]
        This shows that (B) holds with probability 1.  Otherwise,
        there exists a unique array $\ar$ and a unique index $j$ such
        that, independently of the layout, we have
        $\lay(\ar) + j = \add$. Notice that if $\ar$ and $j$ were not
        unique, then it would not be true that the layouts store
        $\id_1, \dots,\id_h$ respectively at $\add_1, \dots, \add_h$.
        Thus, each layout $\lay$ that satisfies the premises of rule
        \ref{WL:Store}, we apply the rule to show the following
        transition:
        \begin{equation*}
          \step {\conf{\frame{\cmemass\expr\exprtwo\sep\cmdtwo}{\regmap'}{\km[\syscall]}:\st', \lay\lcomp \rfs''}}{}\\
          {\conf{\frame{\cmdtwo}{\regmap'}{\km[\syscall]}:\st', \update{(\lay\lcomp \rfs'')}\add{\val_\exprtwo}}}.
        \end{equation*}
        Using \Cref{rem:memupdtostupd}, we further conclude:
        \begin{equation*}
          \step {\conf{\frame{\cmemass\expr\exprtwo\sep\cmdtwo}{\regmap'}{\km[\syscall]}:\st', \lay\lcomp \rfs''}}{}\\
          {\conf{\frame{\cmdtwo}{\regmap'}{\km[\syscall]}:\st', \lay\lcomp {(\update{\rfs''}{(\ar, j)}{\val_\exprtwo})}}}.
        \end{equation*}
        By observing the uniqueness of the target configuration modulo
        $\lay$, and given that
        $\refs_\system(\frame{\cmdtwo}{\regmap'}{\km[\syscall]}:\st')
        \subseteq \refs_\system(\syss(\syscall))$ follows from (H), we
        conclude that (A) holds. If $\cmdtwo = \cnil$ and $\st'$ is
        empty (C) holds, instead.
        \proofcase{$\add\notin P$}
        %
        Observe that, for every layout $\lay$ such that
        $\add \notin \underline{\lay}(\Idk)$,
        only rule \ref{WL:Store-Error} applies, which shows the following transition:
        \[
          \step {\conf{\frame{\cmemass\expr\exprtwo\sep\cmdtwo}{\regmap'}{\km[\syscall]}:\st', \lay\lcomp \rfs''}}\err.
        \]
        Thus, we observe that:
        \begin{equation*}
          \Pr_{\lay \leftarrow\mu}\big[\step {\conf{\frame{\cmemass\expr\exprtwo\sep\cmdtwo}{\regmap'}{\km[\syscall]}:\st', \lay\lcomp \rfs''}}\err\,\, \big|\\
          \forall 1\le i\le h.\lay(\id_i)=\add_i\big]
        \end{equation*}
        is greater than
        \[
          \Pr_{\lay \leftarrow\mu}\big[ \add \notin \underline\lay(\Idk) \mid
          \forall 1\le i\le h.\lay(\id_i)=\add_i\big]
        \]
        which, by definition, is greater than $\delta_{\mu}$.  
        This shows that (B) holds.
      \end{proofcases}
      \proofcase{$\cmemread \vx \expr\sep \cmdtwo$} Analogous to the
      case of store operations.
      \proofcase{$\ccall \expr {\exprtwo_1, \dots, \exprtwo_k}\sep
        \cmdtwo$} We start by observing that there exists a unique
      address $\add$ such that for every $\lay$ that satisfies the
      precondition, we have $\toAdd{\sem
        \expr_{\regmap,\lay}}=\add$. Similarly, we introduce the
      values $\val_1, \dots,\val_k$, which correspond to the semantics
      of $\exprtwo_1,\dots,\exprtwo_k$ evaluated under $\regmap$ and
      every layout that satisfies the precondition.  Then, we observe
      that there exists a set $P$ such that for each layout under
      consideration, it holds that
      $\underline{\lay} (\refs_\system(\syss(\syscall))\setminus
      \Sys)=P$. The proof proceeds by cases on whether $\add \in P$.
      \begin{proofcases}
        \proofcase{$\add\in P$} In this case, there is a unique identifier $\id_j$ such that for every layout $\lay$ that satisfies the precondition, we have $\add \in \underline{\lay}(\id_j)$. We analyze two cases based on whether $\id_j$ is a function identifier $ \fn$.
        \begin{proofcases}
          \proofcase{$\id_j= \fn$} In this case, from the definition
          of $\lcomp$, we deduce that for each of these layouts we
          have $\lay\lcomp \rfs''(\add) = \rfs''(\fn) = \rfs(\fn)$,
          where the last step follows from the assumption
          $\rfs''\eqon{\Fn}\rfs$.  Since $\add \in P$, we conclude
          that, independently of the specific layout, if the
          preconditions hold, then:
          \begin{equation*}
            \step
            {\conf{\frame{\ccall  \expr {\exprtwo_1, \dots, \exprtwo_k}\sep \cmdtwo}{\regmap}{\km[\syscall]}:\st', \lay\lcomp \rfs''}}{}\\
            {\conf{\frame{\rfs(\fn)}{\regmap_0'}{\km[\syscall]}:\frame{\cmdtwo}{\regmap}{\km[\syscall]}:\st', \lay\lcomp \rfs''}},
          \end{equation*}
          where $\regmap_0'= \regmap_0[\vx_1,\dots,\vx_k\upd \val_1,\dots,\val_k]$.
          By definition, $\refs_\system(\syss(\syscall))$
          contains all the identifiers within
          $\rfs(\fn) = \rfs(\id_j)$ because $\id_j \in \refs_\system(\syss(\syscall))$ and
          $\refs$ is closed under procedure calls.
          This shows that (A) holds.
          \proofcase{$\id_j =\ar$} In this case, since the set of
          array identifiers and that of functions are disjoint, we
          conclude that for every layout $\lay$ that satisfies the
          preconditions, we have that
          $\add \notin \underline \lay(\Fn[{\km}])$.  This means that
          for each of these layouts, we can show:
          \begin{equation*}
            \step
            {\conf{\frame{\ccall  \expr {\exprtwo_1, \dots, \exprtwo_k}\sep \cmdtwo}{\regmap}{\km[\syscall]}:\st', \lay\lcomp \rfs''}}
            {\err},
          \end{equation*}
          and this means that (B) holds with probability 1.
        \end{proofcases}
        \proofcase{$p \notin P$} Analogous to the corresponding case
        for store operations.
      \end{proofcases}
      \proofcase{$\csyscall \expr {\exprtwo_1, \dots, \exprtwo_k}\sep \cmdtwo$}
      Analogous to the previous case.
    \end{proofcases}
  \end{proofcases}
    \proofcase{B} From the definition of $\to^{!\nat}$, we observe that if
    \[
      \nstep {! \nat} {\conf{\frame{\syss(\syscall)}{\regmap}{\km[\syscall]}, \lay\lcomp \rfs'}}\err,
    \]
    then
    \[
      \nstep {! \nat+1} {\conf{\frame{\syss(\syscall)}{\regmap}{\km[\syscall]}, \lay\lcomp \rfs'}}\err.
    \]
    This implies that
    \[
      \Pr_{\lay\leftarrow \mu}\Big[ \nstep {! \nat+1} {\conf{\frame{\syss(\syscall)}{\regmap}{\km[\syscall]}, \lay\lcomp \rfs'}}\err \,\,\Big|\,\,
      \forall 1\le i\le h.\lay(\id_i)=\add_i\Big]
    \]
    is greater than
    \[
      \Pr_{\lay\leftarrow \mu}\Big[ \nstep {! \nat} {\conf{\frame{\syss(\syscall)}{\regmap}{\km[\syscall]}, \lay\lcomp \rfs'}}\err \,\,\Big|\,\,
      \forall 1\le i\le h.\lay(\id_i)=\add_i\Big],
    \]
    which proves the claim.
    \proofcase{C} Similar to the previous case.
  \end{proofcases}
\end{proof}

\begin{remark}
  \label{rem:deltanu}
  \[
    \delta_\nu \ge \min_{\syscall \in \Sys}\frac{\kappa_{\km}/\theta-|{\Idk}|} {\kappa_{\km}/\theta-|{\refs_\system(\syss(\syscall))\setminus \Sys}|}
  \]
\end{remark}

\begin{proof}
  We want to prove that
  \begin{multline*}
     \min\bigl \{
      \displaystyle{\Pr_{\lay\leftarrow \nu}}
      [ \add \notin \underline\lay(\Id) \mid \lay(\id_i^\syscall)=p_i,
      \text{ for } 1\le i \le h]  \mid \syscall \in \Sys, p,p_1,\dots,p_h \in \Addk \land {}\\[3pt]
      \phantom{{} \mid} p \notin \{p_i, \ldots, p_i +\size{\id_i^\syscall}-1\}, \text{ for } 1\le i \le h
      \bigr \}.
    \end{multline*}
    is greater than
    $\min_{\syscall \in \Sys}\frac{\kappa_{\km}/\theta-|{\Idk}|} {\kappa_{\km}/\theta-|{\refs_\system(\syss(\syscall))}\setminus \Sys|}$.
    To establish this, we fix a system call $\syscall \in \Sys$
    and addresses
    $p,p_1,\dots,p_h \in \Addk$ that achieve the minimum.
    In particular, we have:
  \[
    p \notin \{p_i, \ldots, p_i +\size{\id_i^\syscall}-1\}, \text{ for } 1\le i \le h.
    \tag{\dag}
  \]
  Given that the minimum exists, $p_1, \ldots, p_h$ are the starting
  addresses of the slots $s_1, \ldots, s_h$ where the references
  $\refs_\system(\syss(\syscall))\setminus \Sys$ of $\syscall$ are allocated.
  Now, assume that $\add$ is located within one of these slots, say
  $s_j$. From assumption (\dag), we deduce that $\add$ cannot be
  allocated, since no object other than $\id_j$ can be placed in that
  slot. As a result, we obtain:
  \[
    \Pr_{\lay\leftarrow \nu}
    [ \add \notin \underline\lay(\Id) \mid \lay(\id_i^\syscall)=p_i,
    \text{ for } 1\le i \le h] = 1,
  \]
  which establishes the claim.

  If, on the other hand, $\add$ is located in a different slot $s$, we observe that:
  \begin{multline*}
    \Pr_{\lay\leftarrow \nu}
    [ \add \notin \underline \lay(\Id) \mid \lay(\id_i^\syscall)=p_i,
    \text{ for } 1\le i \le h] \ge\\
    \Pr_{\lay\leftarrow \nu}
    [ \mathit{slotof}(\add) \notin \lay(\Id) \mid \lay(\id_i^\syscall)=p_i,
    \text{ for } 1\le i \le h]
    \tag{$*$}
  \end{multline*} 
  where $\mathit{slotof}$ associates each address with the starting
  address of its corresponding slot.
  Due to the definition of $\nu$, the slots are sampled uniformly and
  they are independent. Therefore, the probability to the right in
  ($*$) is given by the ratio of the free slots and those that are not
  occupied by elements of
  $\refs_\system(\syss(\syscall))\setminus \Sys$, because of the
assumption (\dag). In conclusion ($*$) is equal to:
  \[
    \frac{\kappa_{\km}/\theta-|{\Idk}|}
    {\kappa_{\km}/\theta-|{\refs_\system(\syss(\syscall))}\setminus \Sys|} \ge
    \min_{\syscall \in \Sys}\frac{\kappa_{\km}/\theta-|{\Idk}|}
    {\kappa_{\km}/\theta-|{\refs_\system(\syss(\syscall))}\setminus \Sys|}.
  \]
\end{proof}

\paragraph{Technical observations}

\begin{remark}
  \label{rem:expreval}
Let $ \{\id_1, \dots, \id_h\} \subseteq \Id $ be a set of identifiers, $ \{\add_1, \dots, \add_h\} $ a set of addresses, and $ W \subseteq \{\lay \in \Lay \mid \forall 1 \leq i \leq h, \, \lay(\id_i) = p_i \} $ a set of layouts. Given an expression $ \expr $ such that $ \ids(\expr) \subseteq \{\id_1, \dots, \id_h\} $ and a register map $ \regmap $, there exists a value $ \val \in \Val $ such that
\[
\forall \lay \in W, \, \sem{\expr}_{\regmap, \lay} = \val.
\]
\end{remark}
\begin{proof}
  We proceed by cases on the size of $ W $. If $ |W| = 0 $ or $ |W| = 1 $, the claim is trivial. Otherwise, we prove the following auxiliary claim:
  \[
    \forall \expr. \ids(\expr) \subseteq \{\id_1,\dots,\id_h\}\Rightarrow \forall \lay_1, \lay_2  \in \Lay. \forall 1\le i \le h. \lay_1(\id_i) = \lay_2(\id_i) \Rightarrow
    \sem \expr_{\regmap, \lay_1} =\sem \expr_{\regmap, \lay_2},
  \]
  which can be proved by induction on the syntax of $ \expr $. The main claim follows from the IH on $ W \setminus \{\lay\} $ for some layout $ \lay $, the application of the auxiliary claim to a pair of layouts $ \overline{\lay} \in W \setminus \{\lay\} $ and $ \lay $, and the transitivity of equality.
\end{proof}

\begin{remark}
  \label{rem:memupdtostupd}
  Let $ \add \in \Add $ be an address, $ \ar \in \Arr $ be an array, and $ \rfs $ be a store. For $ 0 \leq i < \size(\ar) $, if $ \add = \lay(\ar) + i $, then
  \[
    \update{(\lay \lcomp \ars)}{\add}{\val} = \lay \lcomp (\update{{\rfs}}{(\add, i)}{\val}).
  \]
\end{remark}

\begin{proof}
  The claim follows directly from unrolling the definitions.
\end{proof}

\begin{remark}
  \label{rem:simispreserved}
  For every layout $ \lay \in \Lay $, and pair of configurations $ \conf{\st, \lay \lcomp \rfs} $ and $ \conf{\st', \mem} $ such that $ \nstep*{\conf{\st, \lay \lcomp \rfs}}{\conf{\st', \mem}} $, it holds that $ \mem = \lay \lcomp \rfs' $ for some $ \rfs' \eqon{\Fn} \rfs $.
\end{remark}

\begin{proof}
  The proof proceeds by induction on the length of the reduction. The base case follows from the reflexivity of $ \eqon{\Fn} $. The inductive case follows from introspection of the rule that has been used for the last transition. The only non-trivial case occurs for the \ref{WL:Store} rule, where the premise $ \add \in \underline{\lay} (\Ar[\opt]) $ guarantees the existence of a pair $ (\ar, i) $ such that $ 0 \leq i < \size\ar $. The observation follows from \Cref{rem:memupdtostupd}.
\end{proof}

\begin{lemma}
  \label{lemma:contextpluggingtech}
  For every layout $\lay\in \Lay$, configuration $\conf{\frame{\cmd}{\regmap}{\opt}, \mem}$ and non-empty stack $\fr:\st$,  $\nat\in \Nat$, and configuration $\confone$:
  \begin{varitemize}
  \item if $\nstep \nat {\conf{\frame{\cmd}{\regmap}{\opt}, \lay\lcomp \rfs}} {\conf{\st', \mem'}}$, then:
    \[
      \nstep \nat {\conf{\frame{\cmd}{\regmap}{\opt}:\fr:\st, \mem}} {\conf{\st':\fr:\st, \mem'}}
    \]
  \item if $\nstep \nat {\conf{\frame{\cmd}{\regmap}{\opt}, \lay\lcomp \rfs}} {\err}$, then:
    \[
      \nstep \nat {\conf{\frame{\cmd}{\regmap}{\opt}:\fr:\st, \mem}} {\err}
    \] 
  \item if $\nstep \nat {\conf{\frame{\cmd}{\regmap}{\opt}, \lay\lcomp \rfs}} {\unsafe}$, then:
    \[
      \nstep \nat {\conf{\frame{\cmd}{\regmap}{\opt}:\fr:\st, \mem}} {\unsafe}
    \] 
  \end{varitemize}
\end{lemma}

\begin{proof}
  By induction on $\nat$. The base case is trivial.  The inductive
  case follows from the IH and by introspection of the rule that
  has been used for showing the last transition. In particular, it suffices to
  observe that every rule that can be applied to show
  $
  \step {\conf{\st', \mem'}}
    {\conf{\st'', \mem'}}
  $
  can also be applied to show
  $
  \step
  {\conf{\st':\fr:\st, \mem'}}
  {\conf{\st'':\fr:\st, \mem'}}.
  $
\end{proof}

\begin{lemma}
  \label{lemma:contextplugging}
  For every layout $\lay\in \Lay$, configuration $\conf{\frame{\cmd}{\regmap}{\opt}, \lay\lcomp\rfs}$ and non-empty stack $\frame \cmdtwo {\regmap'}{\opt}:\st$,  $\nat\in \Nat$, and configuration $\confone$:
  \begin{varitemize}
  \item if $\Eval {\conf{\frame{\cmd}{\regmap}{\opt}, \lay\lcomp \rfs}} = {\val, \rfs'}$, then:
    \begin{equation*}
      \nstep * {\conf{\frame{\cmd}{\regmap}{\opt}:\frame \cmdtwo {\regmap'}{\opt}:\st, \lay\lcomp\rfs}} {}\\
      {\conf{\frame \cmdtwo {\update{\regmap'}\ret{\val}}{\opt}:\st, \lay\lcomp\rfs'}}
    \end{equation*}
  \item if $\Eval {\conf{\frame{\cmd}{\regmap}{\opt}, \lay\lcomp \rfs}} = \err$, then:
    \[
      \nstep * {\conf{\frame{\cmd}{\regmap}{\opt}:\frame \cmdtwo {\regmap'}{\opt}:\st, \lay\lcomp\rfs}} {\err}
    \] 
  \item if $\Eval {\conf{\frame{\cmd}{\regmap}{\opt}, \lay\lcomp \rfs}} = \unsafe$, then:
    \[
      \nstep * {\conf{\frame{\cmd}{\regmap}{\opt}:\frame \cmdtwo {\regmap'}{\opt}:\st, \lay\lcomp\rfs}} {\unsafe}
    \] 
  \item if $\Eval {\conf{\frame{\cmd}{\regmap}{\opt}, \lay\lcomp \rfs}} = \Omega$, then:
    \[
      \Eval {\conf{\frame{\cmd}{\regmap}{\opt}:\frame \cmdtwo {\regmap'}{\opt}:\st, \lay\lcomp\rfs}} = \Omega
    \] 
  \end{varitemize}
\end{lemma}

\begin{proof}
  All cases follow directly from \Cref{lemma:contextpluggingtech}. The
  first three cases are omitted, as the most interesting one is the
  last.
  Expanding the definition of $ \Eval \cdot $, we obtain that $\Eval {\conf{\frame{\cmd}{\regmap}{\opt}, \lay\lcomp \rfs}} = \Omega$ is equivalent to stating that for every $\nat$,
  \[
    \exists \conftwo. \nstep \nat {\conf{\frame{\cmd}{\regmap}{\opt}, \lay\lcomp \rfs}} \conftwo.
    \tag{\dag}
  \]
  Our goal is to show that for every $\nat$:
  \[
    \exists \conftwo'. \nstep \nat {\conf{\frame{\cmd}{\regmap}{\opt}:\frame \cmdtwo {\regmap'}{\opt}:\st, \lay\lcomp \rfs}} {\conftwo'}
  \]
  Fix $ \nat $. The claim follows by applying (\dag) to $\nat$ and then using \Cref{lemma:contextpluggingtech} on the $\nat$-step reduction  $\nstep \nat {\conf{\frame{\cmd}{\regmap}{\opt}, \lay\lcomp \rfs}} \conftwo$.
\end{proof}

%
%
%


\subsection{Appendix for \Cref{sec:safety2}}
\label{sec:appsafety2}

\subsubsection{Speculative Semantics of $\Cmd$}

The speculative semantics of $\Cmd$ is in \Cref{fig:scen2sem1,fig:scen2sem1bis}.

\begin{figure*}
  \centering
  \columnwidth=\linewidth
  \begin{framed}
    \[
      \Infer[SI][Load-Step][\textsc{SLoad-Step}]
      {
        \lay \red \sframe{\frame{\cmemread[\lbl] \vx \expr\sep\cmd}{\regmap}{\opt}\cons\st}{\bmem}{\boolms}\cons\cfstack
        \sto{\dstep}{\omem \add} \sframe{\frame{\cmd}{\update \regmap x \val}{\opt} \cons \st}{\bmem}{\boolms}\cons\cfstack
      }
      {\toAdd{\sem\expr_{\regmap, \lay}} = \add &
        \bufread {\bm\buf\mem} \add 0 =\val, \bot &
        \add \in \underline \lay(\Ar[\opt]) &
        \fbox{$\opt = \km[\syscall] \Rightarrow \add \in \underline \lay(\caps(\syscall))$}
      }
    \]

    \resizebox{\textwidth}{!}{\(
      \Infer[SI][Load][\textsc{SLoad}]
      {\sstep
        {\sframe{\frame{\cmemread[\lbl] \vx  \expr\sep\cmd}{\regmap}{\opt}\cons \st} {\bm\buf\mem}{\boolms}\cons\cfstack}
        {\sframe{\frame{\cmd}{\update \regmap x \val}{\opt}\cons\st}{\bm\buf\mem}{\boolms\lor\bool'}\cons
          \sframe{\frame{\cmemread[\lbl] \vx  \expr\sep\cmd}{\regmap}{\opt}\cons\st}{\bm\buf\mem}{\boolms}\cons \cfstack}  {{{\dload[\lbl] i}}} {\omem \add}}
      {\toAdd{\sem\expr_{\regmap, \lay}} = \add &
        \bufread {\bm\buf\mem} \add i =\val, \bool' &
        \add \in \underline \lay(\Ar[\opt]) &
        \fbox{$\opt = \km[\syscall] \Rightarrow \add \in \underline \lay(\caps(\syscall))$}
      }
      \)
    }
    
    \[
      \Infer[SI][Load-Err][\textsc{SLoad-Error}]
      {\sstep
        {\sframe{\frame{\cmemread[\lbl] \vx \expr\sep\cmd}{\regmap}{\opt}\cons \st}{\bm\buf\mem}{\boolms}\cons\cfstack}
        {\sconf{\err, \boolms}\cons \cfstack}  {\dir} {\onone}}
      { \toAdd{\sem\expr_{\regmap, \lay}} = \add &
        \add \notin \underline \lay(\Ar[\opt]) &
        \dir \in \{\dstep, \dload[\lbl]\}
      }
    \]

    \[
      \Infer[SI][Load-Unsafe][\textsc{SLoad-Unsafe}]
      {\sstep
        {\sframe{\frame{\cmemread[\lbl] \vx \expr\sep\cmd}{\regmap}{\km[\syscall]}\cons \st}{\bm\buf\mem}{\boolms}\cons\cfstack}
        {{\unsafe}}   {\dir} {\omem \add}}
      {\toAdd{\sem\expr_{\regmap, \lay}} = \add &
        \add \in \underline \lay(\Ar[\km]) &
        \dir \in \{\dstep,\dload[\lbl]\} &
        \fbox{$\add \notin \underline \lay(\caps(\syscall))$}
      }
    \]
    \[
      \Infer[SI][Store][\textsc{SStore}]
      {
        \lay \red \sframe{\frame{\cmemass \expr \exprtwo\sep\cmd}{\regmap}{\opt}\cons\st}{\bm\buf\mem}{\boolms}\cons\cfstack
        \sto{\dstep}{\omem \add} \sframe{\frame{\cmd}{\regmap}{\opt} \cons \st}{\bm{\bitem \add {\sem \exprtwo_{\regmap,\lay}} \buf}\mem}{\boolms}\cons\cfstack
      }
      {\toAdd{\sem\expr_{\regmap, \lay}} = \add &
        \add \in \underline \lay(\Ar[\opt]) &
        \fbox{$\opt = \km[\syscall] \Rightarrow \add \in \underline \lay(\caps(\syscall))$}
      }
    \]

    \[
      \Infer[SI][Store-Err][\textsc{SStore-Error}]
      {\lay \red \sframe{\frame{\cmemass \expr \exprtwo\sep\cmd}{\regmap}{\opt}\cons\st}{\bm\buf\mem}{\boolms}\cons\cfstack
        \sto{\dstep}{\onone} 
        {\sconf{\err, \boolms}\cons \cfstack}}
      { \toAdd{\sem\expr_{\regmap, \lay}} = \add &
        \add \notin \underline \lay(\Ar[\opt]) 
      }
    \]

    \[
      \Infer[SI][Store-Unsafe][\textsc{SStore-Unsafe}]
      {\lay \red \sframe{\frame{\cmemass \expr \exprtwo\sep\cmd}{\regmap}{\km}\cons\st}{\bm\buf\mem}{\boolms}\cons\cfstack
        \sto{\dstep}{\onone} {\unsafe}}
      {\toAdd{\sem\expr_{\regmap, \lay}} = \add &
        \add \in \underline \lay(\Ar[\km]) &
        \fbox{$\add \notin \underline \lay(\caps(\syscall))$}
      }
    \]
    \resizebox{\textwidth}{!}{\(
      \Infer[SI][Call][\textsc{SCall}]{
        \lay \red \sframe{\frame{\ccall \expr {\exprtwo_1,\dots,\exprtwo_h}\sep\cmd}{\regmap}{\opt}\cons\st}{\bmem}{\boolms}\cons\cfstack
        \sto{\dstep}{\omem \add} \sframe{\frame{\mem(\add)} {\regmap_0[\vx_1,\dots,\vx_h\upd\sem{\exprtwo_1}_{\regmap,\lay}, \dots, \sem{\exprtwo_h}_{\regmap,\lay}]}{\opt}\cons\frame{\cmd}{\regmap}{\opt} \cons \st}{\bmem}{\boolms}\cons\cfstack
      }
      {\toAdd{\sem\expr_{\regmap, \lay}} = \add &
        \add \in \underline \lay(\Fn[\opt]) &
        \fbox{$\opt = \km[\syscall] \Rightarrow \add \in \underline \lay(\caps(\syscall))$}
      }
      \)}
    
    \[
      \Infer[SI][Call-Unsafe][\textsc{SCall-Unsafe}]{\sstep
        {\sframe{\frame{\ccall \expr{\vec \exprtwo}\sep\cmd}{\regmap}{\km[\syscall]}\cons\st}{\bm\buf\mem}{\boolms}\cons\cfstack}
        {\unsafe}
        {\dstep}
        {\ojump \add}}
      {\toAdd{\sem\expr_{\regmap, \lay}} = \add &
        \add \in \underline \lay(\Fn[\km]) &
        \fbox{$\add \notin \underline \lay(\caps(\syscall))$}
      }
    \]

    \[
      \Infer[SI][Call-Err][\textsc{SCall-Error}]{
        \lay \red \sframe{\frame{\ccall \expr {\exprtwo_1,\dots,\exprtwo_h}\sep\cmd}{\regmap}{\opt}\cons\st}{\bmem}{\boolms}\cons\cfstack
        \sto{\dstep}{\onone} \sconf{\err,\boolms}\cons\cfstack
      }
      {\toAdd{\sem\expr_{\regmap, \lay}} = \add &
        \add \notin \underline \lay(\Fn[\opt])
      }
    \]
    
    \resizebox{\textwidth}{!}{\(
      \Infer[SI][System-Call][\textsc{SSC}]{
        \lay \red \sframe{\frame{\csyscall \syscall {\exprtwo_1,\dots,\exprtwo_h}\sep\cmd}{\regmap}{\opt}\cons\st}{\bmem}{\boolms}\cons\cfstack
        \sto{\dstep}{\onone} \sframe{\frame{\syss(\add)} {\regmap_0[\vx_1,\dots,\vx_h\upd\sem{\exprtwo_1}_{\regmap,\lay}, \dots, \sem{\exprtwo_h}_{\regmap,\lay}]}{\km[\syscall]}\cons\frame{\cmd}{\regmap}{\opt} \cons \st}{\bmem}{\boolms}\cons\cfstack
      }
      { }
      \)}
    
    \[
      \Infer[SI][Pop][\textsc{SPop}]{\sstep
        {\sframe{\frame{\cnil}{\regmap}{\opt}\cons\frame{\cmd}{\regmap'}{\opt'} \cons \st}{\bm\buf\mem}{\boolms}\cons\cfstack}
        {\sframe{\frame{\cmd}{\update{\regmap'}{\ret}{\regmap(\ret)}}{\opt'} \cons \st}{\bm\buf\mem}{\boolms}\cons\cfstack}
        {\dstep}
        {\onone}}
      {}
    \]
  \end{framed}
  \caption{Speculative rules for $\Cmd$ and a system $\system=(\rfs, \syss, \caps)$, Part I.}
  \label{fig:scen2sem1}
\end{figure*}

\begin{figure*}[t]
  \centering
  \columnwidth=\linewidth
  \begin{framed}
    \ 
    
    \[
      \Infer[SI][Op][\textsc{SOp}]{\sstep
        {\sframe{\frame{\vx \ass \expr\sep\cmd}{\regmap}{\opt}\cons\st}{\bm{\buf}{\mem}}{\boolms}\cons \cfstack}
        {\sframe{\frame{\cmd}{\update \regmap x {\sem \expr_{\regmap, \lay}}}{\opt}\cons\st}{\bm{\buf}{\mem}}{\boolms}\cons \cfstack} {\dstep}{\onone}}{}
    \]

    \[
      \Infer[SI][Skip][\textsc{SSkip}]{\sstep
        {\sframe{\frame{\cskip\sep\cmd}{\regmap}{\opt}\cons \st}{\bm{\buf}{\mem}}{\boolms}\cons \cfstack}
        {\sframe{\frame{\cmd}{\regmap}{\opt}\cons\st}{\bm{\buf}{\mem}}{ \boolms}\cons \cfstack} {\dstep}{\onone}}{}
    \]
    
    \resizebox{\textwidth}{!}{\(
      \Infer[SI][Loop-Step][\textsc{SLoop}]
      {\sstep
        {\sframe{\frame{\cwhile[\lbl] \expr \cmdtwo\sep\cmd}{\regmap}{\opt}\cons\st}{\bm\buf\mem}{\boolms}\cons\cfstack}
        {\specconfone_{d} \cons\cfstack}
        {\dstep}{\obranch d}}
      {
        {\toBool{\sem\expr_{\regmap, \lay}}} = d &
        \specconfone_\ctrue =\sframe{\frame{\cmdtwo\sep\cwhile[\lbl] \expr \cmdtwo\sep\cmd}{\regmap}{\opt}\cons\st}{\bm{\buf}{\mem}}{\boolms} &
        \specconfone_\cfalse =\sframe{\frame{\cmd}{\regmap}{\opt}\cons\st}{\bm \buf \mem}{\boolms}
      }
    \)}\\
    
    \resizebox{\textwidth}{!}{\(
      \Infer[SI][Loop-Branch][\textsc{SLoop-Branch}]
      {\sstep
        {\sframe{\frame{\cwhile[\lbl] \expr \cmdtwo\sep\cmd}{\regmap}{\opt}\cons\st}{\bm{\buf}{\mem}}{\boolms}\cons\cfstack}
        {\specconfone_{d}\cons  \sframe{\frame{\cwhile[\lbl]\expr \cmdtwo\sep\cmd}{\regmap}{\opt}\cons\st}{\bm{\buf}{\mem}}{\boolms}\cons\cfstack}
        {\dbranch d}{\obranch {d}}
      }
      {
        \specconfone_\ctrue =\sframe{\frame{\cmdtwo\sep\cwhile[\lbl] \expr \cmdtwo\sep\cmd}{\regmap}{\opt}\cons\st}{\bm{\buf}{\mem}}{\boolms\lor(d\neq \toBool{\sem\expr_{\regmap, \lay}})} &
        \specconfone_\cfalse =\sframe{\frame{\cmd}{\regmap}{\opt}\cons\st}{\bm{\buf}{\mem}}{\boolms\lor(d\neq \toBool{\sem\expr_{\regmap, \lay}})}
      }
      \)}

    \resizebox{\textwidth}{!}{\(
      \Infer[SI][If][\textsc{SIf}]{\sstep
        {\sframe{\frame{\cif[\lbl] \expr {\cmd_\ctrue} {\cmd_\cfalse}\sep\cmd}{\regmap}{\opt}\cons\st}{\bm{\buf}{\mem}}{\boolms}\cons\cfstack}
        {\sframe{\frame{\cmd_{d}\sep\cmd}{\regmap}{\opt}\cons\st}{\bm{\buf}{\mem}}{ \boolms}\cons\cfstack} {{\dstep}}{\obranch{d}}}
      {\toBool{\sem \expr_{\regmap, \lay}}=d}
    \)}

    \[
      \Infer[SI][If-Branch][\textsc{SIf-Branch}]{\sstep
        {\specconfone\cons\cfstack}
        {\sframe{\frame{\cmd_{d}\sep\cmd}{\regmap}{\opt}\cons\st}{\bm\buf\mem}{\boolms\lor (d\neq\sem \expr_{\regmap, \lay})}\cons
          \specconfone\cons\cfstack}
        {{{\dbranch[\lbl] {d}}}}
        {\obranch d}}
      {\specconfone=\sframe {\frame{\cif[\lbl] \expr {\cmd_\top} {\cmd_\bot}\sep\cmd}{\regmap}{\opt}\cons\st}{\bm\buf\mem}{\boolms}}
    \]

    
    \[
      \Infer[SI][Backtrack-Top][\textsc{Bt}_{\top}]{
        \lay \red \specconfone \cons\cfstack \sto{\dbt}{\obt \top} \cfstack
      }{
        \specconfone = \sframe{\st}{\bm\buf\mem}{\top} \lor \specconfone = \sconf{\err,\top}
      }
      \quad
      \Infer[SI][Backtrack-Bot][\textsc{Bt}_{\bot}]{
        \lay \red \specconfone \cons\cfstack \sto{\dbt}{\obt \bot} \specconfone \cons \nil
      }{
        \specconfone = \sframe{\st}{\bm\buf\mem}{\bot} \lor \specconfone = \sconf{\err,\bot}
        & \cfstack \neq \nil
      }
    \]

    \[
      \Infer[SI][Fence]{\sstep
        {\sframe{\frame{\cfence\sep\cmd}{\regmap}{\opt}\cons\st, \bm\buf\mem, \bot}\cons \cfstack}
        {\sframe{\frame{\cmd}{\regmap}{\opt}\cons\st,(\nil, \overline {\bm\buf\mem}), \bot}\cons \cfstack} {\dstep}{\onone}}{}
    \]
  \end{framed}
  \caption{Speculative rules for $\Cmd$ and a system $\system=(\rfs, \syss, \caps)$, Part II.}
  \label{fig:scen2sem1bis}
\end{figure*}

\subsubsection{Semantics of $\SpCmd$} The semantics of $\SpCmd$ is in
\Cref{fig:scen2sem2,fig:scen2sem2bis,fig:scen2sem2bisbis}.

\begin{figure*}[t]
  \centering
  \columnwidth=\linewidth
  \begin{framed}
    \
    
    \[
      \Infer[AL][Poison]{
        \lay \red \aconf{\frame{\cpoison \dir\sep\adversary}{\regmap}{\opt}\cons\st }{\mem}{\Ds}{\Os}
        \ato \aconf{\frame{\adversary}{\regmap}{\opt}\cons\st}{\mem}{\dir:\Ds}{\Os}
      }{}
    \]

    \[
      \Infer[AL][Obs][\textsc{Observe}]{
        \lay \red \aconf{\frame{\vx \ass \eobs\sep\adversary}{\regmap}{\opt}\cons\st }{\mem}{\Ds}{\obs:\Os}
        \ato \aconf{\frame{\adversary}{\update \regmap\vx \obs}{\opt}\cons\st}{\mem}{\Ds}{\Os}
      }{}
    \]

    \[
      \Infer[AL][Obs-End][\textsc{Observe-End}]{
        \lay \red \aconf{\frame{\vx \ass \eobs\sep\adversary}{\regmap}{\opt}\cons\st }{\mem}{\Ds}{\nil}
        \ato \aconf{\frame{\adversary}{\update \regmap\vx \cnull}{\opt}\cons\st}{\mem}{\Ds}{\nil}
      }{}
    \]

    \[
      \Infer[AL][Spec-Init]{
        \lay \red \aconf{\frame{\cspec\cmd\sep\adversary}{\regmap}{\opt}\cons\st }{\mem}{\Ds}{\Os}
        \ato \hconf{\sframe{\frame{\cmd}{\regmap}{\opt}}{\mem}{\bot}}{\frame{\adversary}{\regmap}{\opt}\cons\st}{\Ds}{\Os}
      }{}
    \]

    \[
      \Infer[AL][Spec-Dir][\textsc{Spec-D}]{
        \lay \red \hconf{\cfstack}{\st}{\dir{\cons}\Ds}{\Os} \ato \hconf{\cfstack'}{\st}{\Ds}{\obs{\cons}\Os}
      }{
        \lay \red \cfstack \sto{\dir}{\obs} \cfstack'
      }
      \quad
      \Infer[AL][Spec-Step][\textsc{Spec-S}]{
        \lay \red \hconf{\cfstack}{\st}{\Ds}{\Os} \ato \hconf{\cfstack'}{\st}{\Ds}{\obs{\cons}\Os}
      }{
        \nf \cfstack \Ds &
        \lay \red \cfstack \sto{\dstep}{\obs} \cfstack'
      }
    \]

    \[
      \Infer[AL][Spec-Bt][\textsc{Spec-BT}]{
        \lay \red \hconf{\cfstack}{\st}{\Ds}{\Os} \ato \hconf{\cfstack'}{\st}{\Ds}{\obs{\cons}\Os}
      }{
        \nf \cfstack \Ds &
        \nf \cfstack \dstep &
        \lay \red \cfstack \sto{\dbt}{\obs} \cfstack'
      }
    \]
    
    \[
      \Infer[AL][Spec-Term]{
        \lay \red
        \hconf{\sframe{\frame{\cnil}{\regmap}{\opt}}{\bm\buf\mem}{\bot}}
              {\frame{\speccmd}{\regmap'}{\opt'}\cons\st}{\Ds}{\Os}
        \ato \aconf{\frame{\speccmd}{\regmap'}{\opt'}\cons\st}{\overline{\bm\buf\mem}}{\Ds}{\Os}
      }{}
    \]

    \[
      \Infer[AL][Spec-Error]{\esstep
        {\hconf{\sconf{\err,\bot}} \st \Ds \Os}
        {\err}
        {}}{}
      \quad
      \Infer[AL][Spec-Unsafe]{\esstep 
        {\hconf\unsafe\st \Ds \Os}
        {\unsafe}
        {}}{}
    \]
  \end{framed}
  \caption{Semantics of the non-standard constructs of $\SpCmd$ for the system $\system =(\rfs, \syss, \caps)$.}
  \label{fig:scen2sem2}
\end{figure*}

\begin{figure*}
  \centering
  \columnwidth=\linewidth
  \begin{framed}
    \[
      \Infer[AL][Load][\textsc{ALoad}]
      { \esstep
        {\ntce{\cmemread \vx \expr\sep\cmd}{\regmap}{\opt}{\st}{\mem}}
        {\ntce{\cmd}{\update{\regmap}{x}{\mem(\add)}}{\opt}{\st}{\mem}}
      }
      {
        \toAdd{\sem{\expr}_{\regmap, \lay}} = \add &
        \add \in \underline \lay(\Ar[\opt]) &
        \fbox{$\opt = \km[\syscall] \Rightarrow \add \in \underline \lay(\caps(\syscall))$}
      }
    \]
    
    \[
      \Infer[AL][Store][\textsc{AStore}]
      { \esstep
        {\ntce{\cmemass \expr \exprtwo\sep\cmd}{\regmap}{\opt}{\st}{\mem}}
        {\ntce{\cmd}{\regmap}{\opt}{\st}{\update{\mem}{\add}{\sem \exprtwo_{\regmap, \lay}}}}
      }
      {
        \toAdd{\sem{\expr}_{\regmap, \lay}} = \add &
        \add \in \underline \lay(\Ar[\opt]) &
        \fbox{$\opt = \km[\syscall] \Rightarrow \add \in \underline \lay(\caps(\syscall))$}
      }
    \]

    \[
      \Infer[AL][Load-Unsafe][\textsc{ALoad-Unsafe}]
      {\esstep
        {\ntce{\cmemread \vx \expr\sep\cmd}{\regmap}{\km[\syscall]}{\st}{\mem}}
        {\unsafe}
      }
      {
        \toAdd{\sem{\expr}_{\regmap, \lay}} = \add &
        \add \in \underline \lay(\Ar[\km]) &
        \fbox{$\add \not\in \underline \lay(\caps(\syscall))$}
      }
    \]
    
    \[
      \Infer[AL][Store-Unsafe][\textsc{AStore-Unsafe}]
      {\esstep
        {\ntce{\cmemass \expr \exprtwo\sep\cmd}{\regmap}{\km[\syscall]}{\st}{\mem}}
        {\unsafe}
      }
      {
        \toAdd{\sem{\expr}_{\regmap, \lay}} = \add &
        \add \in \underline \lay(\Ar[\km]) &
        \fbox{$\add \not\in \underline \lay(\caps(\syscall))$}
      }
    \]

    \[
      \Infer[AL][Load-Err][\textsc{ALoad-Error}]
      {\esstep
        {\ntce{\cmemread \vx \expr\sep\cmd}{\regmap}{\opt}{\st}{\mem}}
        {\err}
      }
      {\toAdd{\sem\expr_{\regmap, \lay}} = \add &
        \add \notin \underline \lay(\Ar[\opt])
      }
    \]
    
    \[
      \Infer[AL][Store-Err][\textsc{AStore-Error}]
      {\esstep
        {\ntce{\cmemass \expr \exprtwo\sep\cmd}{\regmap}{\opt}{\st}{\mem}}
        {\err}
      }
      {\toAdd{\sem\expr_{\regmap, \lay}} = \add &
        \add \notin \underline \lay(\Ar[\opt])
      }
    \]

    \[
      \Infer[AL][Skip][\textsc{ASkip}]
      { \esstep
        {\ntce{\cskip\sep\cmd}{\regmap}{\opt}{\st}{\mem}}
        {\ntce{\cmd}{\regmap}{\opt}{\st}{\mem}}
      }
      {
      }
    \]

    \[
      \Infer[AL][Fence][\textsc{AFence}]
      { \esstep
        {\ntce{\cfence\sep\cmd}{\regmap}{\opt}{\st}{\mem}}
        {\ntce{\cmd}{\regmap}{\opt}{\st}{\mem}}
      }
      {
      }
    \]

    \[
      \Infer[AL][Op][\textsc{AOp}]
      { \esstep
        {\ntce{\vx \ass \expr\sep\cmd}{\regmap}{\opt}{\st}{\mem}}
        {\ntce{\cmd}{\update{\regmap}{\vx}{\sem \expr_{\regmap, \lay}}}{\opt}{\st}{\mem}}
      }
      {
      }
    \]  
    
    \[
      \Infer[AL][If][\textsc{AIf}]
      { \esstep
        {\ntce{\cif \expr {\cmd_\ctrue} {\cmd_\cfalse}\sep\cmdtwo}{\regmap}{\opt}{\st}{\mem}}
        {\ntce{\cmd_{\toBool{\sem \expr_{\regmap, \lay}}}\sep \cmdtwo}{\regmap}{\opt}{\st}{\mem}}
      }
      {
      }
    \]
    
    \[
      \Infer[AL][While][\textsc{AWhile}]
      { \esstep
        {\ntce{\cwhile {\expr} {\cmd}\sep\cmdtwo}{\regmap}{\opt}{\st}{\mem}}
        {\confone_{\toBool{\sem \expr_{\regmap, \lay}}}}
      }
      {
        \confone_{\ctrue} = {\ntce{\cmd\sep\cwhile {\expr} {\cmd}\sep \cmdtwo}{\regmap}{\opt}{\st}{\mem}} &
        \confone_{\cfalse} = {\ntce{\cmdtwo}{\regmap}{\opt}{\st}{\mem}}     }
    \]  
  \end{framed}
  \caption{Semantics of standard construct of $\SpCmd$ for the system $\system =(\rfs, \syss, \caps)$, part I.}
  \label{fig:scen2sem2bis}
\end{figure*}

\begin{figure*}[t]
  \centering
    %
  %

  \begin{framed}
    \resizebox{\textwidth}{!}{\(
      \Infer[AL][Call][\textsc{ACall}]{
        \esstep
        {\ntce{\ccall{\expr}{\exprtwo_1,\dots,\exprtwo_n}\sep\cmd}{\regmap}{\opt}{\st}{\mem}}
        {
          \ntc
          {\mem(\add)}
          {\regmap_0[\vx_{1} \upd \sem{\exprtwo_1}_{\regmap,\lay}, \dots, \vx_{n} \upd \sem{\exprtwo_n}_{\regmap,\lay}]}
          {\opt}
          {\frame{\cmd}{\regmap}{\opt} : \st}
          {\mem}
        }
      }
      {\toAdd{\sem{\expr}_{\regmap, \lay}}=\add &
        \add \in \underline \lay(\Fn[\opt]) &
        \fbox{$\opt = \km[\syscall] \Rightarrow \add \in \underline \lay(\caps(\syscall))$}
      }
      \)}

    \[
      \Infer[AL][Call-Unsafe][\textsc{ACall-Unsafe}]{
        \esstep
        {\ntce{\ccall{\expr}{\exprtwo_1,\dots,\exprtwo_n}\sep\cmd}{\regmap}{\km[\syscall]}{\st}{\mem}}
        {\unsafe}
      }
      {
        \toAdd{\sem{\expr}_{\regmap, \lay}} = \add &
        \add \in \underline \lay(\Fn[\km]) &
        \fbox{$\add \not\in \underline \lay(\caps(\syscall))$}
      }
    \]

    \[
      \Infer[AL][Call-Err][\textsc{ACall-Error}]
      {\esstep
        {\ntce{\ccall \exprtwo {\expr_1, \ldots, \expr_n}\sep\cmd}{\regmap}{\opt}{\st}{\mem}}
        {\err}
      }
      {\toAdd{\sem\expr_{\regmap, \lay}} = \add &
        \add \notin \underline \lay(\Fn[\opt])
      }
    \]

    \resizebox{\textwidth}{!}{\(
      \Infer[AL][System-Call][\textsc{ASC}]{
        \esstep
        {\ntce{\csyscall{\syscall}{\exprtwo_1,\dots,\exprtwo_n}\sep\cmd}{\regmap}{\opt}{\st}{\mem}}
        {
          \ntce
          {\syss(\syscall)}
          {\regmap_0[\vx_{1} \upd \sem{\exprtwo_1}_{\regmap,\lay}, \dots, \vx_{n} \upd \sem{\exprtwo_n}_{\regmap,\lay}]}
          {\km[\syscall]}
          {\frame{\cmd}{\regmap}{\opt} : \st}
          {\mem}
        }
      }
      {
      }
      \)}

    \[
      \Infer[AL][Pop][\textsc{APop}]{
        \esstep
        {\ntce{\cnil}{\regmap}{\opt}{\frame {\cmd} {\regmap'} {\opt'}\cons \st}{\mem}}
        {\ntce{\cmd}{\update{\regmap'}{\ret} {\regmap(\ret)}}{\opt'}{\st}{\mem}}
      }
      {}
    \]
  \end{framed}
  \caption{Semantics of standard construct of $\SpCmd$ for the system $\system =(\rfs, \syss, \caps)$, part II.}
  \label{fig:scen2sem2bisbis}
\end{figure*}

\subsubsection{Buffered Memories}
\label{sec:bufmem}

Buffered memories have already been defined in \Cref{sec:safety2}.
In this section we just give the formal definitions of their lookup
and flushing operations following \cite{HighAssurance}. The operation of
lookup a value from a buffered memory is defined as follows:
\begin{align*}
  \bufread {\mem} \add k &\defsym \mem(a), \bot &  &\\
  \bufread {\bm{\bitem \add \nat \cons \buf}\mem} \add 0 &\defsym \nat, \bot &  &\\
  \bufread {\bm{\bitem \add \nat \cons \buf}\mem} \add {i+1} &\defsym \nat', \top &&\text{if }\bufread {\bm\buf\mem} \add i = \nat', b \\
  \bufread {\bm{\bitem {\add'} \nat \cons \buf}\mem} \add {i} &\defsym \bufread {\bm\buf\mem} \add i &&\text{if }\add\neq\add' 
\end{align*}
and a function for flushing buffers:
\begin{align*}
  \overline {\mem} &\defsym \mem \\
  \overline {\bm{\bitem \add \nat \cons \buf}\mem} &\defsym \update{\overline {\bm \buf \mem}} \add \nat.
\end{align*}
This function commits all the pending stores to the main memory. The
domain of a buffered memory is defined as follows: $\dom(\nil)=\emptyset$,
$\dom(\bitem \add\val\buf)=\{\add\}\cup\dom(\buf)$.

\begin{remark}
  \label{remark:onbuflookup}
  if $\bufread{\bm\buf\mem} \add i= \val, \bot$, then $\bufread{\bm\buf\mem} \add 0= \val, \bot$.  
\end{remark}
\begin{proof}
  the claim is:
  \[
    \forall \bm \buf\mem.\forall i.\forall \val.\bufread{\bm\buf\mem} \add i= \val, \bot \to \bufread{\bm\buf\mem} \add 0= \val, \bot
  \]
  By induction on the length of the buffer.
  \begin{proofcases}
    \proofcase{$\nil$} The claim comes from the definition of lookup.
    \proofcase{$\bitem {\add'} {\overline\val}\cons \buf$} The IH says:
    \[
      \forall i, \val.\bufread{\bm\buf\mem}\add i=\val, \bot \to \bufread{\bm\buf\mem}\add i=\bufread{\bm\buf\mem}\add 0
    \]
    and the claim is:
    \begin{equation*}
      \forall i, \val.\bufread{(\bitem {\add'} {\overline\val}\cons \mu, \mem)}\add i=\val, \bot \Rightarrow\\ \bufread{(\bitem {\add'} {\overline\val}\cons\mu, \mem)}\add i=\bufread{(\bitem {\add'} {\overline\val}\cons\mu, \mem)}\add 0
    \end{equation*}
    By cases on $i$.
    \begin{proofcases}
      \proofcase{$0$} The claim is:
      \begin{equation*}
        \forall \val.\bufread{\bm{\bitem {\add'} {\overline\val}\cons \mu}\mem}\add 0=\val, \bot \Rightarrow\\
        \bufread{\bm{\bitem {\add'} {\overline\val}\cons\mu} \mem}\add 0=\bufread{\bm{\bitem {\add'} {\overline\val}\cons\mu} \mem}\add 0
      \end{equation*}
      Observe that the conclusion is trivial.
      \proofcase{$i+1$} The claim is:
      \begin{equation*}
        \forall i.\forall \val.\bufread{\bm{\bitem {\add'} {\overline\val}\cons \mu}{\mem}}\add {i+1}=\val, \bot \Rightarrow\\
        \bufread{\bm{\bitem {\add'} {\overline\val}\cons \mu}{\mem}}\add {i+1}=\bufread{\bm{\bitem {\add'} {\overline\val}\cons \mu}\mem}\add 0
      \end{equation*}
      Fix $i, \val$, assume $\bufread{\bm{\bitem {\add'} {\overline\val}\cons \mu} \mem}\add {i+1}=\val, \bot$, call this assumption (H). The claim becomes:
      \[
        \bufread{\bm{\bitem {\add'} {\overline\val}\cons\mu} \mem}\add 0=\bufread{\bm{\bitem {\add'} {\overline\val}\cons\mu} \mem}\add 0
      \]
      Observe that it must be the case where $\add'\neq\add$, otherwise from (H) and the definition of lookup, we obtain $\bot=\top$. With this assumption, from (H) we deduce $\bufread{\bm\buf\mem}\add {i+1}=\val, \bot$, and we can rewrite the claim as follows:
      \[
        \bufread{\bm\buf\mem}\add {i+1}=\bufread{\bm\buf\mem}\add 0
      \]
      The claim is a consequence of the IH.
    \end{proofcases}
  \end{proofcases}
\end{proof}

\begin{remark}
  \label{rem:bufreadoverline}
  For every buffered memory $\bm\buf\mem$, and every address $\add$ we have that
  $\bufread{\bm\buf\mem}\add 0=\overline {\bm\buf\mem}(\add)$.
\end{remark}
\begin{proof}
  The proof goes by induction on $\buf$. If it is empty, then the claim is a
  trivial consequence of the definition of lookup.
  Otherwise, the claim is:
  \[
    \bufread{\bm{\bitem {\add'}\val:\buf}\mem}\add 0=\overline {\bm{\bitem {\add'}\val:\buf}\mem}(\add),
  \]
  that rewrites as follows:
  \[
    \bufread{\bm{\bitem {\add'}\val:\buf}\mem}\add 0=\update{\overline {\bm \buf\mem}}{\add'}\val(\add).
  \]
  If $\add=\add'$, the claim is a consequence of the definition of lookup and memory update. Otherwise,
  it is a consequence of the IH.
\end{proof}

\begin{remark}
  \label{rem:overlinewrtdom}
  For every buffered memory $\bm \buf{(\lay\lcomp \rfs)}$
  if $\dom(\buf) \subseteq \lay(\Ar)$,
  then we have that: $\overline {\bm \buf{(\lay\lcomp \rfs)}}
  = \lay \lcomp \rfs'$ for some $\rfs' \sim_{\Fn} \rfs$.
\end{remark}
\begin{proof}
  The proof is by induction on the length of the buffer. If it is 0,
  then the claim is trivial. Otherwise, it is a consequence of
  \Cref{rem:memupdtostupd}.
\end{proof}

\subsubsection{Omitted Proofs and Results}
\label{sec:proofs2}

In the following we will assume, without lack of generality that all
the memories $\mem$ within a configuration that is reached during
the evaluation of a configuration whose memory is $\lay\lcomp \rfs$
is such that $\mem =\lay\lcomp \rfs'$ for some $\rfs' \sim_{\Fn} \rfs$.
This is justified by \Cref{rem:simispreserved2,rem:simispreservedspec}.

\begin{proof}[Proof of \Cref{lemma:cttosafe}]
  Assume that a system call $\syscall$
  of a system $\system = (\syss, \rfs, \caps)$ is not \emph{speculative
    kernel safe}; this means that for some $\nat$ layout $\lay$,
  register map $\regmap$, buffered memory $(\buf, \lay \lcomp \rfs')$
  with $\rfs'\sim_{\Fn}\rfs$,
  sequence of directives $\Ds$, sequence of observations $\Os$,
  and mis-speculation flag $\boolms$
  we have that: 
  \[
    \nsstep \nat {\sframe{\frame{\syss(\syscall)}{\regmap}{\km[\syscall]}}{\bm{}\lay \lcomp \rfs'}\boolms} \unsafe \Ds \Os.
  \]
  By introspection on the rules, we deduce that the rule applied
  must be one among \ref{SI:Load-Unsafe},
  \ref{SI:Store-Unsafe}, \ref{SI:Call-Unsafe}. In all these cases,
  the rightmost observation within $\Os$ must be $\omem \add$
  for some address $\add \in \underline \lay(\Idk)$.
  More precisely, $\add$ belongs to $\underline \lay(\Fnk)$ if
  the rule was \ref{SI:Call-Unsafe} and to $\underline \lay(\Ark)$
  otherwise. In the following we just show this last case.
  From the definition of $\underline \lay(\Ark)$, we deduce that
  there are $\ar \in \Ark$ and $0\le i < \size \ar$ such that
  $\lay(\ar)+i = \add$. Our goal now, is to show that there is a
  layout $\lay'$ such that $p \notin
  \underline \lay'(\Idk)$, which means that $\add$ is not allocated
  in $\lay'$. To build $\lay'$, we go by cases on $\add' = \lay(\ar)$.
  \begin{proofcases}
    \proofcase{$\kappa_\um$} In this case, the array is stored
    at the beginning of the kernel-space address space.
    From the assumption on the size of this address space, there are
    at least $2\cdot\max_{\id \in \Idk} \size\id> \size \ar$ free
    addresses in the set $\{\kappa_\um+\size \ar, \dots,
    \kappa_\um+\kappa_{\km}-1\}$, so the array can be moved in
    this space, leaving the address $\add$ not allocated. We call
    $\lay'$ one such layout.
    \proofcase{$\kappa_{\km}-1-\size{\ar}$} Analogous to the case above.
    \proofcase{$\kappa_\um < \add' <\kappa_{\km}-1-\size{\ar}$} Due to
    the pigeonhole principle, in at least one of the
    address spaces $\{\kappa_\um, \dots, \add' -1\}$ and
    $\{\add'+\size \ar, \dots, \kappa_{\km}-1\size{\ar}\}$ there are
    at least $\max_{\id \in \Idk} \size\id> \size \ar$ not allocated
    address, this means that $\ar$ can be moved to one of these sub-spaces,
    leaving free the gap $\lay(\ar), \dots, \lay(\ar)+i$
    and in particular $\add$. We call $\lay'$ one such layout.
  \end{proofcases}
  From the \emph{speculative side-channel layout-non-interference}
  assumption, we deduce that there is $\cfstack'$ such that
  \[
    \nsstep[\system][\lay'] \nat {\sframe{\frame{\syss(\syscall)}{\regmap}{\km[\syscall]}}{\bm{\buf}{\lay' \lcomp \ars'}}{\boolms}} {\cfstack'} \Ds \Os,
  \]
  Observe that, in particular, this transition produces
  the sequence of observations $\Os$. By applying 
  \Cref{rem:noaddobs}, we deduce that $\omem \add$ does not
  appear in $\Os$, but this is absurd, because we showed that
  the rightmost observation of $\Os$ was exactly $\omem \add$.
\end{proof}

\begin{proof}[Proof of \Cref{thm:scenario2}]
  We fix a system $\system= (\rfs, \syss, \caps)$ and go by contraposition.
  We assume that there are an unprivileged command $\speccmd\in \SpCmd$,
  an initial register map $\regmap$, a number of steps $\nat$ and a layout such that:
  \[
    \nesstep \nat {\conf{\frame \speccmd \regmap \um, \lay \lcomp \rfs, \nil, \nil}} \unsafe.
  \]
  We first observe that $\nat \neq 0$, and by introspection of
  the rules of the semantics, we observe that the last rule must be
  one among \ref{AL:Load-Unsafe}, \ref{AL:Store-Unsafe}, \ref{AL:Call-Unsafe}
  and\\ \ref{AL:Spec-Unsafe}. We go by cases on these rules; in particular, the
  proof in the case of the first three rules is analogous, so we take the case
  of the rule \ref{AL:Load-Unsafe} as an example.
  \begin{proofcases}
    \proofcase{\ref{AL:Load-Unsafe}} By introspection of the rule,
    we deduce that there is a configuration
    \[
      \conf{\frame {\cmemread \vx \expr} {\regmap'}{\km[\syscall]}:\st, \mem, \Ds, \Os} 
    \]
    such that
    \small
    \begin{equation*}
      \nesstep \nat {\conf{\frame \speccmd \regmap \um, \lay \lcomp \rfs,\nil, \nil}}
      {\conf{\frame {\cmemread \vx \expr} {\regmap'}{\km[\syscall]}:\st, \mem, \Ds, \Os}}\\ \ato \unsafe.
    \end{equation*}
    \normalsize
    With an application of \Cref{lemma:thereisasyscall1}, we deduce that
    there is a configuration
    \[
      {\conf{\frame {\syss(\syscall)} {\regmap_0[\vx_1, \dots, \vx_k\upd \val_1, \dots, \val_k]}{\km[\syscall]}, \mem', \Ds', \Os'}},
    \]
     a prefix $\st'$ of $\st$ and $\nat'\in \Nat$ such that:
    \begin{equation*}
      \lay \red {\conf{\frame {\syss(\syscall)} {\regmap_0[\vx_1, \dots, \vx_k\upd \val_1, \dots, \val_k]}{\km[\syscall]}, \mem', \Ds', \Os'}} \\ \ato^{\nat'}
      {\conf{\frame {\cmemread \vx \expr} {\regmap'}{\km[\syscall]}:\st', \mem, \Ds, \Os}}
    \end{equation*}
    In the following, we write $\rfs' \sim_{\Fn} \rfs$, instead of $\mem'$
    as a consequence of \Cref{rem:simispreserved2}. 
    By introspection of the rule \ref{AL:Load-Unsafe} we deduce that:
    \begin{varitemize}
    \item $\sem \expr_{\regmap', \lay} \in \underline \lay(\Ark)$.
    \item $\sem \expr_{\regmap', \lay} \notin \underline \lay(\caps(\syscall))$.
    \end{varitemize}
    and this allows us to conclude that the same rule applies to
    \[
      {\conf{\frame {\cmemread \vx \expr} {\regmap'}{\km[\syscall]}:\st', \mem, \Ds, \Os}},
    \]
    thus showing that 
    \[
      {\conf{\frame {\syss(\syscall)} {\regmap_0[\vx_1, \dots, \vx_k\upd \val_1, \dots, \val_k]}{\km[\syscall]}, \lay \lcomp \rfs', \Ds', \Os'}}
    \]
    Reduces in $\nat$ steps to $\unsafe$.
    With an application of \Cref{lemma:ordinatytospec}, we deduce that
    also the speculative configuration
    \[
      {\sframe{\frame {\syss(\syscall)} {\regmap_0[\vx_1, \dots, \vx_k\upd \val_1, \dots, \val_k]}{\km[\syscall]}}{\lay \lcomp \rfs'}{\bot}}
    \]
    reduces in $\nat'+1$ steps to $\unsafe$
    using the sequence of directives ${\dstep^{\nat'+1}}$ and this contradicts
    \Cref{lemma:cttosafe} applied to the system call $\syscall$.
    \proofcase{\ref{AL:Spec-Unsafe}} By introspection of the rule,
    we deduce that there is a hybrid configuration
    \[
      {\hconf\unsafe {\frame{\adversary}{\regmap}{\opt}\cons\st}\Ds\Os}
    \]
    such that 
    \[
      \esstep {\hconf\unsafe {\frame{\adversary}{\regmap}{\opt}\cons\st}\Ds\Os} \unsafe
    \]
    and
    \[
      \nesstep \nat  {\conf{\frame \speccmd \regmap \um, \lay \lcomp \rfs, \nil, \nil}} 
      {\hconf\unsafe {\frame{\adversary}{\regmap}{\opt}\cons\st}\Ds\Os}
    \]
    With an application of \Cref{rem:thereisaspecon}, we deduce that
    there is a configuration
    \[
      {\sframe{\frame {\cmd} {\regmap'}{\um}}{\lay \lcomp \rfs'}{\bot}},
    \]
    with $\rfs' \sim_{\Fn}\rfs$, a sequence of directives
    $\Ds'$, a sequence of observations $\Os'$
    and a natural number $\nat' \le \nat$ such that:
    \begin{equation*}
      \nsstep {\nat'}
      {\sframe{\frame {\cmd} {\regmap'}{\um}}{\lay \lcomp \rfs'}{\bot}}
      \unsafe {\Ds'}{\Os'}.
    \end{equation*}
    Because of \Cref{lemma:nobt} we can assume without lack of generality
    that $\Ds'$ does not contain any $\dbt$ directive.
    From \Cref{lemma:thereisasyscall2},
    we deduce that there is configuration
    \begin{equation*}
      {\sframe{\frame {\syss(\syscall)} {\regmap''}{\km[\syscall]}}{\lay \lcomp \rfs'}{\bool}},
    \end{equation*}
    a sequence of directives $\Ds''$, a sequence of observations
    $\Os''$ and a store $\rfs''$ such that:
    \begin{equation*}
      \nsstep {\nat''}
      {\sframe{\frame {\syss(\syscall)} {\regmap''}{\km[\syscall]}}{\lay \lcomp \rfs''}{\bool}}
      \unsafe {\Ds''}{\Os''}.
    \end{equation*}
    but this is in contradiction with \Cref{lemma:cttosafe}
    applied to $\syscall$.
  \end{proofcases}
\end{proof}

\begin{lemma}
  \label{lemma:thereisasyscall1}
  For every system $\system = (\rfs, \syss, \caps)$,
  natural number $\nat$, configurations
  \[
    \confone = \conf{\frame {\speccmd}{\overline \regmap}\um, \overline \mem, \overline \Ds, \overline \Os}\]
  and
  \[
    \conf{\frame{\cmd}{\regmap}{\km[\syscall]}:\st_{\km}:\st_\um, \mem, \Ds, \Os}
  \]
  where $\speccmd$ is unprivileged, $\km(\st_{\km})$
  and $\um(\st_\um)$, such that
  \[
    \nesstep \nat \confone 
    {\conf{\frame{\cmd}{\regmap}{\km[\syscall]}:\st_{\km}:\st_\um,
          \mem, \Ds, \Os}},
  \]
  there is a third configuration
  \[
    \conf{\frame{\syss(\syscall)}{\regmap'}{\km[\syscall]}, \mem', \Ds, \Os},
  \]
  a natural number $\nat'$ such that
  \[
    \nesstep {\nat'}  
    {\conf{\frame{\syss(\syscall)}{\regmap'}{\km[\syscall]}, \mem', \Ds, \Os}}
    {\conf{\frame{\cmd}{\regmap}{\km[\syscall]}:\st_{\km}, \mem, \Ds, \Os}}.
  \]
\end{lemma}

\begin{proof}
  The proof goes by induction on $\nat$. The base case holds for vacuity of the premise
  The inductive case goes by cases on the rule that has been used to show the last transition.
  All the rules that do have as target a non-terminal classic configuration  
  except for \ref{AL:Pop} and \ref{AL:System-Call} share a similar behavior is the same,
  so we just show the case of loads:
  \begin{proofcases}
    \proofcase{\ref{AL:Load}} In this case, the assumption is that  
    \begin{equation*}
      \nesstep \nat \confone 
      {\conf{\frame{\cmemread \vx \expr\sep \cmd}{\regmap''}{\opt}:\st, \mem, \Ds, \Os}} \ato\\
      {\conf{\frame{\cmd}{\update{\regmap''}\vx{\sem\expr_{\regmap'', \lay}}}{\km[\syscall]}:\st_{\km}:\st_\um, \mem, \Ds, \Os}}
    \end{equation*}
    by introspection of the rule, we deduce  $\regmap =  \update{\regmap''}\vx{\sem\expr_{\regmap'', \lay}}$ and
    that flag of the source configuration is equal to that of the target one
    it must be $\km[\syscall]$ as well, so we can apply the IH. This shows that there is a configuration
    \[
      \conf{\frame{\syss(\syscall)}{\regmap'}{\km[\syscall]}, \mem', \Ds, \Os}
    \]
    a natural number $\nat'$ such that
    \begin{equation*}
      \nesstep {\nat'}  
      {\conf{\frame{\syss(\syscall)}{\regmap'}{\km[\syscall]}, \mem', \Ds, \Os}} {}\\
      {\conf{\frame{\cmemread \vx \expr\sep \cmd}{\regmap''}{\km[\syscall]}:\st_{\km}, \mem, \Ds, \Os}}.
    \end{equation*}
    Then we observe that the rule \ref{AL:Load} can be applied to the configuration
    \[
      {\conf{\frame{\cmemread \vx \expr\sep \cmd}{\regmap''}{\km[\syscall]}:\st_{\km}, \mem, \Ds, \Os}}
    \]
    to show the transition to
    \[
      {\conf{\frame{\cmd}{\update{\regmap''}\vx{\sem\expr_{\regmap'', \lay}}}{\km[\syscall]}:\st_{\km}, \mem, \Ds, \Os}}
    \]
    and this shows the claim.
    \proofcase{\ref{AL:System-Call}} In this case, the assumption is that  
    \begin{multline*}
      \nesstep \nat \confone 
      {\conf{\frame{\csyscall \syscall {\expr_1, \dots, \expr_k}\sep \cmd}{\regmap''}{\opt}:\st, \mem, \Ds, \Os}} \ato\\
      {\conf{\frame{\syss(\syscall)}{\regmap}{\km[\syscall]}: \frame{\cmd}{\regmap''}{\opt}:\st_{\km}:\st_\um, \mem, \Ds, \Os}}
    \end{multline*}
    From \Cref{rem:stackinv1}, we deduce that $\um(\st)$ and $\opt = \um$.
    This shows that the claim holds, in particular $\st_{\km}= \varepsilon$,
    $\st_{\um}= \st$, and it is easy to verify that
    \begin{equation*}
      \esstep
      {\conf{\frame{\csyscall \syscall {\expr_1, \dots, \expr_k}\sep \cmd}{\regmap''}{\opt}, \mem, \Ds, \Os}} {}\\
      {\conf{\frame{\syss(\syscall)}{\regmap}{\km[\syscall]}: \frame{\cmd}{\regmap''}{\opt}, \mem, \Ds, \Os}}
    \end{equation*}
    holds.
    \proofcase{\ref{AL:Pop}} In this case, the assumption is that  
    \begin{equation*}
      \nesstep \nat \confone 
      {\conf{\frame \cmdtwo {\regmap''}{\opt}:\frame{\cmd}{\regmap}{\km[\syscall]}:\st, \mem, \Ds, \Os}} \ato\\
      {\conf{\frame{\cmd}{\update{\regmap'}{\ret}{\regmap''(\ret)}}{\km[\syscall]}:\st, \mem, \Ds, \Os}},
    \end{equation*}
    With an application of \Cref{rem:stackinv1}, we rewrite $\st$ as $\st_{\km}:\st_\um$,
    and we deduce that $\opt = \km[\syscall]$, so we can apply the IH.
    It shows that there is a configuration
    \[
      \conf{\frame{\syss(\syscall)}{\regmap'}{\km[\syscall]}, \mem', \Ds, \Os}
    \]
    a natural number $\nat'$ such that
    \begin{equation*}
      \nesstep {\nat'}  
      {\conf{\frame{\syss(\syscall)}{\regmap'}{\km[\syscall]}, \mem', \Ds, \Os}} {}\\
      {\conf{\frame \cmdtwo {\regmap''}{\opt}:\frame{\cmd}{\regmap}{\km[\syscall]}:\st_{\km}, \mem, \Ds, \Os}}.
    \end{equation*}
    To conclude the proof it suffices to verify that
    \begin{equation*}
      \esstep  
      {\conf{\frame \cmdtwo {\regmap''}{\opt}:\frame{\cmd}{\regmap}{\km[\syscall]}:\st_{\km}, \mem, \Ds, \Os}} {} \\
      {\conf{\frame{\cmd}{\update{\regmap'}{\ret}{\regmap''(\ret)}}{\km[\syscall]}:\frame{\cmd}{\regmap}{\km[\syscall]}:\st_{\km}, \mem, \Ds, \Os}}.
    \end{equation*}
  \end{proofcases}
\end{proof}

\begin{lemma}
  \label{lemma:ordinatytospec}
  For every system $\system = (\overline \rfs, \syss, \caps)$,
  store $\rfs \sim_{\Fn}\overline \rfs)$,
  configuration $\conf{\frame{\cmd}{\regmap}{\um}, {\lay \lcomp \rfs}, \Ds, \Os}$,
  $\nat \in \Nat$,
  and every configuration $\conf{\st, {\lay \lcomp \rfs}, \Ds, \Os}$
  if
  \[
    \nesstep \nat {\conf{\frame{\cmd}{\regmap}{\km[\syscall]}, {\lay \lcomp \rfs}, \Ds, \Os}}
    {\conf{\st, {\lay \lcomp \rfs'}, \Ds, \Os}},
  \]
  then $\km[\syscall](\st)$ and there is a buffered memory $\bm \buf {(\lay \lcomp \rfs'')}$ and a sequence of observations $\Os$ such that 
  \[
    \nsstep \nat {\sframe {\frame{\cmd}{\regmap}{\km[\syscall]}} {\lay \lcomp \rfs} \bot} 
    {\sframe \st {\bm \buf {(\lay \lcomp \rfs'')}} \bot} {\dstep^\nat} \Os
  \]
  and $\overline {\bm\buf{(\lay \lcomp \rfs'')}} = {\lay \lcomp \rfs'}$.
\end{lemma}
\begin{proof}
  We go by induction on $\nat$.
  \begin{proofcases}
    \proofcase{0} Trivial.
    \proofcase{$\nat+1$} In this case, the premise is:
    \begin{equation*}
      \nesstep \nat {\conf{\frame{\cmd}{\regmap}{\km[\syscall]}, {\lay \lcomp \rfs}, \Ds, \Os}}
      {\conf{\frame {\cmd'} {\regmap'} {\opt'}: \st', {\lay \lcomp \rfs''}, \Ds, \Os}}\ato\\
      {\conf{\st,{\lay \lcomp \rfs'}, \Ds, \Os}}
    \end{equation*}
    We can apply The IH to the first $\nat$ steps. This shows that there are $\buf$ and $\Os$
    \[
      \nsstep \nat {\sframe {\frame{\cmd}{\regmap}{\km[\syscall]}} {\lay \lcomp \rfs} \bot} 
      {\sframe {\frame {\cmd'} {\regmap'} {\km[\syscall]}:\st'} {\bm {\buf} {(\lay \lcomp \omega)}} \bot} {\dstep^\nat} \Os
    \]
    such that $\overline {\bm\buf{(\lay \lcomp \omega)}} = {\lay \lcomp \rfs''}$ and $\km[\syscall]({\frame {\cmd'} {\regmap'} {\km[\syscall]}:\st'})$.
    We are required to show that
    \[
      \sstep {\sframe {\frame {\cmd'} {\regmap'} {\km[\syscall]}:\st'} {\bm {\buf} {(\lay \lcomp \omega)}} \bot}
      {\sframe{\st}{\bm{\buf'}{(\lay \lcomp \omega')}}{\bot}}{\dstep} \obs,
    \]
    that $\overline {\bm{\buf'}{(\lay \lcomp \omega')}} = {\lay \lcomp \rfs'}$, and that $\km[\syscall](\st)$ holds.
    The proof cases on the $\ato$ relation;
    many of the cases are similar with the others, so we just show
    the most important ones, also note that, in particular, cases \ref{AL:Poison},
    \ref{AL:Obs}, \ref{AL:Spec-Term}, \ref{AL:Spec-Error}, and \ref{AL:Spec-Unsafe}
    can be omitted because of the IH.
    \begin{proofcases}
      \proofcase{\ref{AL:Fence}} In this case, the assumption rewrites as follows:
      \begin{equation*}
        \esstep  {\conf {\frame{\cfence\sep \cmdtwo}{\regmap}{\km[\syscall]}:\st',  {\lay \lcomp \rfs''}, \Ds, \Os}}{} \\
        {\conf {\frame{\cmdtwo}{
              \regmap
            }{\km[\syscall]}:\st',  {\lay \lcomp \rfs''}, \Ds, \Os}}
      \end{equation*}
      and the goal is to show that there is an observation $\obs$, a buffer $\buf'$ and a store $\omega'$ such that:
      \begin{equation*}
        \sstep  {\sframe {\frame{\cmemread \vx \expr\sep \cmdtwo}{\regmap}{\km[\syscall]}:\st'}{\bm{\buf}{(\lay \lcomp \omega)}}{\bot}}{} \dstep \obs \\
        {\sframe {\frame{\cmdtwo}{
              \regmap
            }{\km[\syscall]}:\st'}{ \bm{\buf'}{(\lay \lcomp \omega')}}{\bot}}
      \end{equation*}
      and $\overline {\bm{\buf'}{(\lay \lcomp \omega)}} = {\lay \lcomp \rfs''}$. By applying the rule \ref{SI:Fence},
      suitable buffers and stores for the target configuration are
      $\buf'=\nil$ and $\omega'=\overline {\bm{\buf}{(\lay \lcomp \omega)}}$, and the transition produces the observation $\onone$; we need to observe that
      ${\lay \lcomp \rfs''}= \overline {\overline {\bm{\buf}{(\lay \lcomp \omega)}}}$, that is a consequence of the IH and of the definition of $\overline \cdot$ on memories.
      Finally, we
      must observe that $\km[\syscall]({\frame{\cmdtwo}{\regmap}{\km[\syscall]}:\st'})$, but this is a direct consequence of the IH,
          and of the fact that $\cmdtwo$ is a sub-term of $\cmemread \vx \expr\sep \cmdtwo$.
      \proofcase{\ref{AL:Load}} In this case, the assumption rewrites as follows:
      \begin{equation*}
        \esstep  {\conf {\frame{\cmemread \vx \expr\sep \cmdtwo}{\regmap}{\km[\syscall]}:\st',  {\lay \lcomp \rfs''}, \Ds, \Os}}{} \\
        {\conf {\frame{\cmdtwo}{
              \update\regmap\vx{{\lay \lcomp \rfs''}(\toAdd{\sem \expr_{\regmap, \lay}})}
            }{\km[\syscall]}:\st',  {\lay \lcomp \rfs''}, \Ds, \Os}}
      \end{equation*}
      and the goal is to show that there is an observation $\obs$ and a buffer $\buf'$ such that:
      \begin{equation*}
        \sstep  {\sframe {\frame{\cmemread \vx \expr\sep \cmdtwo}{\regmap}{\km[\syscall]}:\st'}{\bm{\buf}{(\lay \lcomp \omega)}}{\bot}}{} \dstep \obs \\
        {\sframe {\frame{\cmdtwo}{
              \update\regmap\vx{{\lay \lcomp \rfs''}(\toAdd{\sem \expr_{\regmap, \lay}})}
            }{\km[\syscall]}:\st'}{ \bm{\buf'}{(\lay \lcomp \omega')}}{\bot}}
      \end{equation*}
      and $\overline {\bm{\buf'}{(\lay \lcomp \omega)}} = {\lay \lcomp \rfs''}$. We choose $\buf'=\buf$ and $\omega=\omega'$ and the equality is a consequence of the IH. The only rule that can apply is \ref{SI:Load-Step}; with this choice the observation produced is $\omem {\toAdd{\sem \expr_{\regmap, \lay}}}$; we need to observe that
      ${\lay \lcomp \rfs''}(\toAdd{\sem \expr_{\regmap, \lay}})= \bufread {\bm{\buf'}{(\lay \lcomp \omega)}} {\toAdd{\sem \expr_{\regmap, \lay}}} 0$, that is a consequence of \Cref{rem:bufreadoverline}. Finally, we
      must observe that
      \[
        \km[\syscall]({\frame{\cmdtwo}{
              \update\regmap\vx{{\lay \lcomp \rfs''}(\toAdd{\sem \expr_{\regmap, \lay}})}
            }{\km[\syscall]}:\st'}),
        \]
        but this is a direct consequence of the IH,
          and of the fact that $\cmdtwo$ is a sub-term of $\cmemread \vx \expr\sep \cmdtwo$.
      \proofcase{\ref{AL:Store}} In this case, the assumption rewrites as follows:
      \begin{equation*}
        \esstep {\conf {\frame{\cmemass\expr\exprtwo\sep \cmdtwo}{\regmap}{\km[\syscall]}:\st',  {\lay \lcomp \rfs''}, \Ds, \Os}}{} \\
        {\conf {\frame{\cmdtwo}{\regmap}
            {\km[\syscall]}:\st',  \update{{\lay \lcomp \rfs''}} {\toAdd{\sem \expr_{\regmap, \lay}}} {\sem \exprtwo_{\regmap, \lay}}, \Ds, \Os}}
      \end{equation*}
      and the goal is to show that there is an observation $\obs$ and a buffer $\buf'$ such that:
      \begin{equation*}
        \sstep  {\sframe {\frame{\cmemass \vx \expr\sep \cmdtwo}{\regmap}{\km[\syscall]}:\st'}{\bm{\buf}{(\lay \lcomp \omega)}}{\bot}}{} \dstep \obs \\
        {\sframe {\frame{\cmdtwo}{
              \regmap
            }{\km[\syscall]}:\st'}{ \bm{\buf'}{(\lay \lcomp \omega)}}{\bot}}
      \end{equation*}
      and $\overline {\bm{\buf'}{(\lay \lcomp \omega)}} = \update{{\lay \lcomp \rfs''}} {\toAdd{\sem \expr_{\regmap, \lay}}} {\sem \exprtwo_{\regmap, \lay}}$. We observe that the rule \ref{SI:Store} applies and produces a target configuration that matches the one we are looking for, in particular, the buffer it produces is
      \[
        \bitem{\toAdd{\sem \expr_{\regmap, \lay}}} {\sem \exprtwo_{\regmap, \lay}} \cons \buf.
      \]
      We observe that the conclusion
      \[
        \overline {\bm{\bitem{\toAdd{\sem \expr_{\regmap, \lay}}} {\sem \exprtwo_{\regmap, \lay}} \cons \buf}{(\lay \lcomp \omega)}}= \update{{\lay \lcomp \rfs''}} {\toAdd{\sem \expr_{\regmap, \lay}}} {\sem \exprtwo_{\regmap, \lay}}
      \]
      comes from the rewriting of the function $\overline \cdot$ and form the assumption
      $\overline {\bm \buf {(\lay \lcomp \omega)}} = {\lay \lcomp \rfs''}$.
      \proofcase{\ref{AL:Call}} In this case, we observe that this rule and \ref{SI:Call}
      share the same premises, so also the second one can be applied. By introspection of these rules,
      we deduce that said $\cmd'= \ccall \exprtwo {\expr_1, \dots, \expr_k}\sep \cmdtwo$, the target configurations
      are respectively
      \[
        {\conf{\frame {{\lay \lcomp \rfs''}(\sem \exprtwo_{\regmap', \lay})} {\regmap_0'} {\km[\syscall]}: \frame {\cmdtwo} {\regmap'} {\km[\syscall]}: \st', {\lay \lcomp \rfs''}, \Ds, \Os}}
      \]
      and
      \[
        {\sframe{\frame {{\lay \lcomp \rfs''}(\sem \exprtwo_{\regmap', \lay})} {\regmap_0'} {\km[\syscall]}: \frame {\cmdtwo} {\regmap'} {\km[\syscall]}: \st'}{\mu{(\lay \lcomp \omega)}}\bot},
      \]
      where $\regmap_0'=\update{\regmap_0}{\vx_1,\dots, \vx_k}{\sem {\expr_1}{\regmap', \lay}, \dots,  {\expr_k}{\regmap', \lay}}$. For this reason, the conclusion on the buffered memory is a consequence of the IH, and we can deduce
      \[
        \km[\syscall](\frame {{\lay \lcomp \rfs''}(\sem \exprtwo_{\regmap', \lay})} {\regmap_0'} {\km[\syscall]}: \frame {\cmdtwo} {\regmap'} {\km[\syscall]}: \st')
      \]
      from the IH and the premises of the rule \ref{AL:Call}, which guarantee  $\sem \exprtwo_{\regmap', \lay}\in \underline \lay(\Fnk)$; this, in turn, has as consequence that there is $\fn \in \Fnk$ such that $\lay(\fn)=\sem \exprtwo_{\regmap', \lay}$. Thus, by definition of $\cdot \lcomp \cdot$, we conclude that
      \[
        {\lay \lcomp \rfs''}(\sem \exprtwo_{\regmap', \lay}) =\omega(\fn) = \overline \rfs(\fn).
      \]
      For this reason, it suffices to observe that $\km[\syscall](\overline \rfs(\fn))$ holds by definition of system.
    \end{proofcases}
  \end{proofcases}
\end{proof}

\begin{lemma}
  \label{lemma:thereisasyscall2}
  For every system $\system =(\rfs, \syss, \caps)$,
  configuration
  \[
    \sframe {\frame \cmd \regmap \um} {{(\bm\buf\mem)}} {\bot}
  \]
  such that $\um(\cmd)$
  speculative stack $\cfstack = \sframe {\frame \cmd {\regmap'} {\km[\syscall]}:\st'} {{(\bm\buf\mem)}'} {\boolms}:\cfstack'$,
  sequence of directives $\Ds$ without $\dbt$ directives,
  sequence of observations $\Os$
  and layout $\lay$
  such that
  \[
    \nsstep \nat {\sframe {\frame \cmd \regmap \um} {{(\bm\buf\mem)}} {\bot}} \cfstack \Ds \Os,
  \]
  there is a configuration $\sframe {\frame {\syss(\syscall)} {\regmap_0[\vx_0, \dots, \vx_h \upd \val_1, \dots, \val_h]} {\km[\syscall]}:\st''} {{(\bm\buf\mem)}''} {\boolms}$
  a sequence of directives $\Ds'$, a sequence of observations $\Os'$,
  and a natural number $\nat'\le \nat$ such that: 
  \begin{equation*}
    \nsstep {\nat'} {\sframe {\frame {\syss(\syscall)} {\regmap''} {\km[\syscall]}} {{(\bm\buf\mem)}''} {\boolms''}}{}{\Ds'}{\Os'}\\   {\sframe {\frame \cmd {\regmap'} {\km[\syscall]}:\overline \st} {{(\bm\buf\mem)}'} {\boolms'}:\overline \cfstack},
  \end{equation*}
  for some $\overline \cfstack$, where in particular
  $\overline \st$ is a prefix of $\st'$ such that $\km[\syscall](\frame \cmd {\regmap'} {\km[\syscall]}:\overline \st)$ and there is a stack $\st''$ such that
  $\um(\st'')$ and $\st'=\overline \st:\st''$.
\end{lemma}

\begin{proof}
  By induction on $\nat$.
  \begin{proofcases}
    \proofcase{0} Holds by vacuity of the premise.
    \proofcase{$\nat+1$} The premise rewrites as
    \[
      \nsstep \nat {\sframe {\frame \cmd \regmap \um} {{(\bm\buf\mem)}} {\bot}} T \Ds \Os \sto \dir  \obs \cfstack.
    \]
    We go by cases on the rule that has been applied to show the last transition. Most of these cases are similar to the others; for this reason, we just show some of the most interesting ones. In particular, since $\Ds$ does not contain any $\dbt$ directive, we can assume without lack of generality that the rules for backtracking are not employed. 
    \begin{proofcases}
      \proofcase{\ref{SI:Op}}
        In this case we can assume that
        \[
          T = \sframe {\frame {\vx \ass \expr\sep\cmdtwo} {\regmap'''} {\opt}:\st'''} {{(\bm\buf\mem)'''}} {\boolms'''}: T'.
          \tag{\dag}
        \]
        By introspection of the rule, we deduce that
        the flag of the target configuration is equal to that of the source
        configuration, so we deduce $\opt = \km[\syscall]$, this also means that
        we can apply the IH on the first $\nat$ steps and deduce that 
        \begin{equation*}
          \nsstep {\nat'} {\sframe {\frame {\syss(\syscall)} {\regmap''} {\km[\syscall]}} {{(\bm\buf\mem)}''} {\boolms''}}{}{\Ds'}{\Os'}\\
          {\sframe {\frame {\vx \ass \expr\sep\cmdtwo} {\regmap'''} {\opt}:\overline \st'''} {{(\bm\buf\mem)'''}} {\boolms'''}: \overline T'},
        \end{equation*}
        $\overline \st'''$ is a suffix of $\st'''$
        such that $\km[\syscall](\overline \st''')$ and there is a stack $G$ such that
        $\um(G)$ and $\st'''=\overline \st''':G$.
        Then, we observe that 
        \begin{equation*}
          \sstep {\sframe {\frame {\vx \ass \expr\sep\cmdtwo} {\regmap'''} {\opt}:\overline \st'''} {{(\bm\buf\mem)'''}} {\boolms'''}: \overline T'}{}  \dstep\onone\\
          {\sframe {\frame {\cmdtwo} {\update{\regmap'''}\vx {\sem\exprtwo_{\regmap''', \lay}}} {\opt}:\overline \st'''} {{(\bm\buf\mem)'''}} {\boolms'''}: \overline T'},
        \end{equation*}
        and that:
        \begin{equation*}
          \sstep {\sframe {\frame {\vx \ass \expr\sep\cmdtwo} {\regmap'''} {\opt}:\st'''} {{(\bm\buf\mem)'''}} {\boolms'''}: T'}{}  \dstep\onone\\
          {\sframe {\frame {\cmdtwo} {\update{\regmap'''}\vx {\sem\exprtwo_{\regmap''', \lay}}} {\opt}:\st'''} {{(\bm\buf\mem)'''}} {\boolms'''}: T'},
        \end{equation*}
        The conclusions that we are required to show are a direct consequence of the IH,
        and of the observation that $\km[\syscall](\vx \ass \expr\sep\cmdtwo)$ has as
        consequence $\km[\syscall](\cmdtwo)$.

        \proofcase{\ref{SI:Load}}
        In this case we can assume that
        \[
          T = \sframe {\frame {\cmemread[\lbl] \vx \expr\sep\cmdtwo} {\regmap'''} {\opt}:\st'''} {{(\bm\buf\mem)'''}} {\boolms'''}: T'.
          \tag{\dag}
        \]
        By introspection of the rule, we deduce that
        the execution mode flag of the target configuration is equal to that of the source
        configuration, so we deduce $\opt = \km[\syscall]$, this also means that
        we can apply the IH on the first $\nat$ steps and deduce that 
        \begin{equation*}
          \nsstep {\nat'} {\sframe {\frame {\syss(\syscall)} {\regmap''} {\km[\syscall]}} {{(\bm\buf\mem)}''} {\boolms''}}{}{\Ds'}{\Os'}\\
          {\sframe {\frame {\cmemread[\lbl] \vx \expr\sep\cmdtwo} {\regmap'''} {\opt}:\overline \st'''} {{(\bm\buf\mem)'''}} {\boolms'''}: \overline T'},
        \end{equation*}
        $\overline \st'''$ is a suffix of $\st'''$
        such that $\km[\syscall](\overline \st''')$ and there is a stack $G$ such that
        $\um(G)$ and $\st'''=\overline \st''':G$.
        Then, we call $\add$ the value of $\toAdd{\sem \expr_{\regmap''', \lay}}$ and
        $(\val, \bool)$ the pair that is returned by
        $\bufread{{(\bm\buf\mem)'''}} {\add} {i}$. Then, we observe that 
        \begin{multline*}
          \sstep {\sframe {\frame {\cmemread[\lbl] \vx \expr\sep\cmdtwo} {\regmap'''} {\opt}:\st'''} {{(\bm\buf\mem)'''}} {\boolms'''}: T'}{}  {\dload[\lbl]{i}}{\omem \add}\\
          \sframe {\frame {\cmdtwo} {\update{\regmap'''}\vx {\val}} {\opt}: \st'''} {{(\bm\buf\mem)'''}} {\boolms'''\lor \bool} :\\
          \sframe {\frame {\cmemread[\lbl] \vx \expr\sep\cmdtwo} {\regmap'''} {\opt}: \st'''} {{(\bm\buf\mem)'''}} {\boolms'''}:  T',
        \end{multline*}
        and that:
        \begin{multline*}
          \sstep {\sframe {\frame {\cmemread[\lbl] \vx \expr\sep\cmdtwo} {\regmap'''} {\opt}:\overline \st'''} {{(\bm\buf\mem)'''}} {\boolms'''}: \overline T'}{}  {\dload[\lbl] i}{\omem \add}\\
          \sframe {\frame {\cmdtwo} {\update{\regmap'''}\vx {\val}} {\opt}: \overline \st'''} {{(\bm\buf\mem)'''}} {\boolms'''\lor \bool} :\\
          \sframe {\frame {\cmemread[\lbl] \vx \expr\sep\cmdtwo} {\regmap'''} {\opt}: \overline \st'''} {{(\bm\buf\mem)'''}} {\boolms'''}:  \overline T',
        \end{multline*}
        The conclusions that we are required to show are a direct consequence of the IH and of the observation that and of the observation that $\km[\syscall](\cmemread \vx \expr\sep\cmdtwo)$ has as
        consequence $\km[\syscall](\cmdtwo)$.

        \proofcase{\ref{SI:Call}} In this case we can assume that $T$ is
        \[
          \sframe {\frame {\ccall \expr {\exprtwo_1, \dots, \exprtwo_k}\sep\cmdtwo} {\regmap'''} {\opt}:\st'''} {{(\bm\buf\mem)'''}} {\boolms'''}: T'.
        \]
        By introspection of the rule, and by knowing that the execution flag of the
        target configuration is $\km[\syscall]$, we deduce that this must be the case
        also for the source configuration, so  $\opt = \km[\syscall]$, this also means that
        we can apply the IH on the first $\nat$ steps and deduce that 
        \begin{equation*}
          \nsstep {\nat'} {\sframe {\frame {\syss(\syscall)} {\regmap''} {\km[\syscall]}} {{(\bm\buf\mem)}''} {\boolms''}}{}{\Ds'}{\Os'}\\
          {\sframe {\frame {\ccall \expr {\exprtwo_1, \dots, \exprtwo_k}\sep\cmdtwo} {\regmap'''} {\opt}:\overline \st'''} {{(\bm\buf\mem)'''}} {\boolms'''}: \overline T'},
        \end{equation*}
         $\overline \st'''$ is a suffix of $\st'''$
        such that $\km[\syscall](\overline \st''')$ and there is a stack $G$ such that
        $\um(G)$ and $\st'''=\overline \st''':G$.
        Then, we call $\add$ the value of $\toAdd{\sem \expr_{\regmap''', \lay}}$ and
        we observe that form the premise of the rule and the definition of $\underline\lay$,
        we can deduce that there is $\fn \in \Fn[\opt]$ such that
        $\lay(\fn) = \add$. From this observation, and \Cref{rem:simispreserved2},
        we deduce that $(\bm\buf\mem)''' = \bm{\buf'}{\lay\lcomp \rfs'}$
        for a $\rfs'\sim_{\Fn}\rfs$
        So, from the definition of $\cdot\lcomp\cdot$ and these observations we conclude
        that the executed procedure is exactly $\rfs(\fn)$.
        We also call $\overline \regmap$ the register map that
        is obtained by evaluating the semantics of the arguments in
        $\regmap'''$ and updating the argument registers of $\regmap_0$ with these values.
        Then, 
        by introspection of the rule, we observe that
        \small
        \begin{multline*}
          \sstep {\sframe {\frame {\ccall \expr {\exprtwo_1, \dots, \exprtwo_k}\sep\cmdtwo} {\regmap'''} {\opt}: \st'''} {{(\bm\buf\mem)'''}} {\boolms'''}:  T'}{}  {\dstep}{\ojump \add}\\
          {\sframe {\frame {\rfs(\fn)} {{\overline \regmap}} {\opt}:\frame {\cmdtwo} {{\regmap'''}} {\opt}: \st'''} {{(\bm\buf\mem)'''}} {\boolms'''}:  T'},
        \end{multline*}
        \normalsize
        and that also:
        \small
        \begin{multline*}
          \sstep {\sframe {\frame {\ccall \expr {\exprtwo_1, \dots, \exprtwo_k}\sep\cmdtwo} {\regmap'''} {\opt}:\overline \st'''} {{(\bm\buf\mem)'''}} {\boolms'''}: \overline T'}{}  {\dstep}{\ojump \add}\\
          {\sframe {\frame {\rfs(\fn)} {{\overline \regmap}} {\opt}:\frame {\cmdtwo} {{\regmap'''}} {\opt}:\overline \st'''} {{(\bm\buf\mem)'''}} {\boolms'''}: \overline T'}.
        \end{multline*}
        \normalsize
        can be shown. The conclusions that we are required to show are a direct consequence of the IH,
        and of the observation that $\km[\syscall](\rfs(\fn))$ holds by definition of $\rfs$
        because $\fn \in \Fnk$.

        \proofcase{\ref{SI:System-Call}} In this case we can assume that $T$ is
        \[
          \sframe {\frame {\csyscall \syscalltwo {\exprtwo_1, \dots, \exprtwo_k}\sep\cmdtwo} {\regmap'''} {\opt}:\st'''} {{(\bm\buf\mem)'''}} {\boolms'''}: T'.
        \]
        By introspection of the rule and the target configuration, we deduce that
        $\syscalltwo = \syscall$; and that
        $\cfstack$ has the following shape:
        \[
          \sframe {\frame {\syss(\syscall)}{\regmap_0'}{\km[\syscall]}:\frame {\cmdtwo} {\regmap'''} {\opt}:\st'''} {{(\bm\buf\mem)'''}} {\boolms'''}: T'.
        \]
        where $\regmap_0'$ is obtained by updating the argument
        registers of $\regmap_0$ with
        the evaluation of $\expr_1, \dots, \expr_k$. 
        To show the claim, it suffices to set $\nat'=0$, $\Ds'=\nil$, $\Os=\nil$,
        $\overline \st = \varepsilon$, $\st''=\st'$, and
        the observation $\km[\syscall]({\frame {\syss(\syscall)}{\regmap_0'}{\km[\syscall]}})$
        holds for definition of $\syss$, while $\um({\frame {\cmdtwo} {\regmap'''} {\opt}:\st'''})$ is a consequence of \Cref{rem:stackinv2}.
        In particular, we can refuse the assumption that $\op=\km[\syscalltwo]$, because in such case we could not have
        a system call invocation as a command. Then, from the invariant on the composition on the stack,
        we deduce $\um({\csyscall \syscall{\expr_1, \dots,\expr_k}\sep \cmdtwo})$.
        \proofcase{\ref{SI:Pop}} In this case we can assume that $T$ is
        \[
          \sframe {\frame {\cnil} {\regmap'''} {\opt}:\st'''} {{(\bm\buf\mem)'''}} {\boolms'''}: T'.
        \]
        And by introspection of the rule, we deduce that
        \[
          \st'''= \frame \cmd {\update{\regmap'}\ret\val} {\km[\syscall]}:\st'
        \]
        for some $\val$;
        from this observation and
        \Cref{rem:stackinv2}, we deduce that $\opt=\km[\syscall]$,
        so we can apply the IH to the first $\nat$ steps and obtain that
        \begin{equation*}
          \nsstep {\nat'} {\sframe {\frame {\syss(\syscall)} {\regmap''} {\km[\syscall]}} {{(\bm\buf\mem)}''} {\boolms''}}{}{\Ds'}{\Os'}\\
          {\sframe {\frame {\cnil} {\regmap'''} {\opt}:\overline \st'''} {{(\bm\buf\mem)'''}} {\boolms'''}: \overline T'},
        \end{equation*}
         $\overline \st'''$ is a suffix of $\st'''$
        such that $\km[\syscall](\overline \st''')$ and there is a stack $G$ such that
        $\um(G)$ and $\st'''=\overline \st''':G$. This shows that, in particular,
        the first frame of $\overline \st'''$ must also be
        $\frame \cmd {\update{\regmap'}\ret\val} {\km[\syscall]}$.
        Otherwise, it either would not be a prefix of $\st'''$ or its
        concatenation with $G$ (observe that $\um(G)$) would not be
        $\st'''$.
         Thanks to this observation, 
         by introspection of the rule, we observe that 
        \begin{equation*}
          \sstep {\sframe {\frame {\cnil} {\regmap'''} {\opt}:\st'''} {{(\bm\buf\mem)'''}} {\boolms'''}: T'}{}  {\dstep}{\onone}\\
          {\sframe {\frame \cmd {\regmap'} {\km[\syscall]}:\st'} {{(\bm\buf\mem)'''}} {\boolms'''}:  T'},
        \end{equation*}
        and that also:
        \begin{equation*}
          \sstep {\sframe {\frame {\cnil} {\regmap'''} {\opt}:\overline \st'''} {{(\bm\buf\mem)'''}} {\boolms'''}: \overline T'}{}  {\dstep}{\onone}\\
          {\sframe {\frame \cmd {\regmap'} {\km[\syscall]}:\st'} {{(\bm\buf\mem)'''}} {\boolms'''}:  T'},
        \end{equation*}
        This concludes the proof.
      \end{proofcases}    
  \end{proofcases}
\end{proof}

\subsubsection{On the fencing transformation }
\label{sec:proofs3}

\begin{proof}[Proof of \Cref{thm:mitigation}]
  This result is a direct consequence of \Cref{lemma:fencesempres,lemma:fencesafeimpo}
  and of \Cref{prop:mitigation}.
\end{proof}

We extend the non speculative semantics of our language
by defining the semantics of the
$\cfence$ instruction as follows:

\[
  \Infer[WL][Fence]
  { \step
    {\ntc{\cfence\sep\cmd}{\regmap}{\opt}{\st}{\mem}}
    {\ntc{\cmd}{\regmap}{\opt}{\st}{\mem}}
  }
  {
  }
\]

\newcommandx{\fwf}[4][1=\system,2=\lay,3=\syscall]{#1\vdash \mathtt{\kwd f}\mathsf{wf}_{#2, #3}(#4)}

In addition, in order to facilitate the proof of the
forthcoming results, we introduce the predicate
of $\fencetrans$-\emph{well-formedness}.
Given a system $\system=(\rfs, \syss, \caps)$, such predicate is defined 
 as follows:

 \[
  \infer{\fwf {\rfs'}}{\forall \fn\in\Fnk. \exists \cmd. \rfs'(\fn)=\fencetrans(\cmd) & {\rfs} \sim_{\Fn} \rfs'}
\]

\[
  \infer{\fwf{\frame \cmd \regmap {\km[\syscall]}:\st}}{\fwf{\st} & \exists \cmd'.\cmd=\fencetrans (\cmd')}
\]
\[
  \infer{\fwf \varepsilon}{} \quad
  \infer{\fwf\err}{}\quad
  \infer{\fwf\unsafe}{}
\]

\[
  \infer{\fwf{\sframe\st{\bm\buf(\lay\lcomp \rfs')}\boolms}}{\fwf {\rfs'} & \fwf\st & \dom(\buf) \subseteq \underline \lay(\Ark)}
\]

\[
  \infer{\fwf{\specconfone:\cfstack}}{\fwf \specconfone }\quad
  \infer{\fwf{\nil}}{}
\]

\begin{lemma}
  \label{lemma:fencesempres}
  The semantics of every system $\system$ is equivalent to that of $\fencetrans (\system)$. 
\end{lemma}
\begin{proof}
  Direct consequence of \Cref{lemma:fencestep}.
\end{proof}

\begin{lemma}
  \label{lemma:fencesafeimpo}
  The transformation $\fencetrans$  \emph{imposes speculative kernel safety}.
\end{lemma}

\begin{proof}
  We assume that there is a system
  $\system = (\rfs, \syss, \funs)$, a system call $\syscall$
  a register map $\regmap$
  buffer $\buf$, an array store $\rfs' \sim_{\Fun} \rfs$, and a natural number $\nat$ such that
  \[
    \nsstep[\system][\lay] \nat {\sframe{\frame {\syss(\syscall)}{\regmap} {\km[\syscall]}}  {\bm\buf{(\lay\lcomp\rfs')}}{\boolms}} \unsafe \Ds \Os.
  \]
  for a sequence of directives $\Ds$ (that we assume free of $\dbt$ because
  of \Cref{lemma:nobt})
  producing a set of observations $\Os$. By case analysis on $\nat$,
  we refuse the case where $\nat= 0$, because that would mean
  \[
    {\sframe{\frame {\syss(\syscall)} {\regmap} {\km[\syscall]}}{\bm\buf{(\lay\lcomp\rfs')}}{\boolms}} = \unsafe.
  \]
  For this reason, in the following, we assume that $\nat>0$. By case analysis on
  the proof relation, we deduce that the rule that has been used to show the last transition
  must be one among \ref{AL:Load-Unsafe}, \ref{AL:Store-Unsafe} and \ref{AL:Call-Unsafe}.
  We just show the case for \ref{AL:Load-Unsafe}, in the other cases it is analogous.
  This means that
  \begin{equation*}
    \nsstep[\system][\lay] {\nat-1} {\sframe{\frame {\syss(\syscall)} {\regmap} {\km[\syscall]}}{\bm\buf{(\lay\lcomp\rfs')}}{\boolms}} {} {\Ds'} {\Os'} \\
    {\sframe{\frame {\cmemread \vx \expr\sep\cmdtwo} {\regmap'} {\km[\syscall]}:\st}{\bm{\buf'}{(\lay\lcomp\rfs'')}}{\boolms'}}:\cfstack  \sto {\dstep} {\omem\add} 
    \unsafe.
  \end{equation*}
  Observe that $\nat >1$: otherwise, if $\nat=1$, we would have ${\cmemread \vx \expr\sep\cmdtwo}= \syss(\syscall)$,
  but from $\system \in \im(\fencetrans)$, we deduce that there is a system $\system'=(\rfs'', \syss',\caps')$ such that $\fencetrans (\system') = \system$. This, in particular, would mean that
  $\syss(\syscall)=\fencetrans(\syss'(\syscall))$, so there is a command $\cmd\in \Cmd$ such that
  $\syss(\syscall)=\fencetrans(\syss'(\cmd))$ but, by induction on the syntax of the command,
  we observe that this is not possible. However, from \Cref{rem:stackinv2}, we deduce that
  $\cmd$ cannot contain system calls, so we can apply  \Cref{lemma:fwf},
  we deduce that there must be a stack of configurations
  \[
    \sframe{\frame{\cmd} {\regmap''} {\km[\syscall]}:\st''}{\bm{\buf''}{(\lay\lcomp\rfs''')}}{\boolms''}:\cfstack'
  \]
  such that:
  \begin{multline*}
    \nsstep[\system][\lay] {\nat-2} {\sframe{\frame {\syss(\syscall)} {\regmap} {\km[\syscall]}}{\bm\buf{(\lay\lcomp\rfs')}}{\boolms}} {}{\Ds''} {\Os''}\\
    {\sframe{\frame{\cmd} {\regmap''} {\km[\syscall]}:\st''}{\bm{\buf''}{(\lay\lcomp\rfs''')}}{\boolms''}:\cfstack'}  \sto \dir \obs\\
    {\sframe{\frame {\cmemread \vx \expr\sep\cmdtwo} {\regmap'} {\km[\syscall]}:\st}{\bm{\buf'}{(\lay\lcomp\rfs'')}}{\boolms'}}:\cfstack  \sto {\dstep} {\omem\add} 
    \unsafe.
  \end{multline*}
  and, in particular, $\fwf{\frame{\cmd} {\regmap''} {\km[\syscall]}:\st''}$.
  From this observation, we deduce that $\cmd = \fencetrans(\cmdtwo)$
  for some $\cmdtwo\in \Cmd$  and that $\fwf{\st''}$ holds.
  For these reasons, going by cases on $\cmdtwo$ and the rule that has been used to show the last transition (knowing that $\dir\neq \dbt$ by assumption),
  we can rewrite the reduction above, that rewrites as follows:
  \begin{multline*}
    \nsstep[\system][\lay] {\nat-2} {\sframe{\frame {\syss(\syscall)} {\regmap} {\km[\syscall]}}{\bm\buf{(\lay\lcomp\rfs')}}{\boolms}} {}{\Ds''} {\Os''}\\
    {\sframe{\frame{\cfence\sep\cmemread \vx \expr\sep\cmdtwo} {\regmap''} {\km[\syscall]}}{\bm{\buf''}{(\lay\lcomp\rfs''')}}{\boolms''}:\cfstack}  \sto \dir \obs\\
    {\sframe{\frame {\cmemread \vx \expr\sep\cmdtwo} {\regmap'} {\km[\syscall]}:\st}{\overline{\bm{\buf''}{(\lay\lcomp\rfs''')}}}{\bot}}:\cfstack  \sto {\dstep} {\omem\add} 
    \unsafe.
  \end{multline*}
  In particular, none of the rules except for \ref{AL:Fence} can have been used.
  This also means that $\boolms'=\boolms''=\bot$ because no $\dbt$
  directive can have been employed.
  From \Cref{lemma:nobtsteponly}, we deduce that there is $\nat'$ such that:
    \begin{equation*}
    \nsstep[\system][\lay] {\nat'} {\sframe{\frame {\syss(\syscall)} {\regmap} {\km[\syscall]}}{\bm\buf{(\lay\lcomp\rfs')}}{\boolms}} {}{\dstep^{\nat'}} {\Os'''}\\
    {\sframe{\frame {\cmemread \vx \expr\sep\cmdtwo} {\regmap'} {\km[\syscall]}:\st}{\overline{\bm{\buf''}{(\lay\lcomp\rfs''')}}}{\bot}}:\cfstack  \sto {\dstep} {\omem\add} 
    \unsafe.
  \end{equation*}
  Finally, we apply \Cref{lemma:stepsemsim} that shows:
  \begin{equation*}
    \nstep[\system][\lay] {\nat'} {\conf{\frame {\syss(\syscall)} {\regmap} {\km[\syscall]}, \overline{\bm\buf{(\lay\lcomp\rfs')}}}}{}\\
    {\conf{\frame {\cmemread \vx \expr\sep\cmdtwo}{\regmap'} {\km[\syscall]}:\st, \overline{\bm{\buf''}{(\lay\lcomp\rfs''')}}}}.
  \end{equation*}
  Finally, by assumption we know that the rule \ref{SI:Load-Unsafe} has been applied
  to show the transition in the speculative semantics. By introspection of that rule,
  we conclude that its premises are also verified by the configuration
  \[
    {\conf{\frame {\cmemread \vx \expr\sep\cmdtwo}{\regmap'} {\km[\syscall]}:\st, \overline{\bm{\buf''}{(\lay\lcomp\rfs''')}}}};
  \]
  this shows that \ref{WL:Load-Unsafe} applies, i.e.
  \begin{equation*}
    \nstep[\system][\lay] {\nat'} 
    {\conf{\frame {\cmemread \vx \expr\sep\cmdtwo}{\regmap'} {\km[\syscall]}:\st, \overline{\bm{\buf''}{(\lay\lcomp\rfs''')}}}}
    \unsafe,
  \end{equation*}
  and this shows $\cbu(\system)$.
\end{proof}

\paragraph{Technical Observations}

\newcommandx{\wf}[3][1=\system,2=\lay]{#1\vdash \mathsf{wf}_{#1}(#3)}

In order to show the analogous result but for the speculative semantics,
we define the judgment $\wf\cdot$. For a system $\system  = (\rfs, \syss, \caps)$,
it is defined as follows:

\[
  \infer{\wf{\sconf{\err, \bot}}}{}
  \quad\infer{\wf\unsafe}{}
  \quad\infer{\wf \nil}{}
\]
\[
  \infer{\wf{\sframe\st {\bm\buf\mem} \boolms:\cfstack} } {\wf\cfstack & \dom(\buf)\subseteq \lay(\Ar) & \mem = \lay \lcomp \rfs' & \rfs' \sim_{\Fn} \rfs }
\]

\begin{remark}
  \label{rem:simispreservedspec}
  For every system $\system = (\rfs, \syss, \caps)$, layout $\lay\in \Lay$, pair of configurations $\cfstack$ and $\cfstack'$, sequences of directives  $\Ds$ and of observations $\Os$, if $\wf\cfstack$ and $\nsstep*{\sframe {\frame \cmd \regmap \um} {\lay\lcomp \rfs} \bot}{\cfstack} \Ds \Os$, then we have $\wf\cfstack$.
\end{remark}
\begin{proof}
  The proof goes by induction on the number of steps.
  \begin{proofcases}
    \proofcase{0} Trivial.
    \proofcase{$\nat+1$} The claim is still a direct consequence of the IH, and in most of the cases the
    proof is straightforward. The most interesting cases are those for stores and fences.
    In the first case,the premise of the rule insures that the address where we write
    is part of the memory of an array. In the second case, it is a consequence of
    \Cref{rem:overlinewrtdom}.
  \end{proofcases}
\end{proof}

\begin{remark}
  \label{rem:simispreserved2}
  For every system $\system = (\rfs, \syss, \caps)$, layout $\lay\in \Lay$, pair of configurations $\conf{\st, \lay\lcomp \rfs, \Ds, \Os}$ and $\conf{\st', \mem,\Ds', \Os'}$ such that $\nesstep*{\conf{\st, \lay\lcomp \rfs, \Ds, \Os}}{\conf{\st', \mem,\Ds', \Os'}}$, we have that $\mem=\lay\lcomp \rfs'$ for some $\rfs'\sim_{\Fn} \rfs$.
\end{remark}
\begin{proof}
  To show this result we need to show, in conjunction with it, that if
  \[
    \nesstep*{\conf{\st, \lay\lcomp \rfs, \Ds, \Os}}{\conf{\cfstack, \st,\Ds', \Os'}}
  \]
  then $\wf\cfstack$ The proof goes by induction on the number of steps.
  \begin{proofcases}
    \proofcase{0} Trivial.
    \proofcase{$\nat+1$} The proof goes by cases on the
    rule that has been applied to show the last transition.
    The cases for the ordinary rules are simple. For this reason
    we just focus on those for the hybrid configuration.
    \begin{proofcases}
      \proofcase{\ref{AL:Spec-Init}} In this case, we can verify that the speculative
      frame of the target configuration enjoys $\wf \cdot$ simply by introspection of
      the target configuration and because of the IH.
      \proofcase{\ref{AL:Spec-Dir}, \ref{AL:Spec-Step}, \ref{AL:Spec-Bt}}
      The claim is a consequence of the IH and of \Cref{rem:simispreservedspec}.
      \proofcase{\ref{AL:Spec-Term}}
      The claim is a consequence of the IH, and of \Cref{rem:overlinewrtdom}
      which ensures that the memory that is extracted form the speculative stack
      enjoys the desired property.
    \end{proofcases}
  \end{proofcases}
\end{proof}

\begin{remark}
  \label{rem:noaddobs}
  For every pair of speculative stacks $\cfstack, \cfstack'$, system $\system=(\rfs, \syss, \caps)$,  $\nat\in\Nat$, sequence of directives $\Ds$, address $\add$ and layout $\lay$ such that $\add \notin \underline\lay(\Id)$, if:
  \[
    \nsstep \nat {\cfstack} {\cfstack'} \Ds {\Os},
  \]
  then $\Os$ does not contain the observation $\omem \add$.
\end{remark}

\begin{proof}
  The proof goes by induction on $\nat$.
  \begin{proofcases}
    \proofcase{0} Trivial
    \proofcase{$\nat+1$} We apply the IH, and we go by cases on the rule
    that has been used to show the last transition. For those rules that do not
    produce a transition like $\omem \add'$, the claim is trivial.
    The rules that can produce a similar transition are \ref{SI:Load-Step},
    \ref{SI:Load}, \ref{SI:Load-Unsafe}, \ref{SI:Store}, \ref{SI:Store-Unsafe},
    \ref{SI:Call},\\ \ref{SI:Call-Unsafe}. All these rules require in their premises that
    $\add' \in \underline\lay(\Ark)$,  $\add' \in \underline\lay(\Fnk)$,
    $\add' \in \underline\lay(\Aru)$, or that  $\add' \in \underline\lay(\Fnu)$, but
    since $\add \notin \underline\lay(\Id)$, we deduce that  $\add$
    does not belong to any of these sets, so it must be that $\add'\neq \add$.
  \end{proofcases}
\end{proof}

\begin{remark}
  \label{rem:stackinv1}
  For every system $\system = (\rfs, \syss, \caps)$,
  every store $\rfs \sim_{\Fun} \rfs'$,
  configuration
  \[
    \confone = \conf{\frame \cmd \regmap \um, \lay\lcomp \rfs', \Os, \Ds},
  \]
  configuration $\conf{\st', \lay\lcomp \rfs'', \Os, \Ds}$, and hybrid configuration
  $\conf{\cfstack, \conf{\st', \Os, \Ds}}$ that are reachable in $\nat$ step from $\confone$,
  there are a pair of stacks  $\st_{\km}, \st_\um$ such that
  $\st' = \st_{\km}:\st_\um$, $\km(\st_{\km})$ and $\um(\st_\um)$.  
\end{remark}
\begin{proof}
  By induction on $\nat$. The base case is trivial, the inductive one comes by cases on the rule that is applied.
  Most of these rules are relatively simple. We just show the most interesting cases:
  \begin{proofcases}
    \proofcase{\ref{AL:Pop}} The target configuration carries a stack which is a suffix of the one in the source configuration,
    so the conclusion is a trivial consequence of the IH.
    \proofcase{\ref{AL:System-Call}} We first apply the IH, which shows that:
    \begin{equation*}
      \nesstep \nat {\conf{\frame \cmd \regmap \um:\st, \lay\lcomp\rfs, \Os, \Ds}}{}\\
      \conf{\frame {\csyscall \syscall{\expr_1, \dots,\expr_k}\sep \cmdtwo} {\regmap'} \opt \st', \lay\lcomp\rfs, \Os, \Ds}
    \end{equation*}
    and that $\frame {\csyscall \syscall{\expr_1, \dots,\expr_k}} {\regmap'} \opt \st'$ is a concatenation
    of a pair of stacks  $\st_{\km}, \st_\um$ such that $\km(\st_{\km})$ and $\um(\st_\um)$.
    Since the current command is a system call, $\st_{\km}$ must be empty and $\opt$ must be $\um$. Observe that
    \begin{equation*}
      \esstep
      {\conf{\frame {\csyscall \syscall{\expr_1, \dots,\expr_k}\sep \cmdtwo} {\regmap'} \um \st', \lay\lcomp\rfs, \Os, \Ds}}{}\\
      {\conf{\frame {\syss(\syscall)} {\regmap''} {\km[\syscall]} : \frame {\cmdtwo} {\regmap'} \um \st', \lay\lcomp\rfs, \Os, \Ds}}
    \end{equation*}
    for a suitable $\regmap''$. The claim just requires verifying that the target configuration above
    satisfies the premises. In particular that $\km(\frame {\syss(\syscall)} {\regmap''} {\km[\syscall]})$
    holds because $\syss(\syscall)$ cannot contain system call invocations for the definition of system.
    \proofcase{\ref{AL:Call}} This case is analogous to the previous one, but instead of
    setting the flag to $\km$ in the new frame, the rule copies it from the topmost frame of the source configuration, and this does not break the invariant. In particular, we observe that the target configuration of the rule looks like the following one:
    \begin{equation*}
      \langle
      \frame
      {\lay\lcomp\rfs''(\add)}
      {\regmap_0[\vx_{1} \upd \sem{\exprtwo_1}_{\regmap,\lay}, \dots, \vx_{n} \upd \sem{\exprtwo_n}_{\regmap,\lay}]}
      {\opt}:\\
      {\frame{\cmd}{\regmap}{\opt} : \st},
      {\lay\lcomp\rfs''}
      \rangle.
    \end{equation*}
    From \Cref{rem:simispreservedspec}, we deduce that $\rfs''\sim_{\Fn}\rfs$, this means that ${\lay\lcomp\rfs''(\add)}$, because from the premises of the rule, we know that $\add \in \underline \lay(\Fun[\opt])$, which means that there is a function $\fn\in \Fun[\opt]$ such that $\lay(\fn)=\add$. From this observation and the definition of $\cdot\lcomp \cdot$, we deduce that $\lay\lcomp\rfs''(\add) = \rfs(\fn)$. So, we go by cases on $\opt$. If it is $\um$, the body of the function is unprivileged, and this shows $\lay\lcomp\rfs''(\add)$, analogously, if it is $\km$, we observe that, this program is in $\Cmd$ and does not have any $\csyscall \cdot \cdot$ instruction inside. This shows $\km[\syscall](\rfs(\fn))$, as required.  
    \proofcase{\ref{AL:Spec-Init}} Observe that this rule simply copies the stack from the source configuration
    to the target configuration, and removes the executed command from the topmost frame; the conclusion is a consequence of the IH.
  \end{proofcases}
\end{proof}

\newcommandx{\swf}[3][1=\system,2=\lay]{#1\vdash \mathsf{swf}_{#2}(#3)}

In order to show a similar result, but for speculative execution,
we strengthen the predicate $\wf\cdot$ into $\swf\cdot$ that is defined as follows:

\[
  \infer{\swf{\sconf{\err, \bot}}}{}
  \quad\infer{\swf\unsafe}{}
  \quad\infer{\swf \nil}{}
\]

\small
\[
  \infer{\swf{\sframe\st {\bm\buf\mem} \boolms:\cfstack} } {\wf{\sframe\st {\bm\buf\mem} \boolms:\cfstack} & \swf{\cfstack} & \st = \st_{\km}:\st_\um & \km[\syscall](\st_{\km}) & \um(\st_\um)}
\]
\normalsize

\begin{remark}
  \label{rem:stackinv2}
  For every system $\system = (\rfs, \syss, \caps)$,
  every stack $\cfstack$ such that $\swf\cfstack$,
  stack $\cfstack'$, such that  $\nsstep \nat \cfstack {\cfstack'} \Ds \Os$ ,
  we have $\swf {\cfstack'}$.  
\end{remark}
\begin{proof}
  By induction on $\nat$. The base case is trivial, the inductive one comes by cases on the rule that is applied. Thanks to \Cref{rem:simispreservedspec},
  we can just focus in showing the additional requirements on the stack. Most of these rules are relatively simple. We just show the most interesting cases:
  \begin{proofcases}
    \proofcase{\ref{SI:Pop}} The target configuration carries a frame-stack which is a suffix of the one in the source configuration, so the conclusion is a trivial consequence of the IH.
    \proofcase{\ref{AL:Call}} We apply the IH, and we rewrite the configuration reached after $\nat$ steps as follows:
    \begin{equation*}
      {\sframe{\frame {\ccall \expr{\expr_1,\dots,\expr_k}\sep\cmdtwo} \regmap \opt:\st}{\bm{\mu'}{\lay\lcomp\rfs}} \boolms\cons\cfstack''}
    \end{equation*}
    we know by IH that it enjoys $\swf \cdot$, and
    we observe that the target configuration looks like the following one:
    \begin{equation*}
      \sframe{\frame {\ccall \expr{\expr_1,\dots,\expr_k}\sep \cmdtwo} {\regmap'} \opt:\frame {\cmdtwo} \regmap \km:\st}{\bm{\mu'}{\lay\lcomp\rfs}} \boolms\cons\cfstack''
    \end{equation*}
    for some $\regmap_0$. From the premises of the rule, we know that $\add \in \underline \lay(\Fn[\opt])$, which means that there is a function $\fn\in \Fn[\opt]$ such that $\lay(\fn)=\add$. From this observation and the definition of $\cdot\lcomp \cdot$, we deduce that $\lay\lcomp\rfs''(\add) = \rfs(\fn)$. So, we go by cases on $\opt$. If it is $\um$, the body of the function is unprivileged, and this shows $\lay\lcomp\rfs''(\add)$, analogously, if it is $\km$, we observe that, this program is in $\Cmd$ and does not have any $\csyscall \cdot \cdot$ instruction inside. This shows $\km[\syscall](\rfs(\fn))$, as required 
    \proofcase{\ref{AL:System-Call}} We first apply the IH, which shows that the configuration
    reached after $\nat$ steps enjoys $\swf \cdot$, namely:
    \begin{equation*}
      \swf {\sframe{\frame {\csyscall \syscall{\expr_1,\dots,\expr_k}\sep \cmdtwo} \regmap \um:\st}{\bm{\mu'}{\lay\lcomp\rfs}} \boolms\cons\cfstack''}
    \end{equation*}
    form this observation, we deduce
    that
   \[
      \um({\frame {\csyscall \syscall{\expr_1,\dots,\expr_k}\sep \cmdtwo} \regmap \um:\st})
    \]
    must hold. The claim comes from introspection of the rule, and from the definition of
    $\syss$.
  \end{proofcases}
\end{proof}


\begin{remark}
  \label{rem:thereisaspecon}
  For every system $\system = (\rfs, \syss, \caps)$,
  natural number $\nat$, layout
  $\lay$, and every configuration
  $\conf{\frame \speccmd \regmap \um, \lay \lcomp \rfs, \nil, \nil}$
  where $\speccmd$ is unprivileged,
  and stack of directives $\Ds$ and stack of observations $\Os$,
  if
  \[
    \nesstep \nat {\conf{\frame \speccmd \regmap \um, \lay \lcomp \rfs, \nil, \nil}} {\conf{\cfstack, \st, \Ds, \Os}}
  \]
  then there is $\nat'\le \nat$
  and a configuration
  $\sframe{\frame{\cmd'}{\regmap'}{\um}}{\lay\lcomp\rfs'}{\bot}$
  with $\um(\cmd')$,
  a sequence of directives $\Ds'$ and a sequence of observations
  $\Os'$
  such that
  \[
    \nsstep {\nat'}
    {\sframe{\frame{\cmd'}{\regmap'}{\um}}{\lay\lcomp\rfs'}{\bot}}
    \cfstack    {\Ds'} {\Os'}.
  \]
\end{remark}

\begin{proof}
  The proof goes by induction on $\nat$.
  \begin{proofcases}
    \proofcase{0} Follows from vacuity of the premise.
    \proofcase{$\nat+1$} The premise is:
      \[
        \nesstep {\nat} {\conf{\frame \speccmd \regmap \um, \lay \lcomp \rfs, \nil, \nil}} \conftwo \ato  {\conf{\cfstack, \st, \Ds, \Os}} 
      \]
      We go by cases on the rule that has been applied to show the last transition.
      \begin{proofcases}
        \proofcase{\ref{AL:Spec-Init}} We observe that
        \[
          \conftwo = \conf{\frame{\cspec {\cmd''}\sep \cmdtwo}{  \regmap''}\opt, \lay\lcomp\rfs''}
        \]
        and the claim follows by introspection of the rule.
        $\um(\cmd'')$ comes from the definition of the language and, in particular,
        $\opt = \um$ comes from \Cref{rem:stackinv1}:  it cannot be that $\opt=\km[\syscall]$
        because, otherwise, we could not have $\cspec {\cmd'}\sep \cmdtwo$ as command.
        \proofcase{\ref{AL:Spec-Step}}  The IH provides 
        \[
          \nsstep {\nat'}
          {\sframe{\frame{\cmd'}{\regmap'}{\um}}{\lay\lcomp\rfs'}{\bot}}
          \cfstack    {\Ds'} {\Os'}.
        \]
        from the premise of the rule, we conclude
        \[
          \sstep
          \cfstack {\cfstack'}    {\dstep} \obs
        \]
        and this concludes the proof.
      \end{proofcases}
      The other cases are analogous.
  \end{proofcases}
\end{proof}

\begin{remark}
  \label{remark:nobttech}
  If
  \[
    \nsstep\nat  {\specconfone} {\cfstack}\Ds \Os
  \]
  and $\cfstack = \specconftwo\cons\cfstack'$ and $\cfstack'$ not empty, then there are $\nat'\le \nat$, ${\Ds'}$, ${\Os'}$ such that: 
  \[
    \nsstep{\nat'} {\specconfone} {\cfstack'} {\Ds'} {\Os'}
  \]
\end{remark}
\begin{proof}
  By induction on $\cfstack$ and then on $\nat$.
  \begin{proofcases}
    \proofcase{$\nil$} Absurd.
    \proofcase{$\specconftwo\cons\cfstack'$} The IH tells that for every $m$, there are $m'\le m$, ${\Ds'}$, ${\Os'}$ such that: 
    \[
      \nsstep{m'} {\specconfone} {\cfstack'} {\Ds'} {\Os'}
      \tag{IHS}
    \]
    we go by induction on $\nat$.
  \begin{proofcases}
    \proofcase{$0$} From the premise we can deduce $\nil=\cfstack'$, which is absurd.
    \proofcase{$\nat+1$} We go by cases on the directive used for the last step:
    \begin{proofcases}
      \proofcase{$\dstep$} The claim is a consequence of the IH on $\nat$.
      \proofcase{$\dload i, \dbranch b$} The witness we need to introduce is  $\nat$ step transition $\nsstep\nat {\specconfone} {\specconftwo\cons\cfstack'}\Ds \Os$.
      \proofcase{$\dbt$} We go by cases on the rule that has been applied. It must be one of \ref{SI:Backtrack-Top}, \ref{SI:Backtrack-Top}. We will just show the result for the case of ordinary configurations. The proof with error configurations is analogous.
      \begin{proofcases}
        \proofcase{\ref{SI:Backtrack-Top}} In this case, we go by cases on the $\nat+1$-th target stack. It cannot be empty. If it has one element, the claim holds for vacuity of the premise, if it has more than one element, then we can apply the IH on $\cfstack$.
        \proofcase{\ref{SI:Backtrack-Bot}} The claim holds for vacuity of the premise: the $\nat+1$-th target stack has just one element.
      \end{proofcases}
    \end{proofcases}
  \end{proofcases}
  \end{proofcases}
\end{proof}

\begin{remark}
  \label{lemma:misspecunset}
  For every $\nat$, initial configuration $\specconfone=\sframe{\frame \cmd \regmap \um} {\mem} \bot$, if:
  \[
    \nsstep \nat  {\specconfone} {\specconftwo} {\Ds} {\Os},
  \]
  and $\specconftwo$ has a mis-speculation flag ($\specconftwo \neq \unsafe$) then the mis-speculation-flag of $\specconftwo$ must be $\bot$.
\end{remark}
\begin{proof}
  The proof is by induction, and the base case is trivial. In the inductive case, we go by cases on the directive of the last transition:
  \begin{proofcases}
    \proofcase{$\dstep$} By introspection of this fragment of the semantics, we deduce that the number of stacks in the $\nat$-th and in the $\nat+1$-th target configurations is the same, so in particular the $\nat$-th target stack must have one entry only. For this reason we can apply the IH and deduce that the mis-speculation flag of the $\nat$-th target configuration $\bot$. Then we observe that for all these rules, the flag is always copied from the source configuration to the target one, so we conclude.
    \proofcase{$\dload i, \dbranch b$} For all the rules that match these directives, the claim hold for vacuity of the premise: the $\nat+1$-th target configuration stack has height greater than $1$. The only exception is \ref{SI:Load-Err}. In this case, the proof is analogous to the case of the same but with $\dstep$ directive, which has already taken in account in the previous step.  
    \proofcase{$\dbt$} We go by cases on the rule that has been applied. It must be one of \ref{SI:Backtrack-Top}, \ref{SI:Backtrack-Bot}. We will just show the result for the ordinary configurations, the case for errors is analogous.
    \begin{proofcases}
          \proofcase{\ref{SI:Backtrack-Bot}} In this case, the claim follows directly from the definition of the rule.
          \proofcase{\ref{SI:Backtrack-Top}} Call $\cfstack_\nat$ the $\nat$-th target configuration stack. Observe that its height cannot be neither $0$ nor $1$. In the first case there would be no next transition, in the other case, we would get for the IH that its mis-speculation flag is unset, which is in contradiction with the assumption on the applied rule. Observe that if its height is greater than $2$, then the height of the $\nat+1$-th target stack is greater than $1$, so the claim holds for vacuity of the premise. Finally, if its height is exactly $2$, then it must be in the shape $\specconftwo_\nat\cons \specconftwo_{\nat}'$. From \Cref{remark:nobttech}, we deduce that there is a sequence of transitions from $\specconfone$ to $\specconftwo_\nat'$ whose length is $\nat'\le \nat$. By applying the IH on this sequence, we can show that $\specconftwo_\nat'$ has mis-speculation flag unset. We conclude observing that the $\nat+1$-th target configuration stack is exactly $\specconftwo_\nat'$.  
        \end{proofcases}
  \end{proofcases}
\end{proof}

\begin{lemma}
  \label{lemma:nobtsteponly}
  For every $\nat$ initial configuration $\confone$, if:
  \[
    \nsstep \nat  {\confone} {\conftwo: \cfstack} {\Ds} {\Os} 
  \]
  and the mis-speculation flag of $\conftwo$ is $\bot$, then there are $\overline \Os$ and $\nat'\le \nat$ such that :
  \[
    \nsstep{\nat'}  {\confone} {\conftwo} {{\dstep^{\nat'}}} {{\overline \Os}}
  \]
\end{lemma}

\begin{proof}
  The proof is by induction on $\nat$. The base case is trivial. For the inductive case, we go by cases on the directive of the $\nat+1$-th transition.
  The premise tells us that:
  \[
    \nsstep \nat  {\confone} {\conftwo: \cfstack} {\Ds} {\Os} \sto \dir \obs {\conftwo': \cfstack'}  
  \]
  and we know that the mis-speculation flag of $\conftwo'$ is $\bot$. We must show that
  there are $\overline \Os'$ and $\nat''\le \nat+1$ such that :
  \[
    \nsstep{\nat''}  {\conftwo} {\conftwo'} {{\dstep^{\nat''}}} {{\overline \Os}}
  \]
  \begin{proofcases}
    \proofcase{$\dstep$} Observe that the mis-speculation flag of $\conftwo$  must be $\bot$, because no rules for the directive $\dstep$ changes it. For this reason we can apply the IH. This shows that:
    \[
      \nsstep{\nat'}  {\confone} {\conftwo} {{\dstep^{\nat'}}} {{\overline \Os}}
    \]
    We examine all the rules matching the transition
    \[
      \sstep {\conftwo: \cfstack}  {\conftwo': \cfstack'}  \dir \obs
    \]
    and we observe that the premises these rules do not depend on $\cfstack$, but they depend on $\conftwo$ only, so if one of these rules is applied to show the transition above, it can be applied only on the transition of the claim. By introspection of all these rules, we observe that the configuration they produce is exactly $\conftwo''$.
    \proofcase{$\dload i$} There are two rules that match this directive and our premises, namely \ref{SI:Load-Err} and \ref{SI:Load}. We show just the case of the second one. The first one can be reduced for the same rule with the $\dstep$ directive.
    If this rule is applied, we can rewrite
    \[
      \sstep {\conftwo: \cfstack}  {\conftwo': \cfstack'}  \dir \obs
    \]
    as follows
    \begin{multline*}
      \sstep       {\sframe{\frame{\cmemread[\lbl] \vx  \expr\sep\cmd}{\regmap}{\opt}\cons \st} {\bm\buf\mem}{\boolms}\cons\cfstack} {}{{{\dload[\lbl] i}}} {\omem \add}\\
      \sframe{\frame{\cmd}{\update \regmap x \val}{\opt}\cons\st}{\bm\buf\mem}{\boolms\lor\bool'}\cons\\
        \sframe{\frame{\cmemread[\lbl] \vx  \expr\sep\cmd}{\regmap}{\opt}\cons\st}{\bm\buf\mem}{\boolms}\cons \cfstack
    \end{multline*}
    From the premises of the rule, we deduce:
    \begin{varitemize}
    \item $\toAdd{\sem\expr_{\regmap, \lay}} = \add$
    \item $\bufread {\bm\buf\mem} \add i =\val, \bool'$
    \item $\add \in \underline \lay(\Ar[\opt])$
    \item $\opt = \km[\syscall] \Rightarrow \add \in \underline \lay(\caps(\syscall))$
    \end{varitemize}
    In particular, from the main premise of the claim, we deduce that $\bool'$ cannot be
    $\top$.  This means that $\bufread {\bm\buf\mem} \add i =\val, \bot$. Our goal is to show
    \begin{equation*}
      \sstep {\sframe{\frame{\cmemread[\lbl] \vx  \expr\sep\cmd}{\regmap}{\opt}\cons \st} {\bm\buf\mem}{\boolms}} {}{{{\dstep}}} {\omem \add}\\
      \sframe{\frame{\cmd}{\update \regmap x \val}{\opt}\cons\st}{\bm\buf\mem}{\boolms\lor\bool'}
    \end{equation*}
    To do so, we can observe that, if we show that $\bufread {\bm\buf\mem} \add 0 =\val, \bot$,
    then premises of the rule \ref{SI:Load-Step} are matched, and thus the proof is concluded.
    This is a consequence of \Cref{remark:onbuflookup}.
    \proofcase{$\dbranch \bool$} Analogous to the case above.
    \proofcase{$\dbt$} If the rule which has been applied is \ref{SI:Backtrack-Bot}, the claim is a direct consequence of the IH. Otherwise, the rule applied is  \ref{SI:Backtrack-Top}. We go by cases on the height of $\conftwo: \cfstack$. If it is 0, we reached a contradiction, if its height is $1$, then the topmost configuration must have mis-speculation flag unset for \Cref{lemma:misspecunset}, ad this contradicts the assumption on the applied rule, which requires it to be set to $\top$. If the height is greater, we deduce that $\conftwo: \cfstack = \conftwo \cons \conftwo'\cons \cfstack$ and that the mis-speculation flag of $\conftwo'$ is $\bot$ by introspection of the rule and by the main assumption of this claim. We apply \Cref{remark:nobttech} on the $\nat$-step reduction from $\confone$ to $\conftwo \cons \conftwo'\cons \cfstack$ in order to show that there are $\nat'\le \nat$ and a $\nat'$-long reduction from $\confone$ to $\conftwo'\cons \cfstack'$. For the IH, there is a $\nat''\le \nat'$ step reduction from $\confone$ to $\conftwo'$ which employs only the directive $\dstep$. Which is our claim.
  \end{proofcases}
\end{proof}

\begin{lemma}
  \label{lemma:nobt}
  For every system $\system = (\rfs, \syss, \caps)$,
  speculative configuration $\specconfone = \sframe {\frame \cmd \regmap \opt} {\mem} \bot$
  non-empty speculative stack $\cfstack$,
  sequence of directives $\Ds$ and
  sequence of observations $\Os$,
  if:
  \[
    \nsstep \nat \specconfone {\cfstack} \Ds \Os
  \]
  then there are $\nat'\le\nat$, $\Ds'$, $\Os'$
  such that 
  \[
    \nsstep {\nat'} \specconfone {\cfstack} {\Ds'} {\Os'}.
  \]
  and $\Ds'$ does not contain any $\dbt$ directive.
\end{lemma}
\begin{proof}
  By induction on $\nat$.
  \begin{proofcases}
    \proofcase{0} Trivial.
    \proofcase{$\nat+1$} The premise tells
    \[
      \nsstep \nat \specconfone {\cfstack'} \Ds \Os \sto\dir \obs \cfstack
    \]
    The IH shows that there are $\nat'\le \nat$,
    $\Ds'$ without $\dbt$ directives and $\Os'$ such that:
    \[
      \nsstep {\nat'} \specconfone {\cfstack'} {\Ds'} {\Os'}.
    \]
    We are required to show that if
    \[
      \sstep {\cfstack'} {\cfstack} \dir \obs,
      \tag{\dag}
    \]
    then there is $\nat''\le \nat+1$, $\Ds''$ and $\Os'$ such that 
    \[
      \nsstep {\nat''} \specconfone {\cfstack'} {\Ds''} {\Os''}.
    \]
    where in particular $\Ds''$ does not contain $\dbt$ directives. We go by cases
    on the directive $\dir$ that is used in (\dag).
    \begin{proofcases}
      \proofcase{$\dir\neq\dbt$} In these cases, the claim holds if we take $\Ds''=\Ds':\dir$, $\Os''=\Os':\obs$ and $\nat''=\nat'+1$.
      \proofcase{$\dir=\dbt$} The applied rule can either be \ref{SI:Backtrack-Top}, or \ref{SI:Backtrack-Bot}. We just take in exam the case of ordinary configurations:
        \begin{proofcases}
          \proofcase{\ref{SI:Backtrack-Top}} In this case, we conclude that  $\cfstack' = \specconftwo :\cfstack''$ by introspection of the rule, we deduce that the mis-speculation flag of the topmost configuration of $\cfstack'$ is $\top$, so we apply \Cref{lemma:misspecunset} to deduce that $\cfstack''\neq \nil$ (otherwise we could not have the flag set to $\top$).  Thanks to this observation, we can use \Cref{remark:nobttech} to show that there is a sequence of transitions from $\cfstack$ to $\cfstack''$ whose length is $\nat'\le \nat$. The conclusion is a consequence of the application of the IH on this intermediate result. This concludes the sub-derivation.
          \proofcase{\ref{SI:Backtrack-Bot}} In this case, we use \Cref{lemma:nobtsteponly} to show that there is a sequence of transitions from $\specconfone$ to $\confone$ containing only the directive $\dstep$, which is stronger than the claim we need.
  \end{proofcases}
    \end{proofcases}
  \end{proofcases}
\end{proof}

\newcommand{\fencerel}{\precsim}

For sake of simplicity, we define the relation $\fencerel$ as follows:

\[
  \infer{\varepsilon \fencerel \varepsilon}{} \quad
  \infer{\rfs \fencerel \rfs'}{\rfs \sim_{\Idu\cup \Ark} \rfs' & \forall \fn \in \Fnk. \rfs'(\fn)= \fencetrans(\rfs(\fn))}
\]
\[    
  \infer{\frame \cmd \regmap {\km[\syscall]}:\st \fencerel \frame {\cmd'} {\regmap} {\km[\syscall]}:\st'}{\st\fencerel\st' & \cmd'=\fencetrans (\cmd)}\quad
  \infer{\frame \cmd \regmap {\um}:\st \fencerel \frame {\cmd} {\regmap} {\um}:\st'}{\st\fencerel\st'}
\]
\[
  \infer{\err \fencerel \err}{}\quad
  \infer{\unsafe \fencerel \unsafe}{} \quad
  \infer{\conf{\st, \lay \lcomp \rfs} \fencerel
    \conf{\st', \lay \lcomp \rfs}}{\rfs\fencerel \rfs' & \st \fencerel\st'}
\]

\begin{lemma}
  \label{lemma:fencestep}
  For every system $\system$, configurations
  $\confone$
  and 
  $\conftwo$ such that $\confone \fencerel \conftwo$, if
  \[
    \step[\lay][\system]  {\confone}
    {\confone'},
  \]
  then
  \[
    \nstep[\lay][\fencetrans(\system)] * {\conftwo}
    {\conftwo'},
  \]
  for some  $\conftwo'$ such that $\confone' \fencerel \conftwo'$.
\end{lemma}

\begin{proof}
  The proof goes first by cases on $\confone$ to refuse the cases
  where it is $\err$ or $\unsafe$. After that we can
  assume that $\confone = \conf{\frame \cmd \regmap \opt : \st, \lay \lcomp \rfs}$,
  and we go by cases on $\opt$.
  \begin{proofcases}
    \proofcase{$\um$}
    We can assume that
    \[
      \confone = \conf{\frame \cmd \regmap \um : \st, \lay \lcomp \rfs}
    \]
    and
    \[
      \conftwo = \conf{\frame \cmd \regmap \um : \st', \lay \lcomp \rfs'}.
    \]
    with $\st \fencerel \st'$ and $\rfs \fencerel \rfs'$. The proof goes by cases on the rule that has been applied to $\cmd$
    \begin{proofcases}
      \proofcase{\ref{WL:Op}}
      The assumption is
      \begin{equation*}
        \step[\lay][\system]  {\conf{\frame {\vx \ass \expr\sep \cmdtwo} \regmap \um : \st,  \lay \lcomp \rfs}}{}\\
        {\conf{\frame {\cmdtwo} {\update\regmap\vx{\sem \expr_{\regmap, \lay}}} \um : \st,  \lay \lcomp \rfs}},
      \end{equation*}
      and by applying the same rule to $\conftwo$, we obtain
      \begin{equation*}
        \step[\lay][\fencetrans(\system)]  {\conf{\frame {\vx \ass \expr\sep \cmdtwo} \regmap \um : \st',  \lay \lcomp \rfs'}}{}\\
        {\conf{
            \frame {\cmdtwo} {\update\regmap\vx{\sem \expr_{\regmap, \lay}}} \um : \st',
            \lay \lcomp \rfs'}}
      \end{equation*}
      and this shows the claim.
      \proofcase{\ref{WL:Skip}} Analogous to the previous one.
      \proofcase{\ref{WL:Fence}} Analogous to the case of assignments.
      \proofcase{\ref{WL:Pop}} Analogous to the case of assignments.
      \proofcase{\ref{WL:If}} Analogous to the case of assignments.
      \proofcase{\ref{WL:While}} Analogous to the case of assignments.
      \proofcase{\ref{WL:Store}} The assumption is
      \begin{equation*}
        \step[\lay][\system]  {\conf{\frame {\cmemass \expr\exprtwo\sep \cmdtwo} \regmap \um : \st,  \lay \lcomp \rfs}}{}\\
        {\conf{\frame {\cmdtwo} {\regmap} \um : \st,  \update{\lay \lcomp {\rfs}}\add\val}},
      \end{equation*}
      where $\val = \sem \exprtwo_{\regmap, \lay}$ and $\add = \toAdd{\sem \expr_{\regmap, \lay}}$.
      Observe that $\add \in \underline \lay(\Aru)$.
      Which means that there is a pair $(\ar, i)$ such that $\lay(\ar)+i=\add$ and $\ar \in \Aru$.
      So, in particular, we have that $\update{\lay \lcomp {\rfs}}\add\val={\lay \lcomp {\update{{\rfs}}{(\ar, i)} \val}}$ because of \Cref{rem:memupdtostupd}.
      We also deduce that by applying the same rule to $\conftwo$, we obtain
      \begin{equation*}
        \step[\lay][\system]  {\conf{\frame {\cmemass \expr\exprtwo\sep \cmdtwo} \regmap \um : \st',  \lay \lcomp \rfs'}}{}\\
        {\conf{\frame {\cmdtwo} {\regmap} \um : \st', \lay \lcomp { \update{{\rfs'}}{(\ar, i)}\val}}},
      \end{equation*}
      and to show the claim, we just need to observe $\update{{\rfs}}{(\ar, i)}\val \fencerel \update{{\rfs'}}{(\ar, i)}\val$, which is a consequence of the assumption $\ar \in \Aru$.
      \proofcase{\ref{WL:Load}} Analogous to the previous case
      \proofcase{\ref{WL:Call}} 
      The assumption is
      \begin{equation*}
        \step[\lay][\system]  {\conf{\frame {\ccall \expr{\exprtwo_1,\dots, \exprtwo_k}\sep \cmdtwo} \regmap \um : \st,  \lay \lcomp \rfs}}{}\\
        \conf{\frame {\lay \lcomp \rfs(\sem \expr_{\regmap, \lay}) } {\regmap_0'} \um : \conf{\frame {\cmdtwo} {\regmap} \um : \st,  \lay \lcomp \rfs}},
      \end{equation*}
      where $\regmap_0'$ is a shorthand for
      \[
        \regmap_0[\vx_1, \dots,\vx_k \upd\sem{\exprtwo_1}_{\regmap, \lay}, \dots, \sem{\exprtwo_1}_{\regmap, \lay}]
      \]
      and from the premise of the rule, we obtain that there is $\fn \in \Fnu$ such that
      $\sem \expr_{\regmap, \lay} = \lay(\fn)$. Thus, from the definition of $\lay \lcomp \rfs$,
      we deduce that $ \lay \lcomp \rfs = \rfs(\fn) = \rfs'(\fn)$ for
      the definition of the $\fencerel$ relation. Thanks to these observations,
      we can show that the application of the same rule to $\conftwo$ gives
      rule to $\conftwo$, we obtain
      \small
      \begin{equation*}
        \step[\lay][\fencetrans(\system)]  {\conf{\frame {\ccall \expr{\exprtwo_1,\dots, \exprtwo_k}\sep \cmdtwo} \regmap \um : \st',  \lay \lcomp \rfs'}}{}\\
        {\conf{\frame {\lay \lcomp \rfs(\sem \expr_{\regmap, \lay}) } {\regmap_0'} \um : \conf{\frame {\cmdtwo} {\regmap} \um : \st',  \lay \lcomp \rfs'}}},
      \end{equation*}
      \normalsize
      and this shows the claim.
      \proofcase{\ref{WL:SystemCall}} 
      The assumption is
      \begin{equation*}
        \step[\lay][\system]  {\conf{\frame {\csyscall \syscall{\exprtwo_1,\dots, \exprtwo_k}\sep \cmdtwo} \regmap \um : \st,  \lay \lcomp \rfs}}{}\\
        \conf{\frame {\syss(\syscall) } {\regmap_0'} {\km[\syscall]} : \conf{\frame {\cmdtwo} {\regmap} \um : \st,  \lay \lcomp \rfs}},
      \end{equation*}
      where $\regmap_0'$ is a shorthand for
      \[
        \regmap_0[\vx_1, \dots,\vx_k \upd\sem{\exprtwo_1}_{\regmap, \lay}, \dots, \sem{\exprtwo_1}_{\regmap, \lay}]
      \]
      By applying the same rule to the configuration $\conftwo$, we obtain
      \begin{equation*}
        \step[\lay][\fencetrans(\system)]  {\conf{\frame {\csyscall \syscall{\exprtwo_1,\dots, \exprtwo_k}\sep \cmdtwo} \regmap \um : \st',  \lay \lcomp \rfs'}}{}\\
        {\conf{\frame {\fencetrans(\syss)(\syscall)} {\regmap_0'} {\km[\syscall]} : \conf{\frame {\cmdtwo} {\regmap} \um : \st',  \lay \lcomp \rfs'}}},
      \end{equation*}
      to conclude, we just need to observe that
      \[
        \frame {\syss(\syscall)} {\regmap_0'} {\km[\syscall]}\fencerel
        \frame {\fencetrans(\syss)(\syscall)} {\regmap_0'} {\km[\syscall]}.
      \]
      \proofcase{\ref{WL:Store-Error}} The assumption is
      \begin{equation*}
        \step[\lay][\system]  {\conf{\frame {\cmemass \expr\exprtwo\sep \cmdtwo} \regmap \um : \st,  \lay \lcomp \rfs}}\err,
      \end{equation*}
      We call $\add = \toAdd{\sem \expr_{\regmap, \lay}}$, and from the premises
      of the rule, we deduce that $\add \notin \underline \lay(\Aru)$.
      This suffices to show
      \begin{equation*}
        \step[\lay][\system]  {\conf{\frame {\cmemass \expr\exprtwo\sep \cmdtwo} \regmap \um : \st',  \lay \lcomp \rfs'}}{\err}
      \end{equation*}
      and to show the claim.
    \end{proofcases}
    \proofcase{$\km[\syscall]$} Under this assumption, most of the cases are analogous to
    the corresponding ones for user-mode execution. Even in this case, we go by induction on the
    rules. We just show some among the most important cases:
    \begin{proofcases}
      \proofcase{\ref{WL:Op}}
      The assumption is
      \begin{equation*}
        \step[\lay][\system]  {\conf{\frame {\vx \ass \expr\sep \cmdtwo} \regmap {\km[\syscall]} : \st,  \lay \lcomp \rfs}}{}\\
        {\conf{\frame {\cmdtwo} {\update\regmap\vx{\sem \expr_{\regmap, \lay}}} {\km[\syscall]} : \st,  \lay \lcomp \rfs}},
      \end{equation*}
      from the assumption $\conftwo$, we deduce that is
      \[
        {\conf{\frame {\vx \ass \expr\sep \fencetrans(\cmdtwo)} \regmap {\km[\syscall]} : \st',  \lay \lcomp \rfs'}}
      \]
      so the same rule can be applied on $\conftwo$ to show:
      \begin{equation*}
        \step[\lay][\fencetrans(\system)]  {\conf{\frame {\vx \ass \expr\sep \cmdtwo} \regmap {\km[\syscall]} : \st',  \lay \lcomp \rfs'}}{}\\
        {\conf{
            \frame {\fencetrans(\cmdtwo)} {\update\regmap\vx{\sem \expr_{\regmap, \lay}}} {\km[\syscall]} : \st',
            \lay \lcomp \rfs'}}
      \end{equation*}
      and the claim comes from the observation that the command of the target configuration is exactly $\fencetrans(\cmdtwo)$.
      \proofcase{\ref{WL:Call}} 
      The assumption is
      \begin{equation*}
        \step[\lay][\system]  {\conf{\frame {\ccall \expr{\exprtwo_1,\dots, \exprtwo_k}\sep \cmdtwo} \regmap {\km[\syscall]} : \st,  \lay \lcomp \rfs}}{}\\
        \conf{\frame {\lay \lcomp \rfs(\sem \expr_{\regmap, \lay}) } {\regmap_0'} {\km[\syscall]} : \conf{\frame {\cmdtwo} {\regmap} {\km[\syscall]} : \st,  \lay \lcomp \rfs}},
      \end{equation*}
      where $\regmap_0'$ is a shorthand for
      \[
        \regmap_0[\vx_1, \dots,\vx_k \upd\sem{\exprtwo_1}_{\regmap, \lay}, \dots, \sem{\exprtwo_1}_{\regmap, \lay}]
      \]
      and from the premise of the rule, we obtain that there is $\fn \in \Fnk$ such that
      $\sem \expr_{\regmap, \lay} = \lay(\fn)$. Thus, from the definition of $\lay \lcomp \rfs$,
      we deduce that $ \lay \lcomp \rfs = \rfs(\fn)$.
      We observe that $\conftwo$ is:
      \[
        {\conf{\frame {\cfence \sep\ccall \expr{\exprtwo_1,\dots, \exprtwo_k}\sep \fencetrans(\cmdtwo)} \regmap {\km[\syscall]} : \st',  \lay \lcomp \rfs'}}
      \]
      Thanks to these observations,
      we can show that the application of the \ref{WL:Fence} rule and of the \ref{WL:Call} rule
      to $\conftwo$ show:
      \begin{multline*}
        \nstep[\lay][\fencetrans(\system)] 2 {\conf{\frame {\ccall \expr{\exprtwo_1,\dots, \exprtwo_k}\sep \cmdtwo} \regmap {\km[\syscall]} : \st',  \lay \lcomp \rfs'}}{}\\
        {\conf{\frame {\lay \lcomp \rfs'(\sem \expr_{\regmap, \lay}) } {\regmap_0'} {\km[\syscall]} : \conf{\frame {\fencetrans(\cmdtwo)} {\regmap} {\km[\syscall]} : \st',  \lay \lcomp \rfs'}}},
      \end{multline*}
      To conclude, we must observe that $\lay \lcomp \rfs'(\sem \expr_{\regmap, \lay}) = \fencetrans(\rfs(\fn))$. This is a consequence of the assumptions $\sem \expr_{\regmap, \lay} = \lay(\fn)$, $\fn\in \Fnk$ and $\rfs \fencerel \rfs'$.
    \end{proofcases}
  \end{proofcases}
\end{proof}

\begin{lemma}
  \label{lemma:fwf}
  For every stack of speculative configurations $\cfstack$
  layout $\lay$,
  system $\system\in \im(\fencetrans)$ such that $\fwf\cfstack$,
  $\dbt$-free sequence of directives $\dir:\Ds$, sequence of observations
  $\obs:\Os$,
  and $\cfstack'$ such that $\lnot (\fwf{\cfstack'})$,
  there is a stack $\cfstack''$ such that
  \[
    (\nsstep \nat \cfstack {\cfstack'} {\Ds:\dir} {\Os:\obs}) \Rightarrow (\nsstep {\nat-1} \cfstack {\cfstack''}{\Ds} {\Os}) 
  \]
  and $\fwf{\cfstack''}$.
\end{lemma}
\begin{proof}
  By cases on $\nat$.
  \begin{proofcases}
    \proofcase{0} Absurd.
    \proofcase{$\nat+1$} In this case, we assume that
    \[
      (\nsstep {\nat+1} \cfstack {\cfstack'}{\Ds:\dir} {\Os:\obs}) 
    \]
    and $\lnot (\fwf{\cfstack'})$. We need to show that
    \[
      (\nsstep {\nat} \cfstack {\cfstack''}{\Ds} {\Os})
    \]
    for some $\cfstack''$ such that $\fwf{\cfstack''}$.
    Assume that the claim does not hold, i.e. that $\lnot\fwf{\cfstack''}$.
    If $\nat = 0$, that stack is exactly $\cfstack$, but this is absurd,
    so we can assume that $\nat>0$, however,
    in this case we would have a contradiction of
    \Cref{rem:fwftech} for the configuration that is reached after $\nat-1$ steps
    because none of the two following configuration satisfies
    $\fwf\cdot$.
  \end{proofcases}
\end{proof}

\begin{remark}
  \label{rem:fwftech}
  For every stack of speculative configurations $\cfstack$
  every layout $\lay$,
  every system $\system\in \im(\fencetrans)$ such that $\fwf\cfstack$,
  and every system-call $\syscall$,
  one of the three following cases holds:
  \begin{varitemize}
  \item $\nf \cfstack{}$
  \item $\forall \dir \neq \dbt.\forall \obs. \sstep \cfstack {\cfstack'} \dir \obs \Rightarrow \fwf{\cfstack'}$
  \item $\forall \dir_1, \dir_2 \neq \dbt.\forall \obs_1, \obs_2. \nsstep 2 \cfstack {\cfstack'} {\dir_1:\dir_2} {\obs_1:\obs_2} \Rightarrow \fwf{\cfstack'}$
  \end{varitemize}
\end{remark}
\begin{proof}
  The proof goes by cases on the proof of the predicate $\fwf\cdot$.
  Many of the cases are trivial. The most interesting ones are when
  the stack of configurations is not-empty, i.e.
  it is like $\sframe{\frame {\fencetrans(\cmd)} \regmap {\km[\syscall]}:\st}
  {\bm\buf(\lay\lcomp \rfs')}\boolms:\cfstack$.
  The proof goes by cases on $\cmd$.
  Observe that from \label{rem:stackinv2cor}, we can avoid the case
  of system calls. 
  \begin{proofcases}
    \proofcase{$\cnil$}
    Observe that if $\st = \nil$, the first claim holds, otherwise:
    \begin{multline*}
      \sstep{\sframe{\frame {\fencetrans(\cnil)} \regmap {\km[\syscall]}:{\frame \cmdtwo {\regmap'} {\km[\syscall]}}:\st'}
        {\bm\buf(\lay\lcomp \rfs')}\boolms:\cfstack}{}{\dstep} {\onone}\\
      {\sframe{{\frame \cmdtwo {\update{\regmap'}\ret{\regmap(\ret)}} {\km[\syscall]}}:\st'}
        {\bm\buf(\lay\lcomp \rfs')}\boolms:\cfstack} 
    \end{multline*}
    and the claim is a consequence of the premise of the proof-rule for $\fwf\cdot$.
    \proofcase{$\cmemread \vx \exprtwo \sep \cmdtwo$}
    Observe that one rule among \ref{SI:Load-Step}, \ref{SI:Load},
    \ref{SI:Load-Unsafe}, and \ref{SI:Load-Err} must apply. If one of the
    last two rules applies, the claim is trivial. Otherwise, the stack looks like:
    \[
      {\sframe{\frame {\fencetrans(\cmemread \vx \exprtwo \sep \cmdtwo)} \regmap {\km[\syscall]}:\st'}
        {\bm\buf(\lay\lcomp \rfs')}\boolms:\cfstack}
    \]
    And if $\boolms = \top$, the first claim holds, because, without backtracking,
    this configuration cannot reduce further. So we can assume that $\boolms=\bot$
    Otherwise, we just show the case for the \ref{SI:Load} rule:
    \begin{multline*}
      \nsstep 2 {\sframe{\frame {\fencetrans(\cmemread \vx \exprtwo \sep \cmdtwo)} \regmap {\km[\syscall]}:\st'}
        {\bm\buf(\lay\lcomp \rfs')}\boolms:\cfstack}{} {\dstep:\dload i } {\onone:\omem \add}\\
      \sframe{\frame {\fencetrans(\cmdtwo)} {\update\regmap\vx \val} {\km[\syscall]}:\st'} 
        {\bm\buf(\lay\lcomp \rfs')}{\boolms\lor\bool'}:\\
        \sframe{\frame {\cmemread \vx \exprtwo \sep \fencetrans (\cmdtwo)} \regmap {\km[\syscall]}:\st'}
        {\overline{\bm\buf(\lay\lcomp \rfs')}}\boolms:\cfstack 
    \end{multline*}
    where $\add = \toAdd{\sem \exprtwo_{\regmap,\lay}}$,
    $\val = \bufread {\bm\buf(\lay\lcomp \rfs')} \add i$.
    Observe that in particular, the fence instruction that precedes the
    assignment has flushed the memory. For this reason, in order to show the claim,
    we just need to verify that:
    \[
      \fwf{\sframe{\frame {\fencetrans(\cmdtwo)} \regmap {\km[\syscall]}:\st'} 
        {\overline{\bm\buf(\lay\lcomp \rfs')}}{\boolms\lor\bool'}}
    \]
    which is almost entirely consequence of the premise of the proof-rule for $\fwf\cdot$.
    In particular, since the domain of $\buf$ does not contain any function address,
    the result is a consequence of \Cref{rem:overlinewrtdom}, which ensures that the
    resulting memory is $\lay\lcomp \rfs''$ and $\rfs'' \sim_{\Fn}\rfs'\sim_{\Fn}\rfs$.
    \proofcase{$\cmemass \expr \exprtwo \sep \cmdtwo$}
    Observe that one rule among \ref{SI:Store},\\
    \ref{SI:Store-Unsafe}, and \ref{SI:Store-Err} must apply. If one of the
    last two rules applies, the claim is trivial. Otherwise, the stack looks like:
    \[
      {\sframe{\frame {\fencetrans(\cmemass \expr \exprtwo \sep \cmdtwo)} \regmap {\km[\syscall]}:\st'}
        {\bm\buf(\lay\lcomp \rfs')}\boolms:\cfstack}
    \]
    We just show the case for the \ref{SI:Store} rule:
    \begin{multline*}
      \nsstep 2 {\sframe{\frame {\fencetrans(\cmemass \expr \exprtwo \sep \cmdtwo)} \regmap {\km[\syscall]}:\st'}
        {\bm\buf(\lay\lcomp \rfs')}\boolms:\cfstack}{} {\dstep:\dstep } {\onone:\omem \add}\\
      \sframe{\frame {\fencetrans(\cmdtwo)} {\regmap} {\km[\syscall]}:\st'} 
        {\bitem \add \val\overline{\bm\buf(\lay\lcomp \rfs')}}{\boolms\lor\bool'}:\cfstack 
    \end{multline*}
    where $\add = \toAdd{\sem \exprtwo_{\regmap,\lay}}$,
    $\val = \sem \exprtwo_{\regmap,\lay}$.
    Observe that the claim requires to verify just that
    \[
      \fwf{\sframe{\frame {\fencetrans(\cmdtwo)} {\regmap} {\km[\syscall]}:\st'} 
        {\bitem\add \val\overline {\bm\buf(\lay\lcomp \rfs')}}{\boolms\lor\bool'}}
    \]
    which is almost completely a consequence of the premise of
    the proof-rule for $\fwf\cdot$. In particular, as in the case above,
    we observe that we just need to observe that
    $\overline{\bm\buf(\lay\lcomp \rfs')}$ is such that
    $\fwf{\overline{\bm\buf(\lay\lcomp \rfs')}}$ and from the premise of the rule
    \ref{SI:Store}, we deduce that $\add \in \Ark$, that shows that the new
    buffer satisfies the desired conditions.
    \proofcase{$\ccall \expr {\exprtwo_1,\dots,\exprtwo_k} \sep \cmdtwo$}
    Observe that either \ref{SI:Call},
    \ref{SI:Call-Unsafe}, or \ref{SI:Call-Err} must apply. If one of the
    last two rules applies, the claim is trivial. Otherwise, The proof proceeds
    as in the previous cases, but the important observation is that:
    \begin{multline*}
      \nsstep 2 {\sframe{\frame {\fencetrans(\cmemass \expr \exprtwo \sep \cmdtwo)} \regmap {\km[\syscall]}:\st'}
        {\bm\buf(\lay\lcomp \rfs')}\boolms:\cfstack}{} {\dstep:\dstep } {\onone:\omem \add}\\
      \sframe{\frame {\fencetrans(\cmdtwo)} {\regmap} {\km[\syscall]}:\frame {\overline{\bm\buf(\lay\lcomp \rfs')}(\add)} {\regmap} {\km[\syscall]}:\st'} 
        {\overline{\bm\buf(\lay\lcomp \rfs')}}{\boolms\lor\bool'}:\cfstack 
    \end{multline*}
    where $\add = \toAdd{\sem \exprtwo_{\regmap,\lay}}$,
    $\val = \sem \exprtwo_{\regmap,\lay}$.
    And we must ensure that the loaded program
    (namely $\overline{\bm\buf(\lay\lcomp \rfs')}(\add)$) 
    is equal to $\fencetrans(\cmd')$ for some command.
    From the premise of the rule \ref{SI:Call} that has
    been applied, we deduce that $\add \in \underline \lay(\Fn[\km])$
    means that there is a function in $\Fn[\km]$
    such that $\lay(\fn)=\add$.
    From \Cref{rem:overlinewrtdom}, we deduce that
    $\overline{\bm\buf(\lay\lcomp \rfs')} = \lay \lcomp \rfs''$
    for some $\rfs''\sim_{\Fn} \rfs'\sim_{\Fn}\rfs$.
    By definition of $\cdot \lcomp \cdot$,
    this also means that $\overline{\lay\lcomp \rfs''}(\add)= \rfs(\add)$
    that satisfies the requirement by hypothesis on $\system$.
  \end{proofcases}
\end{proof}

\begin{lemma}
  \label{lemma:stepsemsim}
  For every pair of configurations
  \[
    \specconfone = \sframe{\frame {\cmd_1} {\regmap_1} {\km[\syscall]}:\st_1} {\bm{\buf_1}{\mem_1}}{{\boolms}_1}:
  \]
  and 
  \[
    \specconftwo = \sframe{\frame {\cmd_2} {\regmap_2} {\km[\syscall]}:\st_2} {\bm{\buf_2}{\mem_2}}{{\boolms}_2},
  \]
  and stack
  if 
  \begin{equation*}
    \sstep {\sframe{\frame {\cmd_1} {\regmap_1}{\km[\syscall]}:\st_1} {\bm{\buf_1}{\mem_1}}{{\boolms}_1}} {}\dstep \obs \\
    {\sframe{\frame {\cmd_2} {\regmap_2}{\km[\syscall]}:\st_2} {\bm{\buf_2}{\mem_2}}{{\boolms}_2}} 
  \end{equation*}
  for some observation $\obs$, then
  \[
    \step {\conf{\frame {\cmd_1} {\regmap_1}{\km[\syscall]}:\st_1}, \overline {\bm{\buf_1}{\mem_1}}} {\conf{\frame {\cmd_2} {\regmap_2}{\km[\syscall]}:\st_2}, \overline {\bm{\buf_2}{\mem_2}}}
  \]
\end{lemma}
\begin{proof}
  The proof goes by cases on the transition relation. Most of the cases are trivial. The most interesting ones are those which interact with memory.
  \begin{proofcases}
    \proofcase{\ref{AL:Fence}} We rewrite the assumption as follows:
    \begin{equation*}
      \sstep {\sframe{\frame {\cfence\sep\cmd} {\regmap_1} {\km[\syscall]}:\st_1} {\bm{\buf_1}{\mem_1}}{{\boolms}_1}}{} \dstep \obs \\
      {\sframe{\frame {\cmd} {\regmap_1} {\km[\syscall]}:\st_1} {\overline{\bm{\buf_1}{\mem_1}}}{{\boolms}_1}} 
    \end{equation*}
    The claim is 
    \begin{equation*}
      \step {\conf{\frame {\cfence\sep\cmd} {\regmap_1}{\km[\syscall]}:\st_1}, \overline {\bm{\buf_1}{\mem_1}}} {}\\
      {\conf{\frame {\cmd_1} {\regmap_1}{\km[\syscall]}:\st_1}, \overline {\bm{\buf_1}{\mem_1}}}
    \end{equation*}
    that is trivial.
    \proofcase{\ref{AL:Store}} We rewrite the assumption as follows:
    \begin{equation*}
      \sstep {\sframe{\frame {\cmemass \expr \exprtwo\sep\cmd} {\regmap_1}{\km[\syscall]}:\st_1} {\bm{\buf_1}{\mem_1}}{{\boolms}_1}} {} \dstep \obs\\
      {\sframe{\frame {\cmd} {\regmap_1}{\km[\syscall]}:\st_1} {{\bm{\bitem \add \val:\buf_1}{\mem_1}}}{{\boolms}_1}}.
    \end{equation*}
    Where $\add = \toAdd{\sem \expr_{\regmap_1,\lay}}$ and $\val = {\sem \exprtwo_{\regmap_1,\lay}}$.
    To show the claim is suffices to observe that 
    \begin{equation*}
      \step {\conf{\frame {\cmemass \expr \exprtwo\sep\cmd} {\regmap_1}{\km[\syscall]}:\st_1}, \overline {\bm{\buf_1}{\mem_1}}}{}\\
      {\conf{\frame {\cmd_1} {\regmap_1}{\km[\syscall]}:\st_1, \update{\overline {\bm{\buf_1}{\mem_1}}}\add\val}}
    \end{equation*}
    and that $\overline {\bm{\bitem \add \val:\buf_1}{\mem_1}} = \update{\overline {\bm{\buf_1}{\mem_1}}}\add\val$ that is a consequence of the definition of $\overline \cdot$.
    \proofcase{\ref{AL:Load}} We rewrite the assumption as follows:
    \begin{equation*}
      \sstep {\sframe{\frame {\cmemread \vx \expr\sep\cmd} {\regmap_1}{\km[\syscall]}:\st_1} {\bm{\buf_1}{\mem_1}}{{\boolms}_1}} {} \dstep \obs\\
      {\sframe{\frame {\cmd} {\update{\regmap_1}\vx {\val}}{\km[\syscall]}:\st_1} {{\bm{\buf_1}{\mem_1}}}{{\boolms}_1}}.
    \end{equation*}
    Where $\val = \bufread {\bm{\buf_1}{\mem_1}} {\toAdd{\sem \expr_{\regmap_1,\lay}}} 0$.
    To show the claim is suffices to apply \Cref{rem:bufreadoverline}.
  \end{proofcases}
\end{proof}


\end{document}